%% file: Community-Size_Paper_Draft-v22-clean.tex
\documentclass[onefignum,onetabnum]{siamart251216}   

\input{ex_shared-v03}
\newcommand\geomb{\overline{\text{Geom}}}
\newcommand\geom{\text{Geom}}
\newcommand{\given}{\,|\,}

\DeclareMathOperator{\Unif}{Unif}

\DeclareMathOperator{\Dir}{Dir}

\usepackage{bm}
\usepackage{booktabs}
\usepackage{microtype}

\ifpdf
\hypersetup{
  pdftitle={Community-Size Biases in Statistical Inference of Communities in Temporal Networks},
  pdfauthor={T. Y. Faust and A. A. Amini and M. A. Porter}
}
\fi

\begin{document}

\maketitle

\begin{abstract}
In the study of time-dependent (i.e., temporal) networks, researchers often examine the evolution of communities, which are sets of densely connected sets of nodes that are connected sparsely to other nodes.
  An increasingly prominent approach to studying community structure in temporal networks is statistical inference.
  In the present paper, we study the performance of a class of statistical-inference methods for community detection in temporal networks. We represent temporal networks as multilayer networks, with each layer encoding a time step, and we illustrate that 
    statistical-inference models that generate community assignments via either a uniform distribution on community assignments or discrete-time Markov processes are biased against generating communities with large or small numbers of nodes. 
  In particular, we demonstrate that statistical-inference methods that use such generative models tend to poorly identify community structure in networks with large or small communities. 
To rectify this issue, we introduce a novel statistical model 
that generates the community assignments of the nodes in given layer (i.e., at a given time) using all of the community assignments in the previous layer.
  We prove results that guarantee that our approach greatly mitigates the bias against large and small communities, so using our generative model is beneficial for studying community structure in networks with large or small communities.
   Our code is available at \url{https://github.com/tfaust0196/TemporalCommunityComparison}.
\end{abstract}





\section{Introduction}\label{intro}

{In network analysis, it is common to study ``communities" of nodes that are connected densely to each other but connected sparsely to other nodes~\cite{Porter09,Fortunato16,Newman18}. Investigations of community structure in networks have led to insights in the study of social networks \cite{Traud2012,Waskiewicz2012,Ozer16}, economic networks \cite{Brauksa13}, citation networks \cite{Ji16}, biological networks \cite{Chen13}, and many other applications. It is important in many situations --- including in the analysis of face-to-face contacts \cite{Fournet14}, transportation \cite{Morer20}, legislation cosponsorships \cite{lee2016, Neal20}, and other applications --- to consider relationships and/or interactions that change with time. 
One can represent such time-dependent data as a temporal network, with the entities and/or the ties between entities changing with time~\cite{Holme12,Holme15,Holme19}. A temporal network is a sequence of networks in which each network encodes the relationships between entities at one time point or time period.
A variety of approaches have been developed to algorithmically detect communities in temporal networks \cite{Rossetti18}. These approaches include statistical inference~\cite{Peixoto17,Sun10,Yang11}, optimization of various objective functions \cite{Mucha10, Bazzi16, Mazza23}, non-negative matrix and tensor factorization \cite{Dunlavy11,Gauvin14,Rossi13}, information-theoretic methods (such as those that minimize description length) \cite{Sun07,Rosvall10,Peixoto17}, local methods \cite{jeub2015,Interdonato17,Li20}, and others.

In the present paper, we study community detection in temporal networks using statistical inference~\cite{peixoto2023}. In an inferential approach, one uses a generative model to algorithmically detect a desired type of network structure. There are statistical-inference approaches to detect many types of mesoscale structures in networks \cite{Zhang15, Wang16, Polanco23}, and there are a particularly large number of such methods for {community structure \cite{peixoto2023,Fortunato16,gaffi2026}}. 
Investigations of community structure in temporal networks using statistical inference have led to insights into a wealth of applications, such as international migration \cite{danchev2018}, academic coauthorship networks \cite{Yang11}, face-to-face social contact networks in animals and high-school students \cite{Matias17}, and role detection in bicycle-sharing networks \cite{carlen2022}. The statistical-inference methods that we consider use a Bayesian approach to sample from a posterior distribution and obtain a network's community structure~\cite{peixoto2019}.

It is desirable to base the generative models that one employs to statistically infer community structure on realistic assumptions~\cite{peixoto2023}.
A generative model that relies on unrealistic assumptions can have a detrimental impact on the accurate detection of communities in networks, and both overfitting and underfitting occur in existing community-detection methods~\cite{ghasemian2020}. Unfortunately, it often is not straightforward to avoid unrealistic assumptions, and many generative models that appear to make reasonable choices include such assumptions. One example of a problematic assumption involves the probability distribution of the number of nodes in a community. 
It is common for generative models of temporal community structure to use a discrete-time Markov process to assign communities \cite{Yang11, Ghasemian16, Matias17, Bazzi20}. However, due to the choices in many such models, community-size distributions become increasingly localized over time. To demonstrate this behavior, we consider the Markov-process models of Yang et al.~\cite{Yang11} and Bazzi et al.~\cite{Bazzi20}. We show that the community-size distributions that arise from these models become increasingly localized over time.
Therefore, at later times of a temporal network, the probability of generating a small or large community is much smaller than the probability of generating a community of moderate size. Real-world temporal networks can have communities of disparate sizes \cite{jeub2015,Palla07, Matias17,ghasemian2020}, including ones that are small or large, so it is desirable that generative models of networks are able to successfully infer small and large communities. 

To mitigate the issue of generative models producing networks that are biased against small and large communities, we introduce a novel community-evolution approach that yields community-size distributions with substantially less localization than 
Markov-process models. Our approach generates the community assignments of a network at a given time using the community assignments of all nodes at the previous time, rather than updating each node separately at each time. 
For a multilayer representation of a temporal network~\cite{kivela2014}, when evolving the community assignments from one layer to the next (i.e., from one time to the next), we generate the number of nodes with community assignments that change between layers, instead of generating the community assignments of each node separately. 
This choice arises from the idea of ``exchangeability" \cite{Bernardo96}. In our context, it signifies that one should not distinguish between nodes with the same community assignment. We demonstrate that statistical-inference methods that use our generative model perform more accurately than the Yang et al.~\cite{Yang11} and Bazzi et al.~\cite{Bazzi20} Markov-process models in networks with small or large communities.}

Our paper proceeds as follows. In Section \ref{Notation}, we introduce our main notation. In Section \ref{sec:StatInf}, we introduce the Yang et al.~\cite{Yang11} and Bazzi et al.~\cite{Bazzi20} Markov-process methods and our statistical-inference approach in more detail. 
{In Section \ref{CommSizeDistLoc}, we demonstrate empirically that our approach leads to significantly less-localized community-size distributions than the Markov-process methods. We also prove a result about the behavior of the single-layer community-size distributions in our approach in the limit of infinitely many layers.} In Section \ref{ComparisonSI}, we demonstrate that our approach identifies communities more accurately in synthetic networks with small or large communities. In Section \ref{FinalConclusions}, we conclude and discuss future directions. In Appendix \ref{Expressions}, we derive necessary results that are required for the statistical-inference approaches we consider. Our code is available at \url{https://github.com/tfaust0196/TemporalCommunityComparison}.


\section{Notation}\label{Notation}

In this section, we introduce notation for both single-layer (i.e., ``monolayer") networks and temporal networks, which we represent as multilayer networks in which each layer corresponds to one time step (which can represent a time point or period) \cite{kivela2014,aleta25}. To simplify notation, we {denote the set $\{1,\ldots,N\}$ by $[N]$.}

We first discuss our notation for monolayer networks. For simplicity, we consider unweighted and undirected networks without self-edges or multi-edges.
A monolayer network is a graph $G = (V,E)$, which consists of a set {$V = [N]$} of nodes (i.e., vertices) and a set $E \subseteq V \times V$ of edges.
We denote an undirected edge by $(i,j)$. We represent a monolayer network $G$ using an adjacency matrix $A \in \{0,1\}^{n \times n}$, where $A_{ij} = 1$ if nodes $i$ and $j$ are connected directly by an edge (i.e., they are adjacent) and $A_{ij} = 0$ otherwise.

We represent a temporal network as a multilayer network in which each layer encodes the adjacencies between nodes at its associated time step. We model a temporal network as a sequence of network layers (i.e., times) $\ell \in {[L]}$.
At each time $\ell \in {[L]}$, we consider all nodes $i \in {[n]}$. 
We refer to an instantiation of a node in a given layer as a node-layer $(i,\ell) \in {[n] \times [L]}$. We again use an adjacency representation, so we have a sequence $(A^{(1)},\ldots,A^{(L)})$, with $A^{(\ell)} \in \{0,1\}^{n \times n}$ for each layer $\ell$. 
We assume that the networks are unweighted and undirected, so $A^{(\ell)}_{ij} = 1$ if node-layers $(i,\ell)$ and $(j,\ell)$ are adjacent and $A^{(\ell)}_{ij} = 0$ if they are not adjacent. For notational convenience, we let $A$ denote the sequence $(A^{(1)},\ldots,A^{(L)})$. Technically, this is an abuse of notation, because we already used $A$ refer to a single adjacency matrix for a monolayer network, but we always clearly state whether we are considering a monolayer network or a temporal network. As general terminology, we refer to $A$ as an ``adjacency structure". For convenience (and despite the additional associated abuse of notation), we also sometimes refer to $A$ as a ``network".

We also need some other foundational notation and terminology. Let
\begin{equation}\label{DeltaDefinition}
	\Delta^{k - 1} := \{ \pi = (\pi_i) \in \mathbb R^k : \pi_i \ge 0 \text{ for all } i \text{ and } \sum_{i = 1}^k \pi_i = 1 \}
\end{equation}
 denote the ``probability simplex" in $\mathbb R^k$.
Additionally, a ``weak composition" of a non-negative integer $n$ into $k$ parts is an ordered $k$-tuple $(n_1,\dots,n_k)\in\mathbb{Z}_{\ge 0}^k$ such that $\sum_{r = 1}^k n_r = n$. 
Finally, we recall the indicator function $\mathbb{1}$, where $\mathbb{1}\{\cdot\} = 1$  if all of its arguments hold and $\mathbb{1}\{\cdot\} = 0$ otherwise.


\section{Statistical-Inference Methods for Community Detection} \label{sec:StatInf}

In this section, we discuss statistical-inference methods for community detection. In Section \ref{StatInf}, we discuss how the methods that we consider use a generative model to infer community structure in networks. In Sections \ref{CommAssnSL} and \ref{CommAssnML}, we introduce the generative models that we will consider, with a particular focus on each method's probability distribution for community assignments.


\subsection{Statistical Inference} \label{StatInf}

In a statistical-inference method for community detection, one chooses a generative model for networks with community structure \cite{peixoto2023}. To do this, one creates a generates random adjacency structure $A$, which is an adjacency matrix for a monolayer network and is a sequence of adjacency matrices for a temporal network, according to some probability distribution $\P(A)$. We choose the probability distribution $\P(A)$ so that the randomly-generated networks have community structure. Let $k$ denote the number of communities in a network. Given $k$, a generative model first generates a vector $g \in [k]^n$ (if we seek a monolayer network) or a matrix $g \in [k]^{n \times L}$ (if we seek a temporal network) of community assignments. Let $g_i \in [k]$ denote the community assignment of node $i$ of a monolayer network, and let $g_{(i,\ell)} \in [k]$ denote the community assignment of node-layer $(i,\ell)$ of a temporal network. The vector 
\begin{equation} \label{gsubell}
	g_{(\ell)} = (g_{(1,\ell)},\ldots,g_{(n,\ell)})
\end{equation} 
encodes the community assignment of each node-layer in layer $\ell$ of a temporal network. 

Using the generative model, we generate an adjacency structure $A$ with an ``expected'' community structure $g$. We expect nodes (respectively, node-layers) with the same community assignment to be more likely to be adjacent to each other than to nodes (respectively, node-layers) with different community assignments. One commonly-employed generative model to generate a monolayer network $A$ given a community assignment $g$ is a \emph{stochastic block model} (SBM)~\cite{peixoto2023,peixoto2019,abbe2018}. In the simplest type of SBM, given a matrix $\psi \in [0,1]^{k \times k}$, we place an edge between nodes $i$ and $j$ independently with probability $\psi_{g_ig_j}$ for each pair $\{i,j\}$ of distinct nodes. 
If the diagonal elements $\psi_{rr}$ are larger than the off-diagonal elements $\psi_{rs}$ (with $r \ne s$), we expect that nodes with the same community assignments are connected more densely than nodes with different community assignments. In other words, we expect that $A$ has community structure that is specified by $g$. 

To make our above intuition precise, let $\P(A|g)$ denote the probability distribution that we obtain a network (i.e., an adjacency structure) $A$ given the community structure $g$.
Let $\P(g)$ denote the probability distribution of the community assignments. This distribution describes the probability of a given community structure $g$ independent of the observed adjacency structure $A$. The primary foci of the present paper are the assumptions about this distribution and the effect of these assumptions on the performance of statistical-inference methods for community detection. In particular, we focus on assumptions about the probability distribution of the number of nodes in a community (i.e., the size of a community). Such size assumptions also affect the behavior and performance of statistical-inference methods that infer a discrete group assignment for each node-layer. This includes the identification of both community structure and other mesoscale structures (e.g., core--periphery structure) in temporal networks. Faust and Porter \cite{Faust25} discussed the effects of similar size assumptions on the behavior of statistical-inference algorithms for core--periphery detection.

The generative model of a random network $A$ with community structure $g$ is
\begin{equation}\label{Bayesianform}
	\P(A,g) = \P(A|g)\P(g)\,.
\end{equation}
Because we fix the number $k$ of communities, one can think of $k$ as an input of the examined methods. We assume a fixed $k$ throughout the entire statistical-inference process, so we omit $k$ from the notation for each probability distribution. 

We consider statistical-inference methods that use a Bayesian approach to infer the community structure $g$ of a network $A$ (which is either a monolayer network or a temporal network). Given $\P(A|g)$ and $\P(g)$, the posterior distribution for the inferred community structure $g$ given the observed network $A$ is
\begin{equation}\label{BayesianggivenA}
	\P(g|A) = \frac{\P(A|g) \P(g)}{\P(A)} \, .
\end{equation}

There are a variety of approaches to sample from $\P(g|A)$ \cite{Yang11, Funke19}. 
In Section \ref{ComparisonSI}, we briefly discuss the specific sampling approaches that we use in our comparisons of various statistical-inference methods. 
However, our focus is the effect of the choice of generative model on the results of statistical inference, so we do not include a detailed discussion of the benefits and drawbacks of different sampling approaches.
For comparisons of such sampling approaches, see \cite{Park22} and \cite{Lye19} for a general comparison and see \cite{Funke19} for a discussion of the application of various methods of statistical inference of community structure.

There are also a variety of generative models $\P(A|g)$ that generate a network $A$ with expected community structure $g$. (See \cite{Funke19} for a review of such models.) However, because our primary concerns are (1) the de facto assumptions about the community assignments that arise from different choices of community-assignment probability distributions $\P(g)$ and (2) the effects of these assumptions on statistical-inference results, we also limit our discussion to the generative models $P(A|g)$ that we use in our comparisons. 


\subsection{Community-Assignment Probability Distributions for Monolayer Networks} \label{CommAssnSL}

As we discussed in Section \ref{StatInf}, our primary goals are to examine (1) the choices of community-assignment probability distributions $\P(g)$ in various generative models of community structure in temporal networks and (2) the influence of these choices on the results of their corresponding statistical-inference approaches for community detection. We now discuss each of the choices of $\P(g)$ in detail. In Sections \ref{unifDistSL} and \ref{nodewiseSL}, we 
consider generative models of monolayer networks, as the examined generative models of temporal networks are extensions of these models. In our discussion, we focus on properties of $\P(g)$ that affect the performance of the corresponding statistical-inference methods. In Section \ref{CommSizeDistLoc}, we support these claims using numerical computations. In Appendix \ref{sec:monolayer-theory}, we give additional theoretical results about community-size localization in monolayer networks. In particular, in Theorem \ref{theorem-fin}, we prove a deep connection between the hierarchical Bayesian model in Section \ref{nodewiseSL} and the limiting behavior of the community-size distribution of such a model.


\subsubsection{{Uniform Distribution on Community Assignments}}\label{unifDistSL}

Our first class of community-assignment probability distributions $\P(g)$ is uniform distributions on community assignments. This type of distribution entails that
\begin{equation}\label{unifSL}
	\P(g) = \frac{1}{k^n}
\end{equation}
 for all community assignments $g \in [k]^n$. Equivalently, such a model generates a community assignment $g$ by generating each node's community assignment $g_i$ independently at random from a uniform distribution on $[k]$.
 
It is common to use a uniform distribution on community assignments (see, e.g., \cite{Hastings06, Ghasemian16}). However, Polanco and Newman~\cite{Polanco23} noted that this distribution choice assumes implicitly that the distributions of the sizes of communities are highly localized (and specifically that the probability of generating a community structure with large or small communities is very small), which typically is an unrealistic assumption in practice.


\subsubsection{{Independent Non-Uniform} Community Assignments} \label{nodewiseSL}

One can also obtain a single-layer community-assignment probability distribution using a hierarchical Bayesian approach. To do this, we
{independently sample non-uniform} community assignments \cite{Yang11}.
{In such an approach,
one independently samples the community assignments of each node from a distribution that is not uniform.}
Specifically, given a probability vector $\pi \in \Delta^{k - 1}$, one samples the community assignment $g_i$ of each node $i$ independently from the distribution $\pi$. We denote this procedure by
\begin{equation}
	g_i \given \pi \sim \pi \,.
\end{equation}
Equivalently, the probability distribution of $g$ is
\begin{equation}\label{nodewisecond}
	\P(g|\pi) = \prod_{r = 1}^k \pi_r^{{n_r}} \, ,
\end{equation}
where $n_r$ is the number of nodes $i$ with $g_i = r$. In other words, $n_r$ is the size of community $r$. 

To remove the dependence on $\pi$ in \eqref{nodewisecond}, we sample $\pi$ from the Dirichlet distribution with {parameter $\gamma = (\gamma_1, \dots, \gamma_k)$.} That is,
\begin{equation}\label{dirichlet}
	\pi \sim \text{Dir}(\gamma) \,.
\end{equation}
In the present paper,
we always choose {$\gamma = (1,\ldots,1)$}, which entails that we sample $\pi$ using the uniform distribution on the probability simplex $\Delta^{k - 1}$.
For this 
choice of $\gamma$, the resulting distribution on $g$ has the following equivalent interpretation: 
\begin{enumerate}
 \item[(1)] {Sample the sizes $n_1,\ldots,n_k$ of communities $1,\ldots,k$ uniformly at random from the set $\mathcal{C}_n^k$ of ordered pairs of $k$ non-negative integers that sum to $n$.}
 \item[(2)] {Sample the community assignment $g$ uniformly at random from the set of all community assignments with $n_r$ nodes in community $r$ for all $r \in [k]$.} 
\end{enumerate}
We summarize this equivalence in the Proposition~\ref{prop:uniform-compositions}, which we prove in Appendix~\ref{app:remaining:proofs}.

\begin{proposition}\label{prop:uniform-compositions}
	Let $\pi \sim \Dir(1,\ldots,1)$ and $g_1,\ldots,g_n\sim\pi$.
	Additionally, let $n_r = \sum_{i = 1}^n\mathbb{1}\{g_i = r\}$ be the size of community $r$ and let ${\bm n} = (n_1,\ldots,n_k)$.
	We then have that $\bm n$ is {uniform} on the set $\mathcal{C}_n^k$ of weak compositions of $n$ into $k$ parts. Namely, 
	\begin{equation*}
		\mathbb{P}(\bm n = c) = \frac{1}{\binom{n + k - 1}{k - 1}} \quad \text{for all }\, c \in\mathcal{C}_n^k \,.
	\end{equation*}
	Moreover, conditional on the community sizes $c_1,\ldots,c_r$, the community labels $g_1,\ldots, g_n$ are uniform over assignments with those counts.
	Specifically,
	\begin{equation*}
		\mathbb{P}(g \given \bm n = c) = \frac{c_1! \times \cdots \times c_k!}{n!} 
	\end{equation*}
	 for any $g$ such that
	 $c_r = \sum_{i = 1}^n\mathbb{1}\{g_i = r\}$
	and $\mathbb{P}(g \given \bm n = c) = 0$ otherwise.
	\end{proposition}
	
This latter
form of the procedure is
a popular approach for setting 
a prior on community assignments in monolayer-network models \cite{Riolo17}.

Although this class of community-assignment probability distributions is perhaps less intuitive than using uniform distributions on community assignments, we show in Section~\ref{CSSL}
that it greatly mitigates localization issues in community-size distributions. In Theorem \ref{theorem-fin}, we prove a deep connection between the present subsubsection's hierarchical Bayesian model and the limiting behavior of the community-size distribution of such a model.


\subsection{Community-Assignment Probability Distributions for Temporal Networks} \label{CommAssnML}

We now build on our discussion of choices of community-assignment probability distributions $\P(g)$ in generative models of monolayer networks (see Section \ref{CommAssnSL}) to examine these choices in temporal networks.

In the present discussion of $\P(g)$, we omit closed-form expressions for some choices of $\P(g)$ because deriving them is cumbersome.
For those probability distributions, we instead provide a procedure to sample from $\P(g)$. In Appendix \ref{Expressions}, we include derivations of the closed-form expressions of $\P(g)$ that are required to sample from the posterior distributions $\P(g|A)$. In Section \ref{Methodology}, we discuss these sampling methods in detail. 


\subsubsection{Uniform Distribution on Community Assignments}	
\label{unifDistML}

As with monolayer networks (see Section \ref{unifDistSL}), one can employ a uniform distribution on community assignments in the study of temporal networks. In contrast to the monolayer setting, because a uniform distribution avoids correlation between the community assignments in different layers (which is generally desirable), Researchers avoid this choice for inference applications.

Analogously to \eqref{unifSL}, we have
\begin{equation}\label{unifprobdist}
	\P(g) = \frac{1}{k^{nL}}
\end{equation}
for all community assignments $g \in [k]^{n \times L}$. As with monolayer networks, it is equivalent to generate a community assignment $g$ by generating the community assignment $g_{(i,\ell)}$ of each node-layer independently at random from a uniform distribution on $[k]$. 

As with monolayer networks, this community-assignment probability distribution implicitly makes the unrealistic assumption that the community-size distributions are highly localized.


\subsubsection{Discrete-Time Markov-Process Models} \label{NodewiseEvolML}

Another class of community-assignment probability distributions use a discrete-time Markov process to generate community assignments \cite{Yang11, Ghasemian16, Matias17, Bazzi20}. 

One uses a nodewise approach (see Section \ref{nodewiseSL}) to generate the community assignments for the first layer of a temporal network. In particular, to sample the community assignment $g_{(i,1)}$, one follows the procedure
\begin{alignat}{3}
    \pi & &&\;\sim \text{Dir}(\gamma) \, , \notag \\
    g_{(i,1)} &\given \pi &&\;\sim\; \pi \, , \label{DirichletNodewise}
 \end{alignat}
where we recall that $\gamma = (1,\ldots,1)$. In the present paper, when we provide a sequence of samples from probability distributions (i.e., a ``sampling procedure''), such as in \eqref{DirichletNodewise}, we perform the sampling in sequence from top to bottom. 

For the second and subsequent layers, one generates the community assignments of each node-layer using a process that depends only on the community assignment of the same node in the previous layer. To do this, one inputs a ``transition kernel'' $K = (K^{(2)}, \ldots, K^{(L)})$, where each $K^{(\ell)}$ is a matrix whose rows sum to $1$, and ``laziness parameters"
 $\alpha_\ell \in [0,1]$.  The $\ell$th laziness parameter $\alpha_\ell$ is the probability that one copies the community assignment from that of the previous layer when generating the community assignment for a node in layer $\ell$. For notational convenience, we write $\alpha = (\alpha_2,\ldots,\alpha_L)$.

One then generates $g_{(i,\ell)}$ according to the conditional distribution
\begin{equation}\label{condsinglenode}
	\P( g_{(i,\ell)} = r \given g_{(i,\ell-1)} =s, \alpha, K) = \alpha_\ell \, \mathbb{1}\{r = s\} + (1-\alpha_\ell) K^{(\ell)}_{sr} \, ,
\end{equation}
where the indicator function $\mathbb{1}\{\cdot\} = 1$ 
if all of its arguments hold
and $\mathbb{1}\{\cdot\} = 0$ 
otherwise.
Sampling the community assignments $g_{i,\ell}$ according to the conditional distributions (\ref{condsinglenode}) is equivalent to using
\begin{equation*}
	\{g_{(i,\ell)}\}_{\ell=2}^L  \given \alpha,
	K \;\sim\; \text{Markov}\left(\left\{\alpha_{\ell} I + (1 - \alpha_\ell) K^{(\ell)}\right\} \right) \, ,
\end{equation*}	
where $\text{Markov}(\{K^{(\ell)}\})$ denotes a discrete-time Markov process with transition kernels $\{K^{(\ell)}\}$.
We thus can write
\begin{alignat}{3}
    \pi & &&\;\sim\; \text{Dir}(\gamma) \, , \notag \\
    g_{(i,1)} &\given \pi &&\;\sim\; \pi \,, \notag\\
    \{g_{(i,\ell)}\}_{\ell=2}^L & \given \alpha, K &&\;\sim\; \text{Markov}\left( \left\{\alpha_{\ell} I + (1-\alpha_\ell) K^{(\ell)}\right\} \right) \, .  \label{MarkovProcessModelNoPrior}
 \end{alignat}

To complete the specification of the model, we need to provide priors on the laziness parameters $\alpha$ and the transition kernel $K$. As we discussed in Section \ref{intro}, many common choices of priors on $\alpha$ and $K$ lead to community-size distributions that become increasingly localized over time. In the present paper, we consider two representative choices of priors for these parameters. 
 
The first choice of priors is from the method of Yang et al.~\cite{Yang11}. In this method, the laziness parameters $\alpha$ are $(0,0,\ldots,0)$. To choose the transition kernel $K$, one first assumes that $K^{(\ell)} = \tilde{K}$ for each layer $\ell \in \{2,\ldots,L\}$. One then imposes an independent Dirichlet prior with parameters $(\mu_{s_1},\ldots,\mu_{s_k})$ on each row of $\tilde{K}$. That is,
\begin{equation} \label{YangPrior}
	\tilde{K}_{s*} \sim \text{Dir}(\mu_{s*})\,,
\end{equation}
where $\tilde{K}_{s*} = (\tilde{K}_{s1},\ldots,\tilde{K}_{sk})$ and $\mu_{s*} = (\mu_{s1},\ldots,\mu_{sk})$.
In the present paper, we assume that $\mu_{s*} = (1,\ldots,1)$ for each $s \in [k]$.
 
In our second choice of priors, we use a method from Bazzi et al.~\cite{Bazzi20} along with our own choices of prior distributions of certain parameters.
In the transition kernel $K$, each probability vector $K^{(\ell)}_{s*}$ is the same for all $s \in [k]$. That is, $K^{(\ell)}_{s*} = \kappa^{(\ell)}$. 
 In the present paper, we sample each entry $\alpha_\ell$ of $\alpha$ uniformly at random from $[0,1]$. 
 We then impose an independent Dirichlet prior with parameters 
 $\mu^{(\ell)} = (\mu^{(\ell)}_1, \ldots, \mu^{(\ell)}_k)$ on 
$\kappa^{(\ell)}$ for each layer $\ell \in \{2,\ldots,L\}$. 
In summary, we set
\begin{align*}
        \alpha_\ell &\sim \text{Unif}(0,1) \, , \\
        \kappa^{(\ell)} &\sim \text{Dir}(\mu^{(\ell)}) \, , \notag \\
        K^{(\ell)}_{s*} &= \kappa^{(\ell)} \, ,
\end{align*}
where $K^{(\ell)}_{s*} =  (K^{(\ell)}_{s1},\ldots,K^{(\ell)}_{sk})$.
We assume that $\mu^{(\ell)} = (1,\ldots, 1)$ for each layer $\ell \in \{2,\ldots,L\}$. For brevity, we refer to this choice of prior distributions as the ``Bazzi et al. prior'' despite the fact that
we choose some of the prior distributions ourselves.

The Yang et al.\ \cite{Yang11} and Bazzi et al.\ \cite{Bazzi20} Markov-process approaches both have community-size distributions with less localization than those for a uniform distribution on community assignments (see Section \ref{unifDistML}). However, in both approaches, the single-layer community-size distributions are more localized for later layers than for earlier layers, which leads to moderate localization in the community-size distributions of a temporal network (see Section \ref{CSML}). Because many Markov-process approaches (e.g., \cite{Ishiguro10, Matias17}) have similar choices of prior distributions to the Yang et al.\ \cite{Yang11} and Bazzi et al.\ \cite{Bazzi20} approaches, we also expect increasing localization over time in other Markov-process models.


\subsubsection{Layerwise-Exchangeable Count-Splitting (LECS) Prior}
\label{novelML}
As we discussed in Section~\ref{NodewiseEvolML}, discrete-time Markov updates tend to localize single-layer community sizes over time. 
To mitigate this issue, we introduce a novel layerwise-exchangeable count-splitting (LECS) prior using a temporal extension of exchangeability~\cite{Bernardo96}.
The core idea is that the community assignments in a 
layer should not distinguish between nodes with indistinguishable community assignments in
the previous layer.

More specifically, the LECS prior specifies a layerwise update that (1) is \emph{node-exchangeable} (i.e., the nodes are indistinguishable) within each community in the previous layer and (2) \emph{microcanonical}\footnote{We borrow the term ``microcanonical" from statistical physics \cite{sethna2021}.} (i.e., uniform over assignments conditional on the given community sizes) at the labeling step. 
We classify each of the $n_r^{(\ell - 1)}$ nodes in each community $r$ of layer $\ell - 1$ into
``remainers" (which remain in $r$ in layer $\ell$) and ``movers" (which we reassign to some community $s \neq r$). We draw the number of remainers from a truncated geometric distribution in a process that we call ``geometric retention"). We then allocate the community assignments of the movers 
by first generating the community sizes by uniformly sampling a weak composition with $k - 1$ parts and then, conditional on the community sizes, assigning the community labels 
uniformly at random. Together, geometric retention and uniformly sampling a weak composition for each community mitigate
the increase in localization with time (see Section~\ref{ArashNumerics}).

\paragraph{Preliminaries}

Recall that a \emph{weak composition} of a non-negative integer $n$ into $k$ parts is an ordered $k$-tuple $(n_1,\dots,n_k)\in\mathbb{Z}_{\ge0}^k$ such that $\sum_{r=1}^k n_r = n$. Let $\mathcal{C}_n^k$ be the set of weak compositions. For an index set $\mathcal{I}\subseteq [n]$ of nodes and a weak composition $c = (c_1,\dots,c_k)\in\mathcal{C}_{|\mathcal{I}|}^k$, we define the set
\begin{equation*}
	\mathsf{Assign}(\mathcal{I},c) = \left\{g\in [k]^{\mathcal{I}}: |\{i\in\mathcal{I}: g_i = r\}| = c_r\ \text{for all} \ r \in [k]\right\} 
\end{equation*}
and denote the uniform distribution on this set of community assigments by $\text{Unif}(\mathsf{Assign}(\mathcal{I},c))$.
Suppose that $\{\mathcal{I}_r\}_{r = 1}^k$ is a partition of $[n]$ and that $g_r\in[k]^{\mathcal{I}_r}$ for each $r$. Let $\bigoplus_{r = 1}^k g_r \in [k]^n$ denote the blockwise concatenation of community assignments. That is, 
$\bigoplus_{r = 1}^k g_r \in [k]^n$ is a map that satisfies the property $(\bigoplus_{r = 1}^k g_r)\!\mid_{\mathcal{I}_r} = g_r$.

\paragraph{First layer}
As we mentioned previously, we assume exchangeability, so the nodes in the first layer are indistinguishable from each other. 
To ensure this, we use a nodewise approach for the first layer (see (\ref{DirichletNodewise})) and thus have
\begin{alignat}{3}
    \pi & &&\;\sim\; \text{Dir}(\gamma) \, , \notag \\
    g_{(i,1)} &| \pi &&\;\sim\; \pi \,, \label{novelFirstLayer}
 \end{alignat}
where $\gamma = (1,\ldots,1)$.

\paragraph{Layerwise evolution}
Given the community assignments $g_{(\ell - 1)}$ for $\ell\ge2$, let 
\begin{equation*}
	\mathcal{G}_r^{(\ell - 1)} = \{i \in [n] :  g_{(i,\ell - 1)} = r \}
\end{equation*}
be the set of nodes in community $r$ in layer $\ell - 1$, which has size $n_r^{(\ell - 1)} = |\mathcal{G}_r^{(\ell - 1)}|$.
For each $r \in [k]$, we generate the layer-$\ell$ assignments for nodes in $\mathcal{G}_r^{(\ell - 1)}$ via a two-stage process.

In the first stage, we generate a ``transition-count vector" $\mathbf{c}_r^{(\ell)} = (c_{r1}^{(\ell)},\dots,c_{rk}^{(\ell)})\in\mathcal{C}_{n_r^{(\ell - 1)}}^k$, where $c_{rs}^{(\ell)}$ is the number of nodes from community $r$ in layer $\ell - 1$ that transition to community $s$ in layer $\ell$. 
As we discussed previously, we do this by sampling the number of remainers and then uniformly splitting the set of movers:
\begin{align}
\begin{split}\label{eq:lecs-count-split}
	p_{r,\ell} &\sim \text{Unif}(0,1) \,, \\
	c_{rr}^{(\ell)}\mid p_{r,\ell} &\sim \geomb\!\left(n_r^{(\ell - 1)},\,p_{r,\ell}\right) \,,\\
	\big(c_{rs}^{(\ell)}\big)_{s\neq r}\ \Big|\ c_{rr}^{(\ell)}\ &\sim\ \text{Unif}\!\left(\mathcal{C}^{k-1}_{\,n_r^{(\ell - 1)} - c_{rr}^{(\ell)}}\right) \,,
\end{split}
\end{align}
where $\geomb(n,p)$ is the truncated geometric distribution on $\{0,\dots,n\}$ with probability mass function (PMF) $\P(m) \propto p^{n - m}$. This distribution biases the number $m = c_{rr}^{(\ell)}$ of remainers toward the total size 
$n = n_r^{(\ell - 1)}$ for large retention probabilities $p$.

In the second stage, we do {microcanonical labeling}. Conditional on the layer-$\ell$ community sizes $\mathbf{c}_r^{(\ell)}$, we uniformly sample the new assignments for the block of nodes from community $r$. That is,
\begin{equation*}
	\mathbf{g}'_{r,\ell} \sim  \Unif\left(\mathsf{Assign}(\mathcal{G}_r^{(\ell - 1)},\,\mathbf{c}_r^{(\ell)})\right) \,.
\end{equation*}
We then obtain the final community-assignment vector for layer $\ell$  by concatenating these assignments for all $r$. That is,
\begin{equation*}
	g_{(\ell)} = \bigoplus_{r = 1}^k \mathbf{g}'_{r,\ell} \,.
\end{equation*}


\section{Localization of Community-Size Distributions}\label{CommSizeDistLoc}

We numerically compute community-size distributions for each of the generative models in Sections \ref{CommAssnSL} and \ref{CommAssnML}, and we prove theoretical results about the behavior of the single-layer community-size distributions in our approach (see Section \ref{novelML}) in the limit of infinitely many layers. We obtain the following results.
\begin{itemize}
\item Monolayer-network generative models: 
\begin{itemize}
\item The community-size distribution of a uniform distribution on community assignments (see Section \ref{unifDistSL}) is highly localized. 
\item The community-size distribution of nodewise community assignments (see Section \ref{nodewiseSL}) is much less localized than the community-size distribution of a uniform distribution on community assignments.
\end{itemize}
\item Temporal-network generative models:
\begin{itemize}
\item The community-size distribution of a uniform distribution on community assignments (see Section \ref{unifDistML}) is highly localized. 
\item The community-size distributions of the Yang et al.\ \cite{Yang11} and Bazzi et al.\ \cite{Bazzi20} discrete-time Markov-process approaches (see Section \ref{NodewiseEvolML}) are much less localized than that of a uniform distribution on community assignments.
However, for both the Yang et al.\ and Bazzi et al.\ approaches,, the single-layer community-size distributions of later layers are more localized than those of earlier layers. 
\item In our approach (see Section \ref{novelML}), the localization of the single-layer community-size distributions increases much more slowly than it does for the Yang et al.\  \cite{Yang11} and Bazzi et al.\ \cite{Bazzi20} approaches. Consequently, the overall community-size distribution is much less localized than those of the Yang et al.\ and Bazzi et al.\ approaches.
\end{itemize}
\end{itemize}

In Section \ref{Methodology}, we discuss the methodology that we use to verify these claims. In Section \ref{CSSL}, we verify {empirically} verify the above claims for monolayer-network generative models. In Section \ref{CSML}, we {empirically} verify the above claims for temporal-network generative models. Finally, in Section \ref{ArashNumerics}, we {prove theoretical results about} the behavior of the single-layer 
community-size distributions {for our 
LECS prior} (see Section~\ref{novelML}) in the limit of infinitely many layers.


\subsection{Methodology}\label{Methodology}

To compute the community-size distributions for each of the generative models in Sections \ref{CommAssnSL} and \ref{CommAssnML}, we begin by using them to generate $M =10^6$ instantiations of community assignments $g$.

For monolayer networks, we report the empirical distribution (i.e., histogram) of the size of community $1$. We let $g^{(m)}$ be the $m$th community-assignment instantiation, and we plot the observed frequencies $P_i$ for each $i \in \{0,\ldots,n\}$, where $P_i$ is defined by the map
\begin{equation}\label{Psubi1}
    i \mapsto P_i := \frac{1}{M}  \sum_{m=1}^M \mathbb{1}\{n_1(g^{(m)}) = i\} \, ,
\end{equation}
where $n_1(g^{(m)})$ is the size of community $1$ in community assignment $g^{(m)}$.
By symmetry, we choose the community label ``1" without loss of generality.
In Figure \ref{samplesinglelayer}, we show an example of such a histogram.

\begin{figure}[h]
\centering
\includegraphics[width=0.75\textwidth]{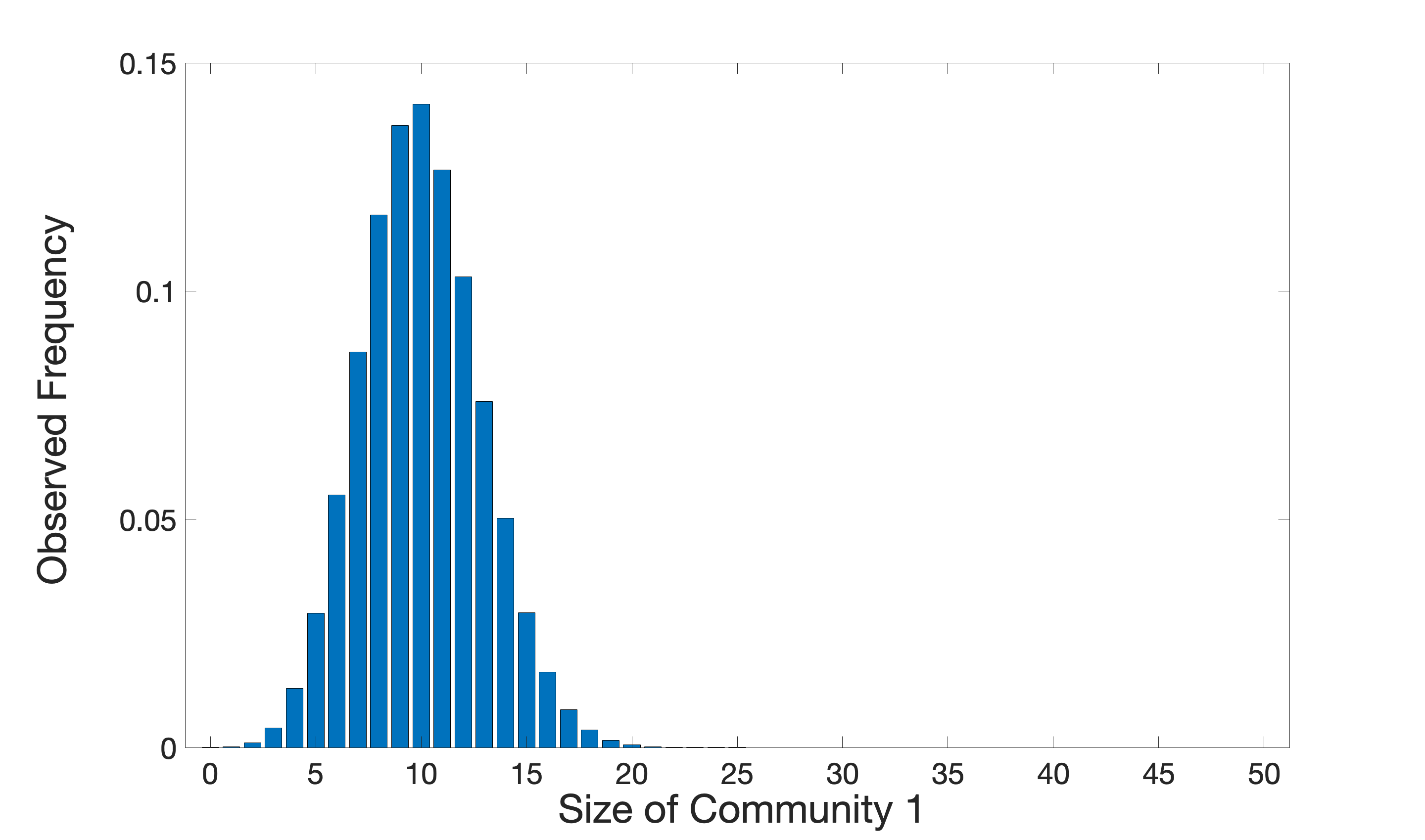}
\caption[An example of the community-size frequency histogram for a 50-node monolayer network.]{An example of the community-size frequency histogram for a 50-node monolayer network. The horizontal axis indicates the number of nodes in the network with community assignment $1$, and the vertical axis indicates the observed frequency of community assignments with that number of nodes in community $1$.
}
\label{samplesinglelayer}
\end{figure}

For a temporal network, we first count the number of node-layers $(i,\ell)$ in layer $\ell$ that are in community $1$, for each layer $\ell \in [L]$. 
We also count the total number of node-layers that are in community $1$ across a whole network. For each layer $\ell$, we then generate a histogram for each layer of the observed frequencies $P_i$ for each $i \in \{0,\ldots,n\}$, where $P_i$ is defined by the map
\begin{equation}\label{Psubi2}
    i \mapsto P_i := \frac{1}{M} \sum_{m = 1}^M \mathbb{1}\{n_1\left(g_{(\ell)}^{(m)}\right) = i\}
\end{equation}
 in that layer and $n_1\left(g_{(\ell)}^{(m)}\right)$ \footnote{This notation differs slightly from the notation $n_1^{(\ell)}$ that we used previously for the size of community $1$ in layer $\ell$ in a community assignment $g$ (see Section \ref{novelML}) . We use the alternate notation $n_1\left(g_{(\ell)}^{(m)}\right)$ in the map \eqref{Psubi2} to highlight the fact that we consider
 the size of community $1$ in layer $\ell$ specifically for the $m$th community-assignment instantiation $g^{(m)}$. 
 Subsequently, we consider only the community sizes of a single community assignment $g$, so we henceforth
 return to using the notation $n_r^{(\ell)}$.} is the size of community $1$ in layer $\ell$ of the community assignment $g^{(m)}$. We also generate a histogram of the observed frequencies $P_i$ for each $i \in \{0,\ldots,n\}$, where $P_i$ is defined by the map
\begin{equation}\label{Psubi3}
     i \mapsto P_i := \frac{1}{M}  \sum_{m = 1}^M  \mathbb{1}\{n_1\left(g^{(m)}\right) = i\}\,
     \end{equation}
for the overall network, where $n_1\left(g^{(m)}\right)$ is the size of community $1$ in the community assignment $g^{(m)}$.
In Figure \ref{samplemultilayer}, we show an example of such a set of histograms.

\begin{figure}[h]
\centering
\subfloat[Community-size histograms for each layer]{\includegraphics[height=0.24\textheight]{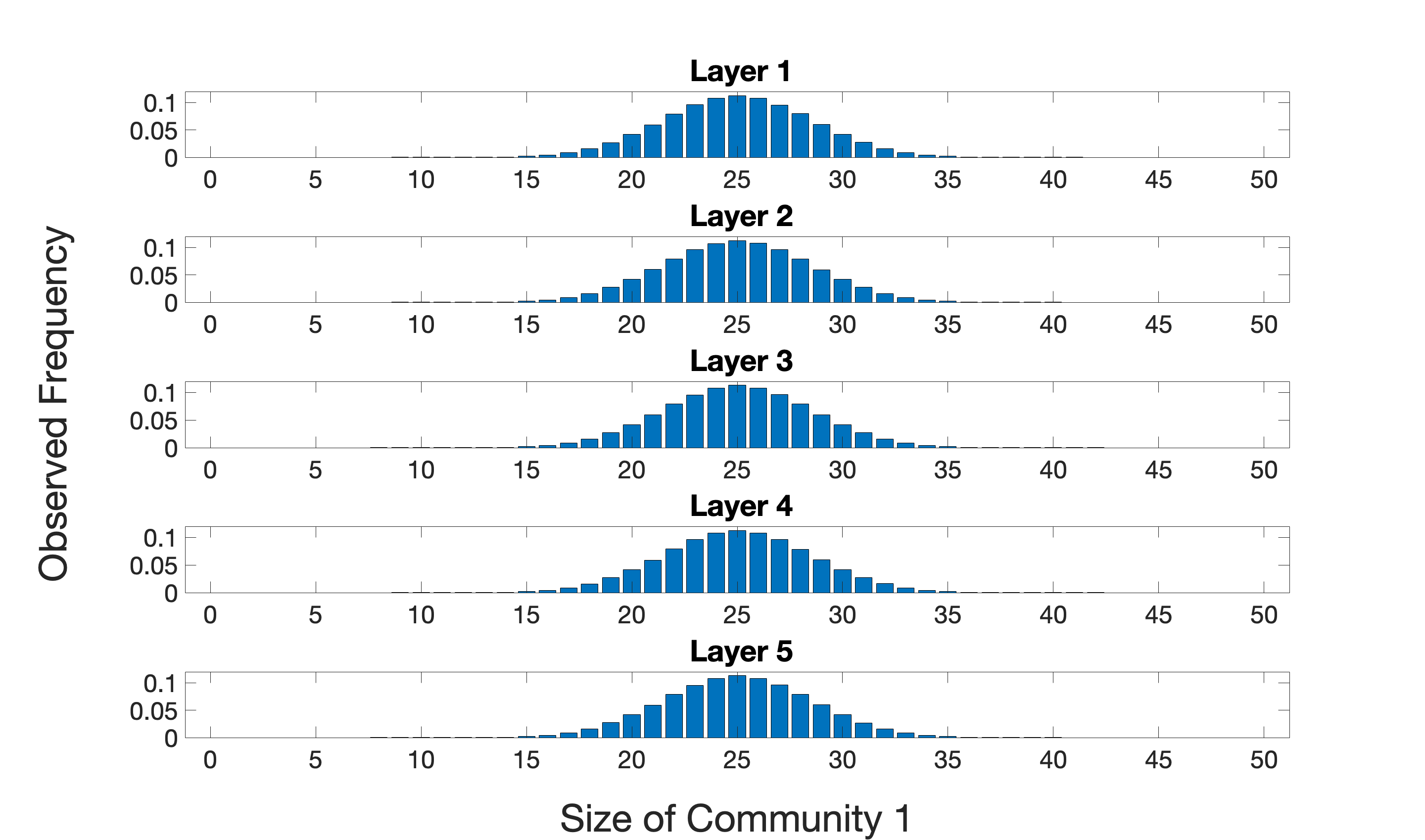}} \\
\subfloat[Overall community-size histograms for the network]{\includegraphics[height=0.24\textheight]{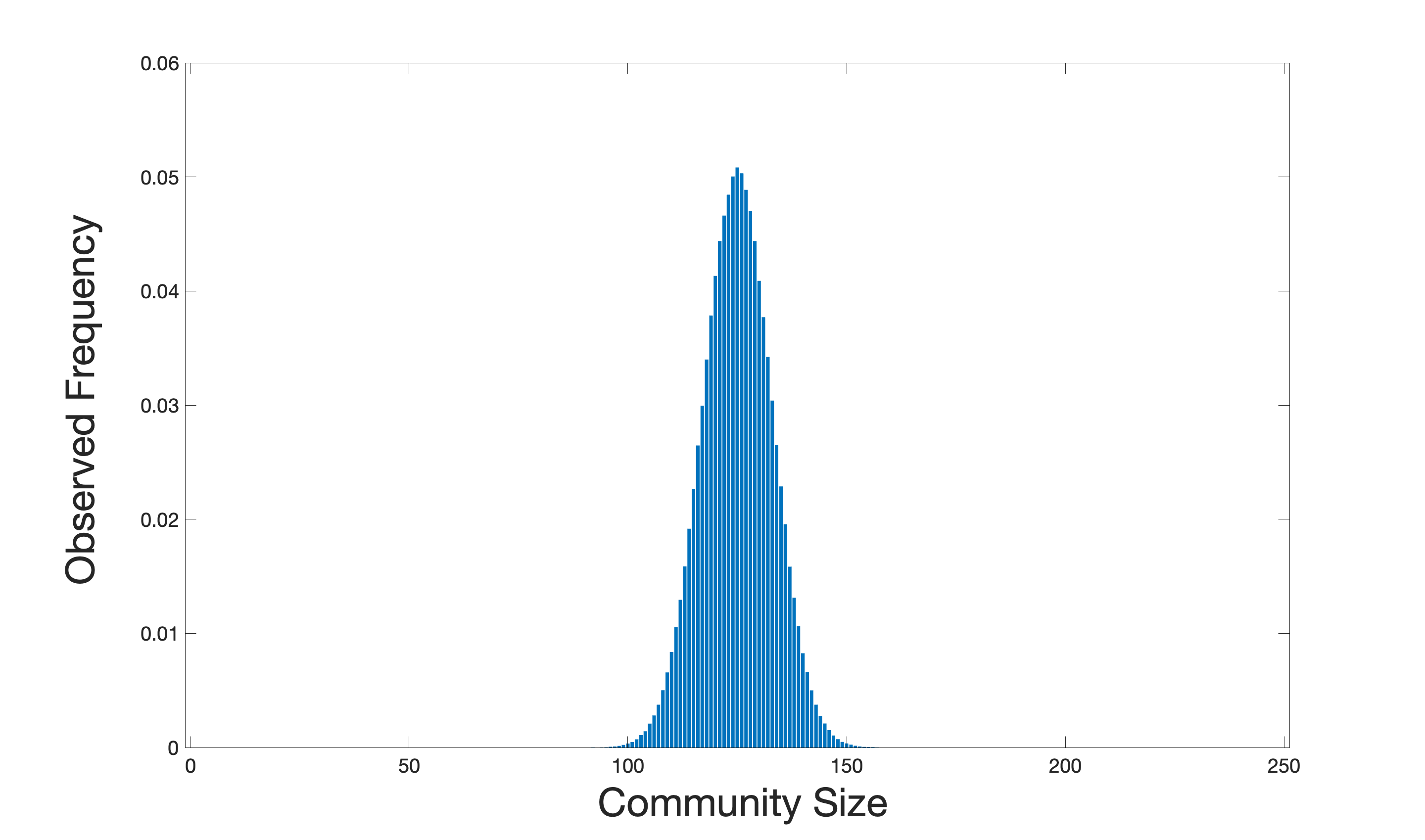}}
\caption[An example of the community-size histograms in a 50-node temporal network with $5$ layers.]{An example of the community-size histograms in a 50-node temporal network with $5$ layers. In both panels (a) and (b), the horizontal axis indicates the number of node-layers with community assignment $1$, and the vertical axis indicates the observed frequency of community assignments with that number of node-layers in community $1$. In the histograms in (a), we consider only the node-layers in the indicated layer when counting the number of nodes in community $1$. In the histogram in (b), we consider all of the node-layers in the network when counting the number of nodes in community $1$. }
\label{samplemultilayer}
\end{figure}

In addition to using histograms to qualitatively compare the amount of localization of \\ community-size distributions between different generative models, we calculate the inverse participation ratio (IPR) \cite{Kramer93} to quantify the amount of localization in each community-size distribution. The IPR of a community-size distribution is the squared $L^2$-norm of the distribution. That is,
\begin{equation}
	\text{IPR} = \sum_{i = 0}^n P_i^2 \, ,
\end{equation}
where $P_i$ is one of (\ref{Psubi1}), (\ref{Psubi2}), or (\ref{Psubi3}), depending on which type of community-size distribution we are considering. Larger values of the IPR indicate more localized distributions.
The minimum value of the IPR is $\frac{1}{n + 1}$ and is attained when $P_i = \frac{1}{n + 1}$ for each $i$. The maximum value of the IPR is $1$ and is attained when $P_j = 1$ for some $j$ and $P_i = 0$ for all $i \ne j$. We report $100 \times \text{IPR}$, which gives the IPR values in percentages.

In the following subsections, we examine the community-size histograms and calculate their IPRs to quantify and compare the amount of localization in the examined generative models for community detection.


\subsection{Community-Size Distributions for Monolayer-Network Models}\label{CSSL}

We consider two community-assignment probability distributions for monolayer-network models: a uniform distribution on community assignments (see Section \ref{unifDistSL}) and nodewise community assignments (see Section \ref{nodewiseSL}).  Using the approach in Section \ref{Methodology}, we generate a community-size histogram for each of the two community-assignment probability distributions. We set the number of nodes to $n = 50$ and generate two community-size histograms for each of the two community-assignment probability distributions. One distribution has $k = 2$ communities, and the other has $k = 5$ communities. 
In Figure~\ref{singlelayercomp}, we show the resulting community-size histograms.
We compute the IPR for the community-size distribution of each example and compile our results in Table \ref{iprSL}.

\begin{figure}[t]
	\centering
	\subfloat[$k = 2$; uniform distribution on community assignments]{
	  \includegraphics[width=0.485\linewidth]{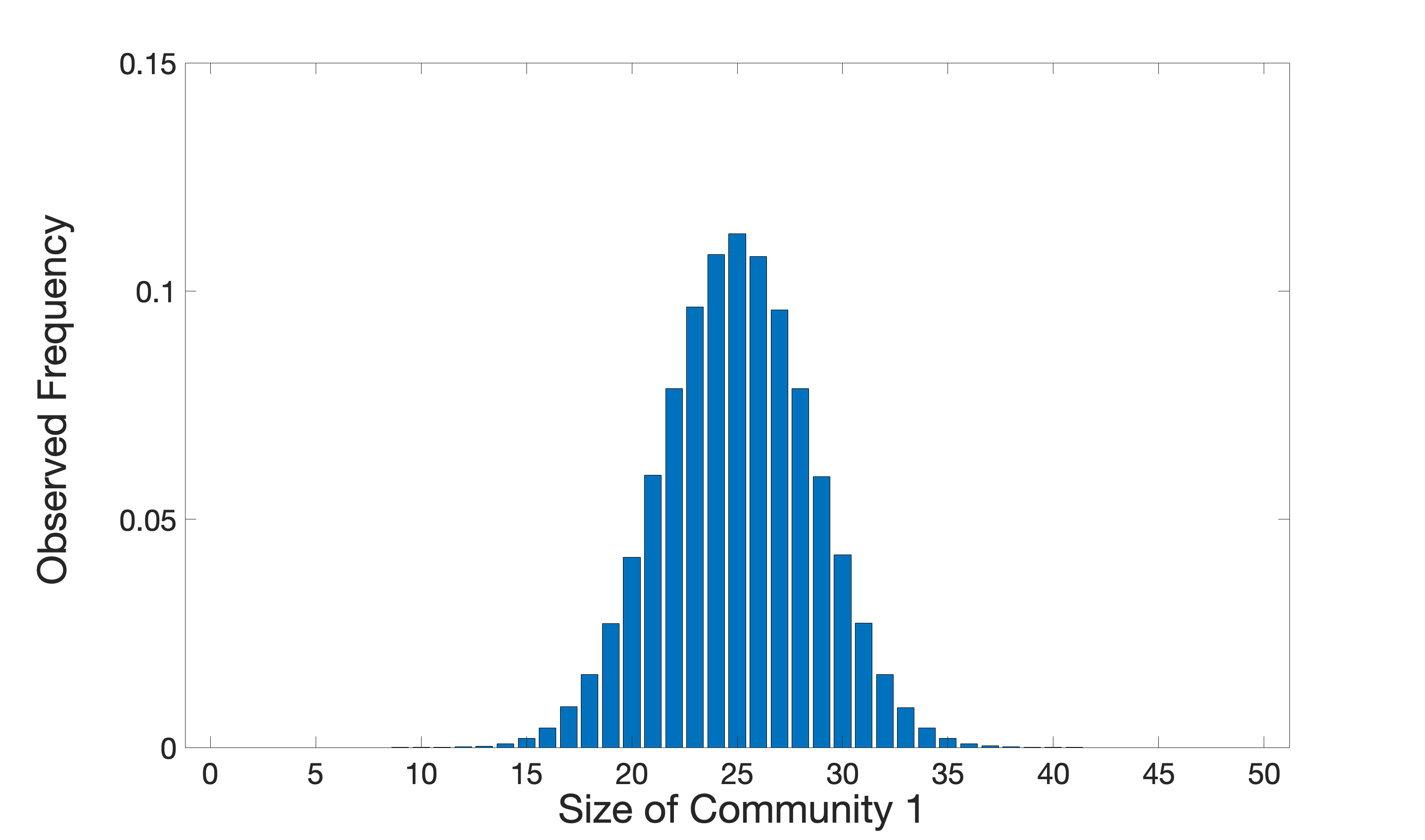}
	}\hfill
	\subfloat[$k = 2$; nodewise community assignments]{
	  \includegraphics[width=0.485\linewidth]{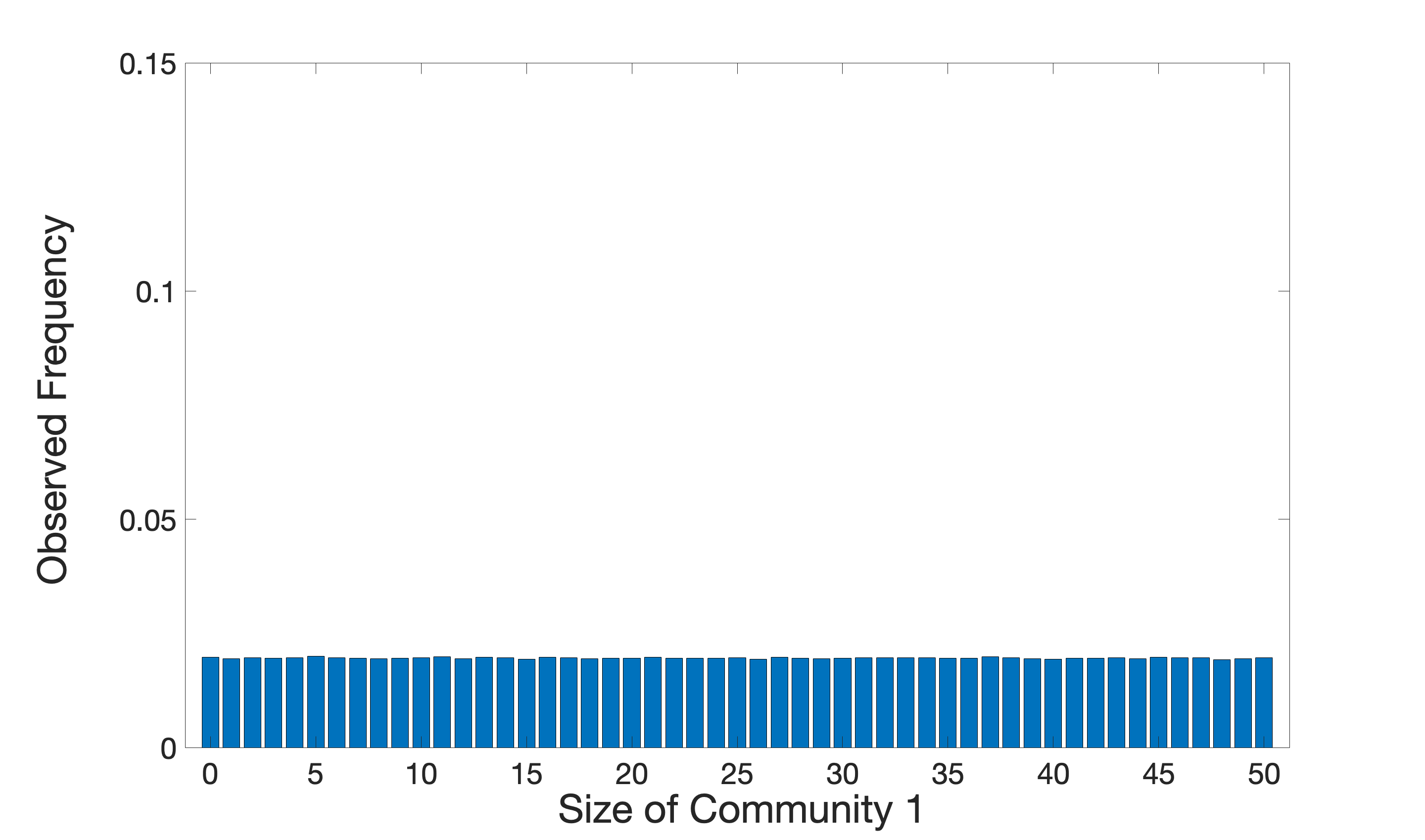}
	}\\[0.6em]
	\subfloat[$k = 5$; uniform distribution on community assignments]{
	  \includegraphics[width=0.485\linewidth]{fig4part1.png}
	}\hfill
	\subfloat[$k = 5$; nodewise community assignments]{
	  \includegraphics[width=0.485\linewidth]{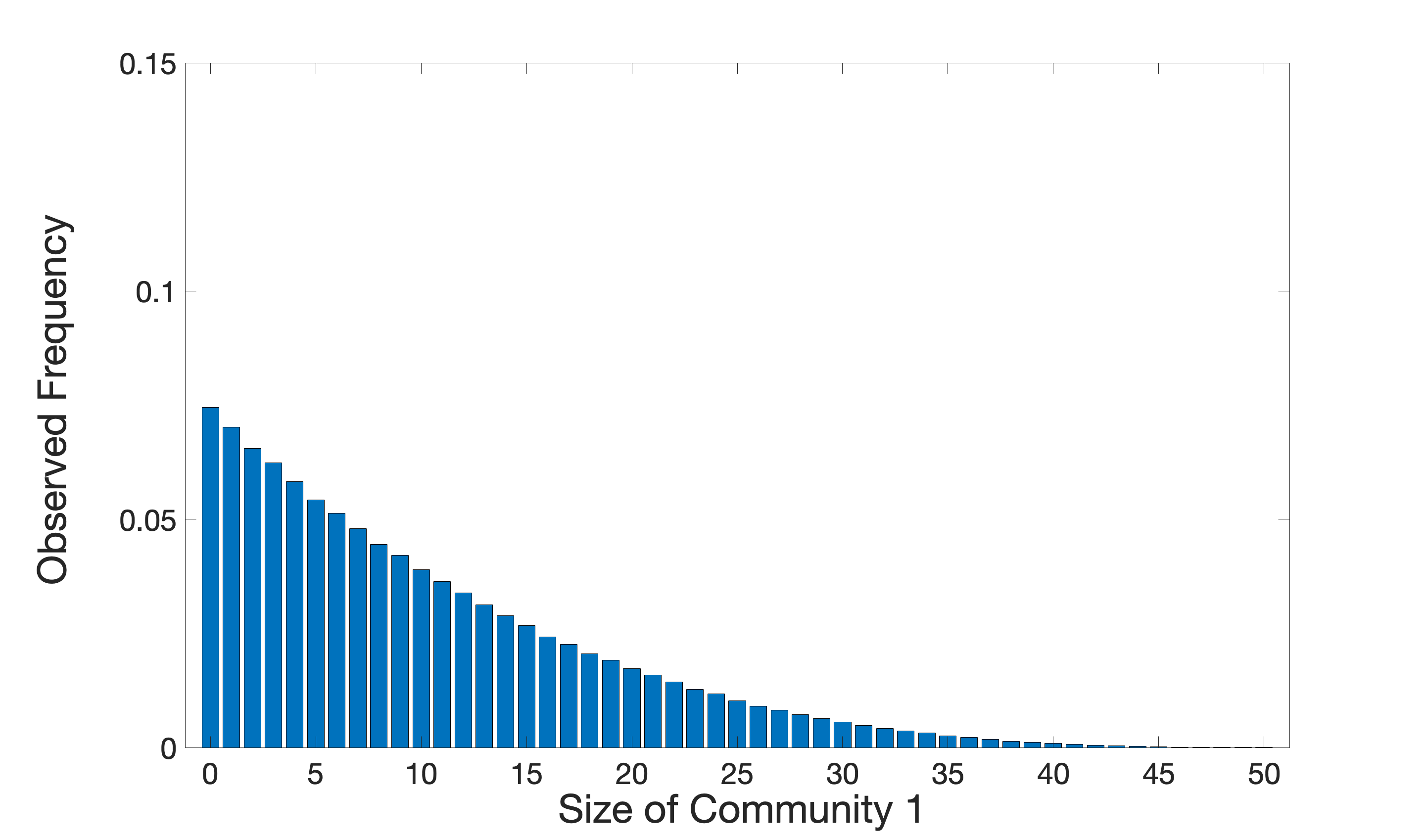}
	}
	\caption[Community-size histograms in 50-node monolayer networks ($k \in \{2,5\})$.]
	{
	Community-size histograms in 50-node monolayer networks for (a, b) $k = 2$ communities and (c, d) $k = 5$ communities. In panels (a) and (c), we show results for uniform distributions on community assignments. In panels (b) and (d), we show results for nodewise community assignments.
	}
	\label{singlelayercomp}
\end{figure}

\begin{table}[t]
	\centering
	\small
	\setlength{\tabcolsep}{6pt}
	\renewcommand{\arraystretch}{1.05}
	\begin{tabular}{lcc}
	\toprule
	Method & IPR for $k = 2$ (\%) & IPR for $k = 5$ (\%) \\
	\midrule
	Uniform distribution on community assignments         & 7.96 & 9.98 \\
	Nodewise community assignments           & 1.96 & 4.36 \\
	\bottomrule
	\end{tabular}
	\caption{
	Values of the inverse participation ratio (IPR) for community-size distributions for monolayer networks with $n = 50$ nodes and $k \in \{2,5\}$ communities.
	}
	\label{iprSL}
	\end{table}

We begin our analysis of our results by restating our claims (see Sections \ref{CommAssnSL} and \ref{CommSizeDistLoc}) for monolayer networks: 
\begin{itemize}
\item The community-size distribution for a uniform distribution on community assignments (see Section \ref{unifDistSL}) is highly localized. 
\item The community-size distribution for nodewise community assignments (see Section \ref{nodewiseSL}) is much less localized than that for a uniform distribution on community assignments (see Section \ref{unifDistSL}). 
\end{itemize}
{In Figure~\ref{singlelayercomp},} we see that both claims appear to hold for both {$k = 2$} communities and $k = 5$ communities. The histograms for a uniform distribution on community assignments appear to be highly localized, whereas the histograms for nodewise community assignments are far less localized. {Additionally, these histograms coincide with theoretical results about community-size localization in monolayer networks that we prove in Appendix \ref{sec:monolayer-theory}. For example, the distributions in Figures~\ref{singlelayercomp}(b) and \ref{singlelayercomp}(d) almost coincide with $\mathrm{Beta}(1,1)$ and $\mathrm{Beta}(4,1)$, respectively, as predicted by Corollary \ref{UnifDistLimitTheory}.}

The IPRs of the community-size distributions in Table \ref{iprSL} confirm these qualitative observations. For example, for $k = 2$, the IPR for the uniform distribution on community assignments is $0.0796$, which is much larger than the IPR value of $0.0196$ for nodewise community assignments. Recall that larger IPR values indicate more localization of a distribution. Therefore, for $k = 2$, the community-size distribution for nodewise community assignments is much less localized than the community-size distribution for a uniform distribution on community assignments. We observe the same behavior for $k = 5$. {The IPR values in Table \ref{iprSL} are similar to the values that we predict using our theoretical results in Appendix \ref{sec:monolayer-theory}. For example, for nodewise community assignments, Proposition \ref{prop:ipr-contrast} predicts IPR values of $2$\% for $k = 2$ and and $4.57$\% for $k = 5$. These values are close to the corresponding empirical IPR values of $1.96$\% and $4.36$\%.} {For uniform assignments, Proposition \ref{prop:ipr-contrast} predicts IPR values of $7.98$\% for $k = 2$ and $9.97$\% for $k = 5$, which are very close to the corresponding empirical IPR values of $7.96$\% and $9.98$\%.}


\subsection{Community-Size Distributions for Temporal-Network Models}\label{CSML}

We consider the following community-assignment probability distributions for temporal-network models: a uniform distribution on community assignments (see Section \ref{unifDistML}), the Yang et al.\ \cite{Yang11} and Bazzi et al.\ \cite{Bazzi20} discrete-time Markov-process approaches (see Section \ref{NodewiseEvolML}), and our LECS-prior-based approach (see Section \ref{novelML}). As with monolayer networks, we use the approach in Section \ref{Methodology} to generate a set of community-size histograms (one for each layer and one for the overall network) in temporal networks for each of the four generative models. We again set the number of nodes to $n = 50$ and again consider $k = 2$ and $k = 5$ communities. We show the histogram for each layer for $k = 2$ in Figure \ref{multilayercomp}, the overall histogram for the network for $k = 2$ in Figure~\ref{multilayercomp2}, the histogram for each layer for $k = 5$ in Figure \ref{multilayercomp3}, and the overall histogram for the network for $k = 5$ in Figure~\ref{multilayercomp4}. We compute the IPR for the community-size distribution for each example, and we compile the IPR values for each layer in Table \ref{iprML} and the IPR values for the overall network in Table \ref{iprML2}.

\begin{figure}[t]
\centering
\subfloat[Uniform distribution on community assignments]{\includegraphics[width=0.49\textwidth]{fig5part1.png}} \;
\subfloat[Yang et al.\ approach]{\includegraphics[width=0.49\textwidth]{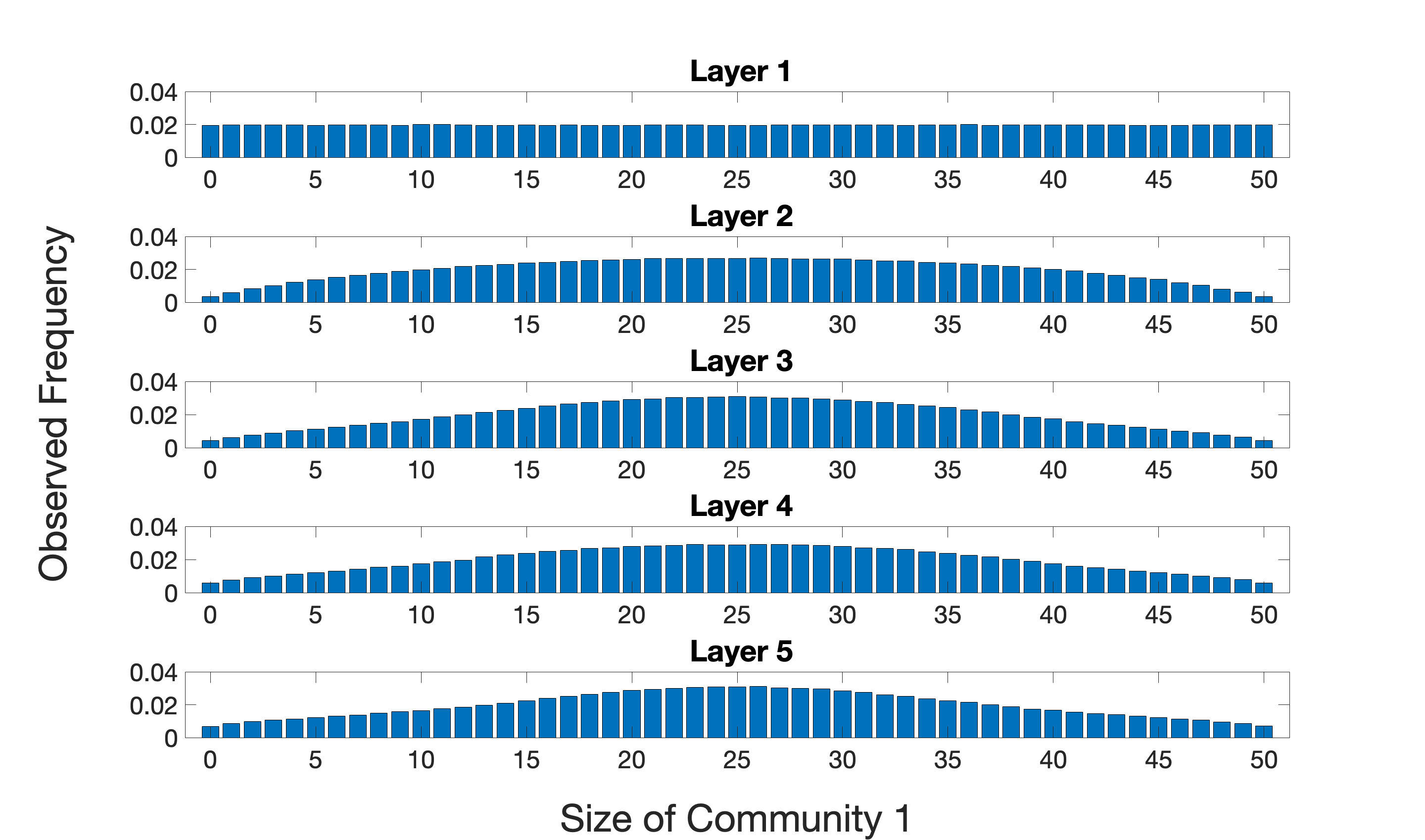}}\\
\subfloat[Bazzi et al.\ approach]{\includegraphics[width=0.49\textwidth]{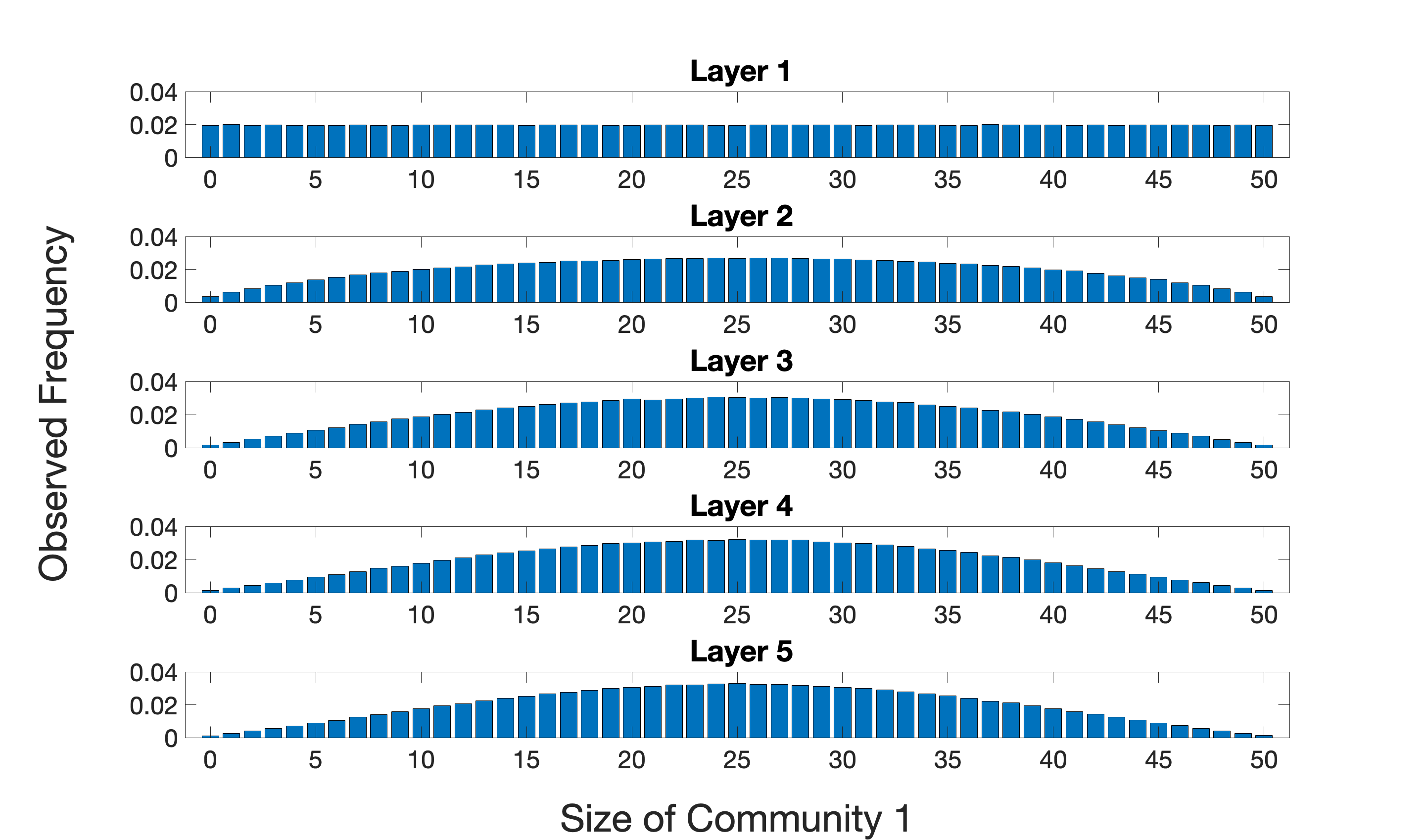}} \;
\subfloat[Our approach]{\includegraphics[width=0.49\textwidth]{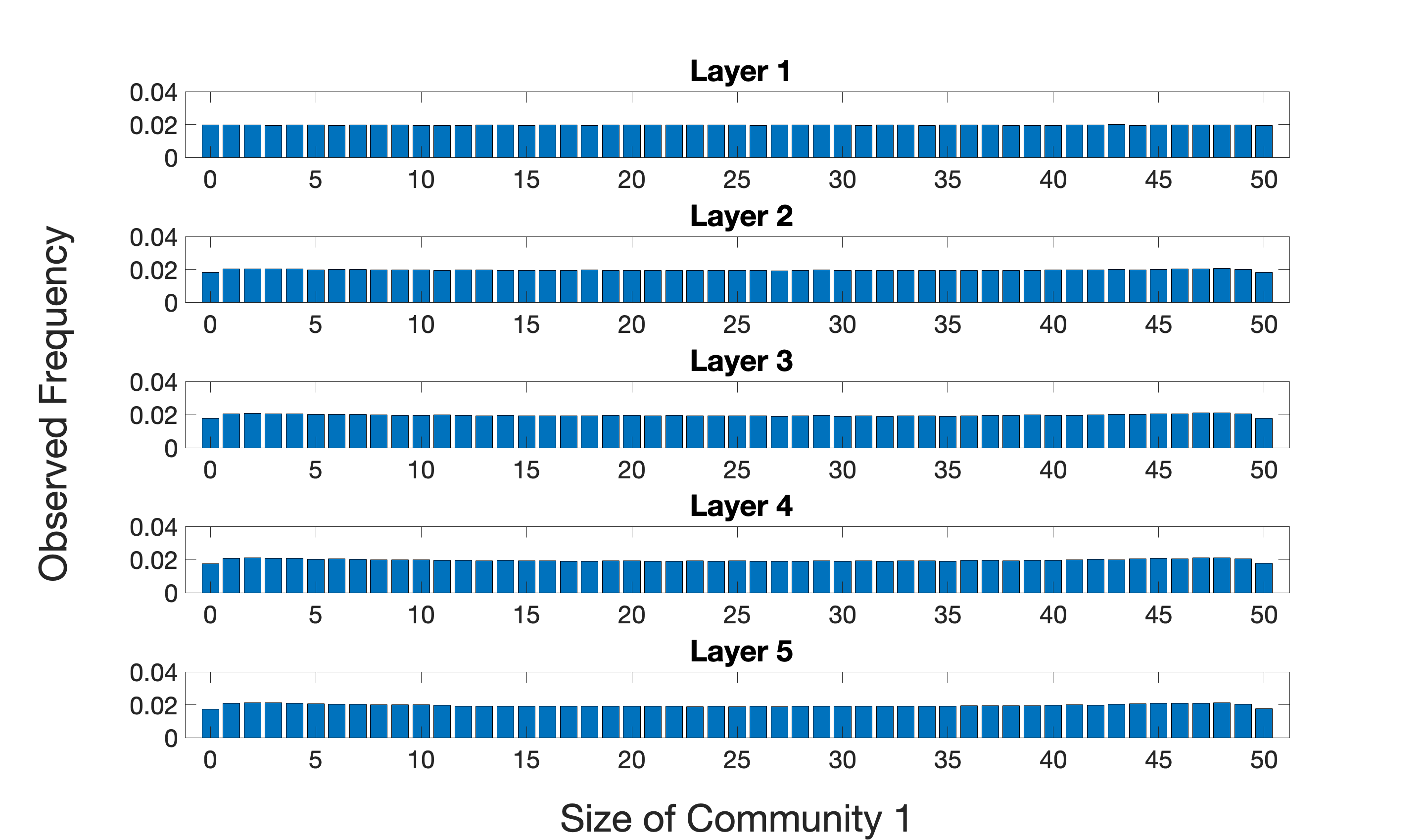}}
\caption[Community-size histograms for each layer of a $5$-layer, $50$-node temporal network 
 for $k = 2$ communities for a uniform distribution on community assignments, the Yang et al.\ approach \cite{Yang11}, the Bazzi et al.\ approach \cite{Bazzi20}, and our LECS-prior-based approach.]{Community-size histograms for each layer of a 5-layer, 50-node temporal network 
 for $k = 2$ communities for (a) a uniform distribution on community assignments, (b) the Yang et al.\ approach \cite{Yang11}, (c) the Bazzi et al.\ approach \cite{Bazzi20}, and (d) our LECS-prior-based approach.}
\label{multilayercomp}
\end{figure}

\begin{figure}[t]
\centering
\subfloat[Uniform distribution on community assignments]{\includegraphics[width=0.49\textwidth]{fig6part1.png}}\;
\subfloat[Yang et al.\ approach]{\includegraphics[width=0.49\textwidth]{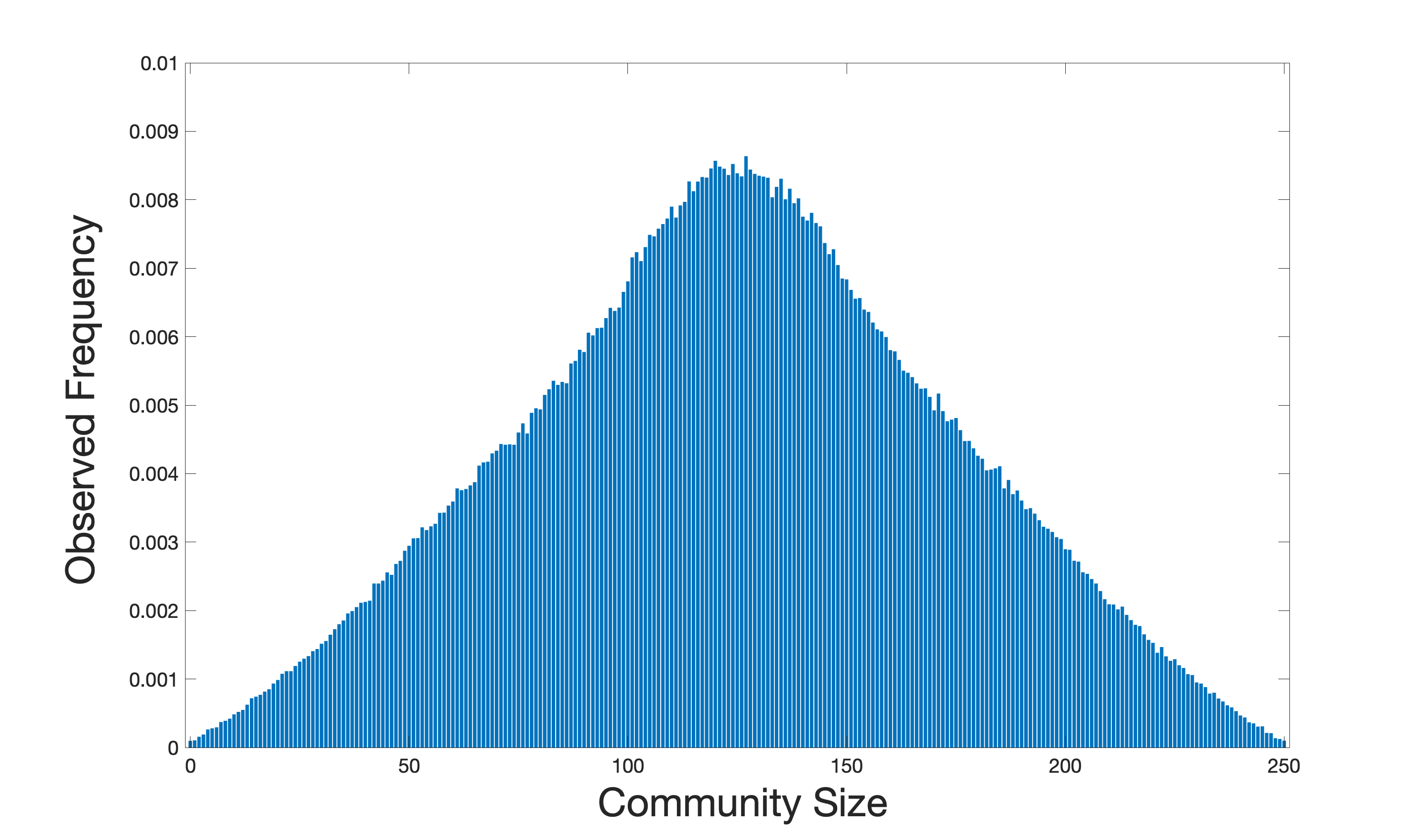}}\\
\subfloat[Bazzi et al.\ approach]{\includegraphics[width=0.49\textwidth]{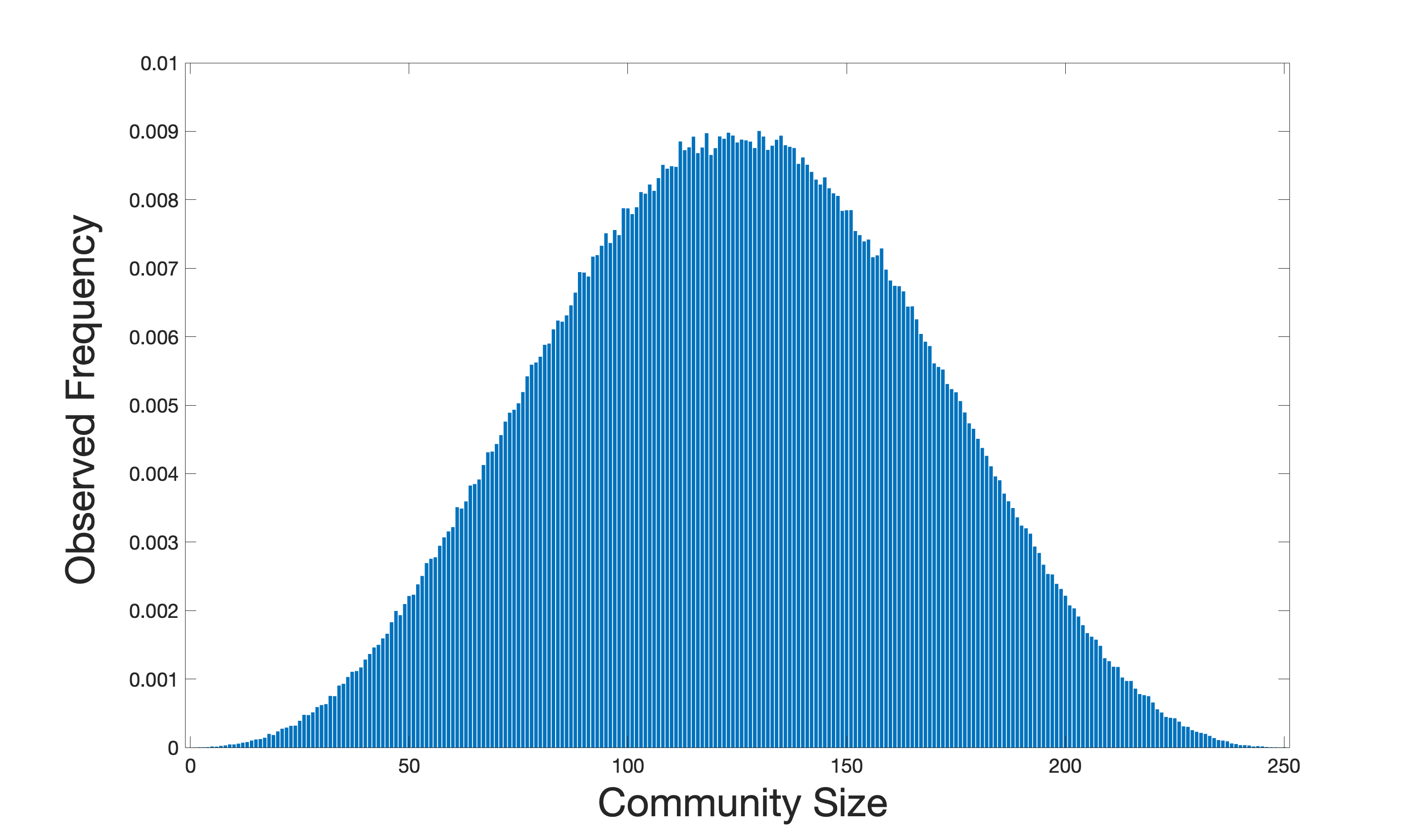}}\;
\subfloat[Our approach]{\includegraphics[width=0.49\textwidth]{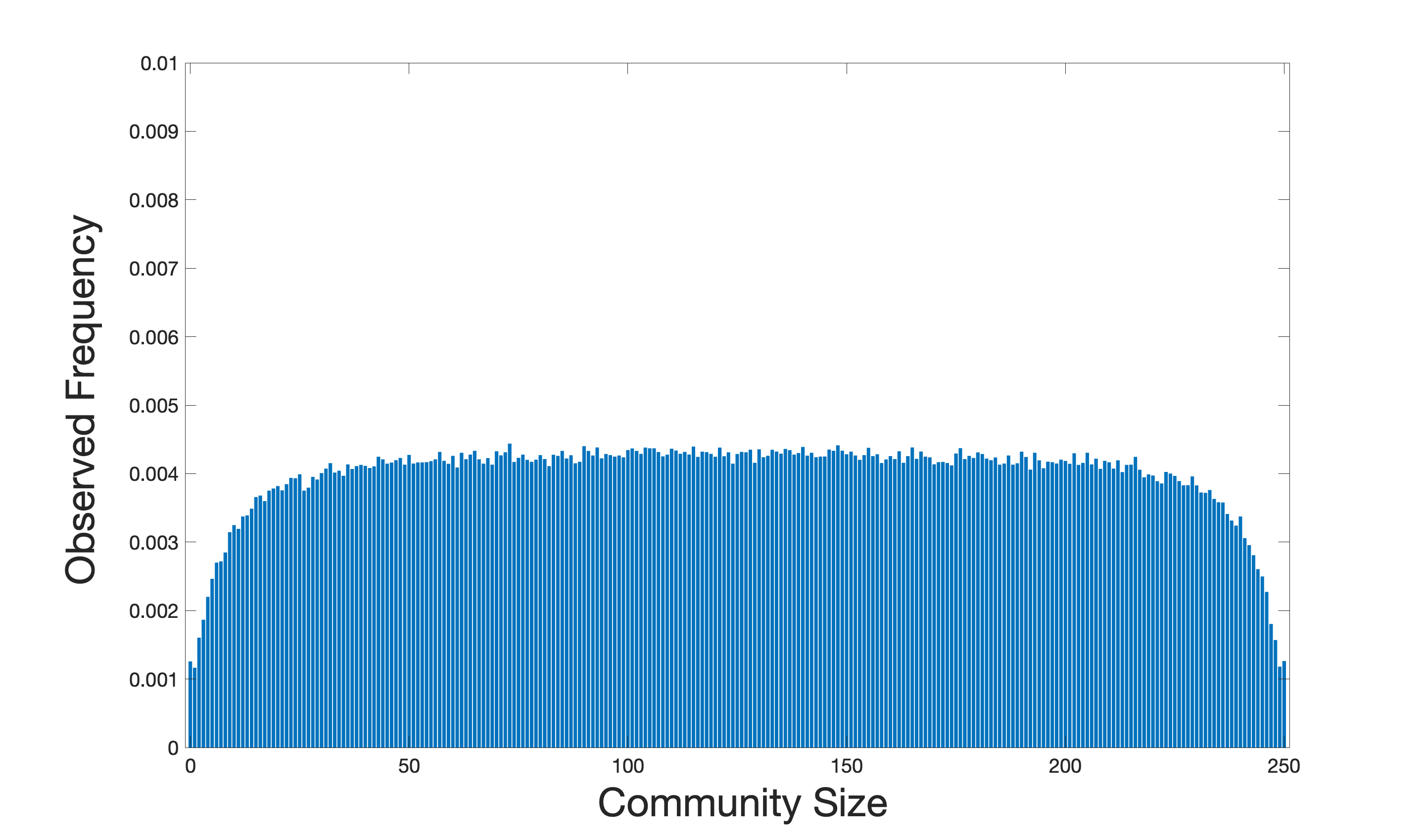}}
\caption[Overall community-size histograms in a 5-layer, 50-node temporal network
 for $k = 2$ communities for a uniform distribution on community assignments, the Yang et al.\ approach, the Bazzi et al.\ approach, and our LECS-prior-based approach.]{Overall community-size histograms in a 
 5-layer, 50-node temporal network
 for $k = 2$ communities for (a) a uniform distribution on community assignments, (b) the Yang et al.\ approach, (c) the Bazzi et al.\ approach, and (d) our LECS-prior-based approach.}
\label{multilayercomp2}
\end{figure}

\begin{figure}[t]\centering
\subfloat[Uniform distribution on community assignments]{\includegraphics[width=0.49\textwidth]{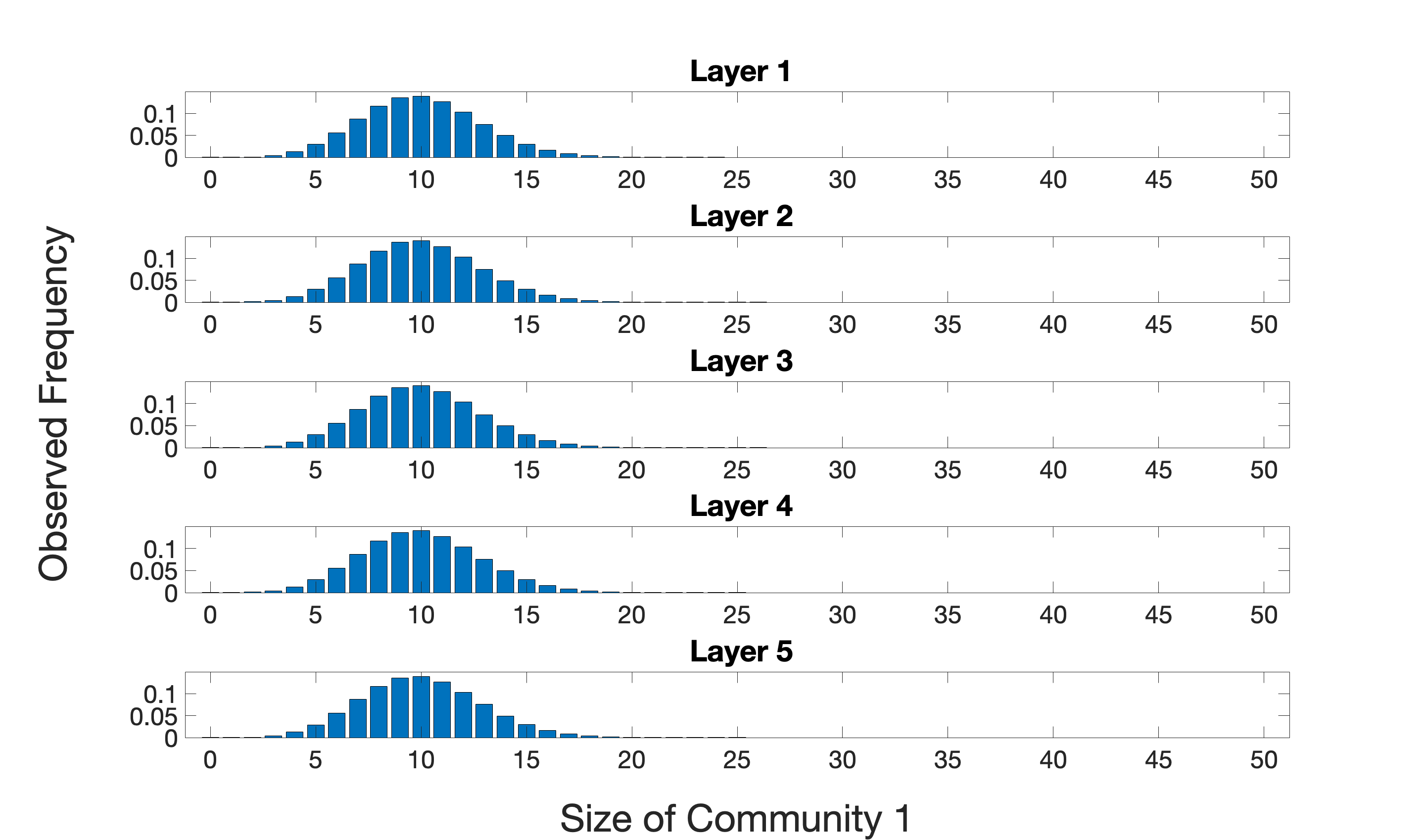}}\;
\subfloat[Yang et al.\ approach]{\includegraphics[width=0.49\textwidth]{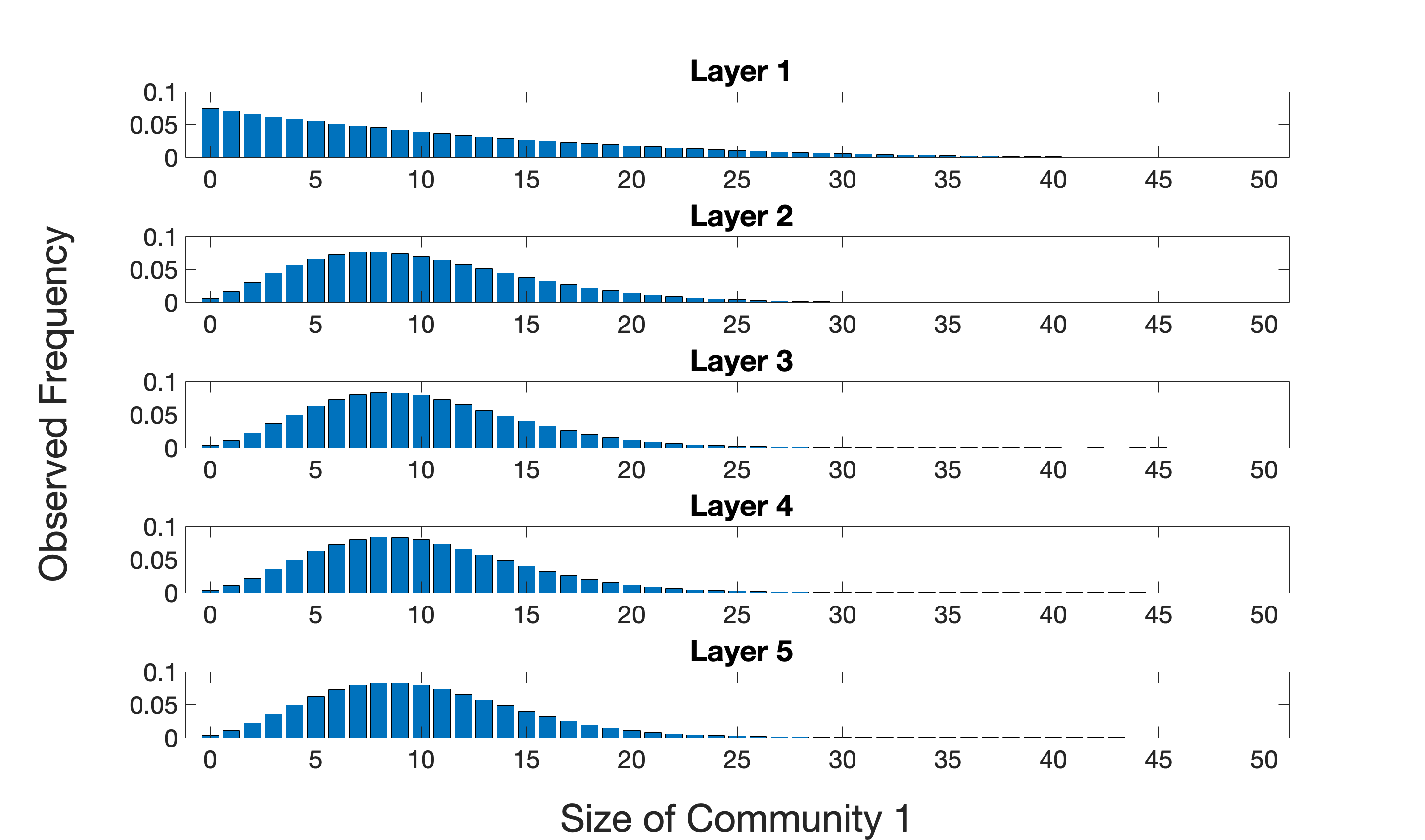}}\\
\subfloat[Bazzi et al.\ approach]{\includegraphics[width=0.49\textwidth]{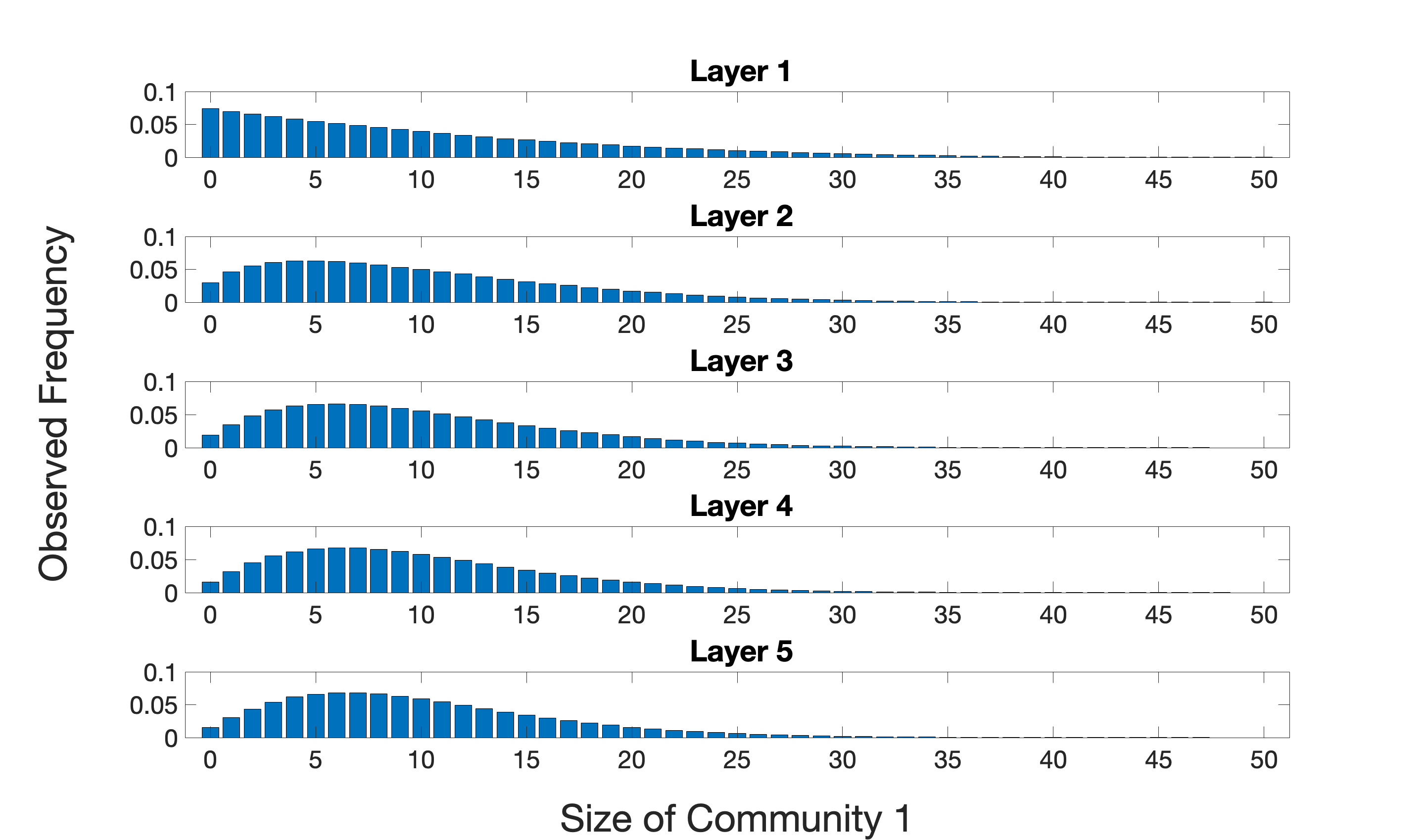}}\;
\subfloat[Our approach]{\includegraphics[width=0.49\textwidth]{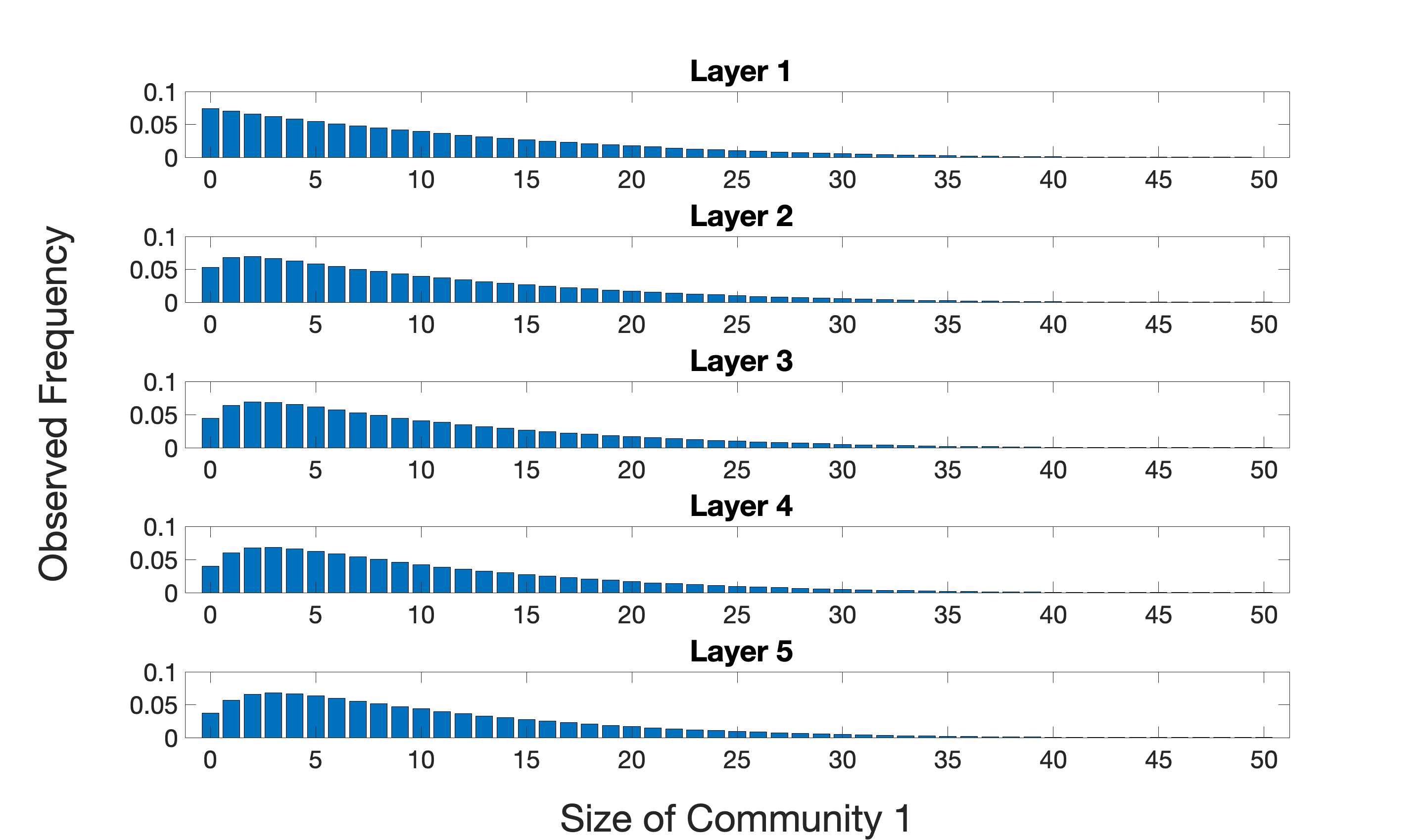}}
\caption[Community-size histograms in each layer of a 5-layer, 50-node temporal network 
 for $k = 5$ communities for a uniform distribution on community assignments, the Yang et al.\ approach, the Bazzi et al.\ approach, and our LECS-prior-based approach.]{Community-size histograms in each layer of a 
 5-layer, 50-node temporal network  
 for $k = 5$ communities for (a) a uniform distribution on community assignments, (b) the Yang et al.\ approach, (c) the Bazzi et al.\ approach, and (d) our LECS-prior-based approach.}
\label{multilayercomp3}
\end{figure}

\begin{figure}[t]\centering
\subfloat[Uniform distribution on community assignments]{\includegraphics[width=0.49\textwidth]{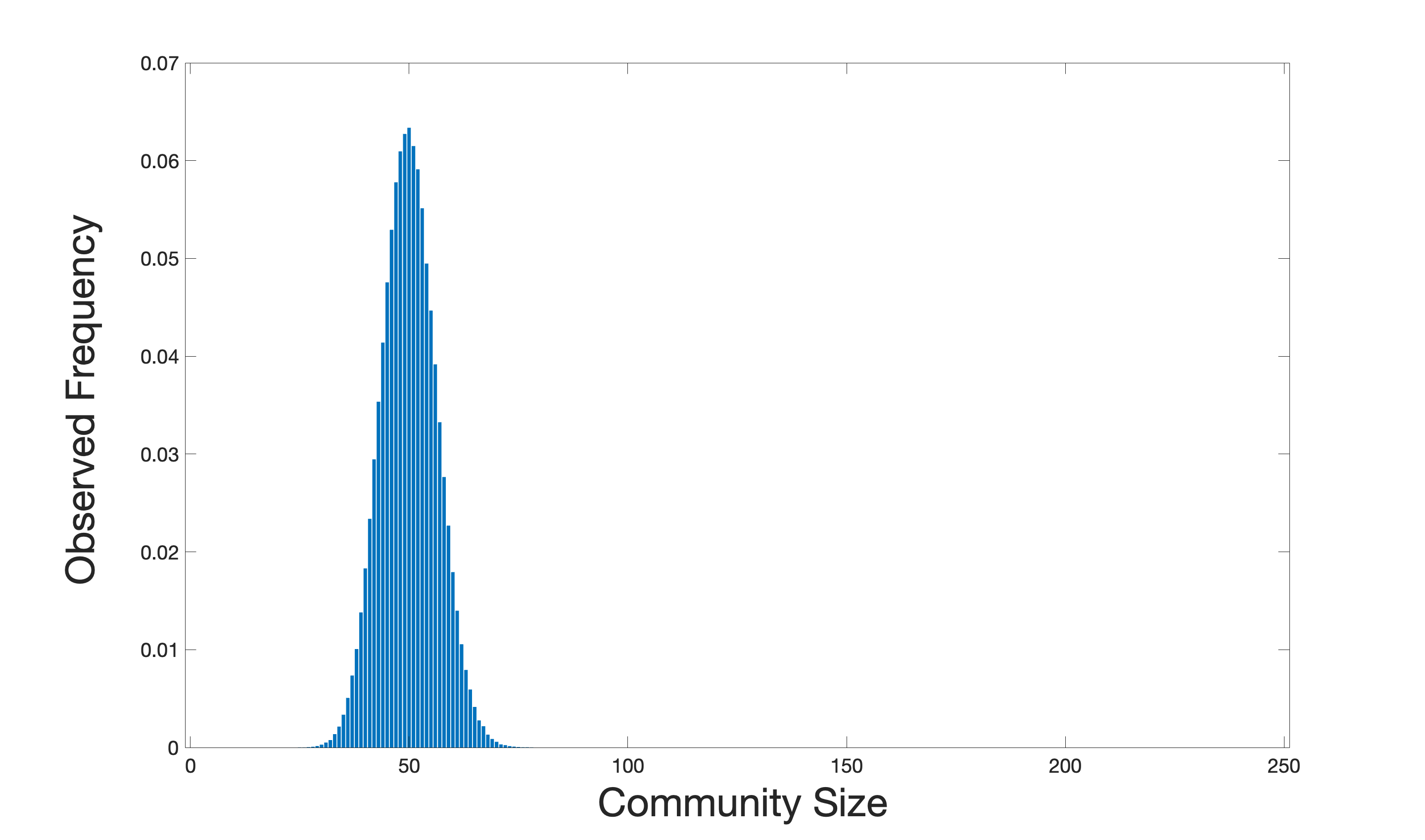}}\;
\subfloat[Yang et al.\ approach]{\includegraphics[width=0.49\textwidth]{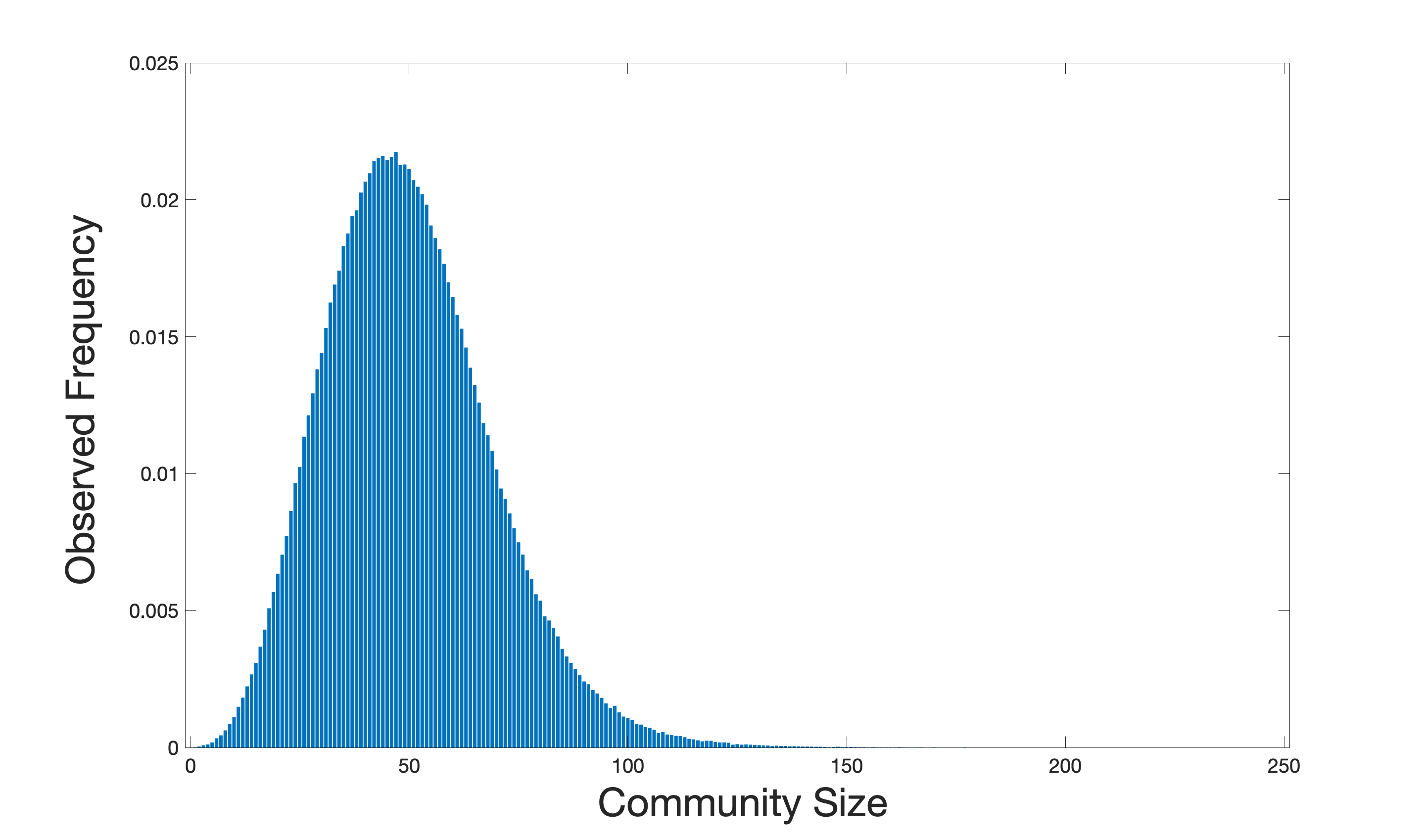}}\\
\subfloat[Bazzi et al.\ approach]{\includegraphics[width=0.49\textwidth]{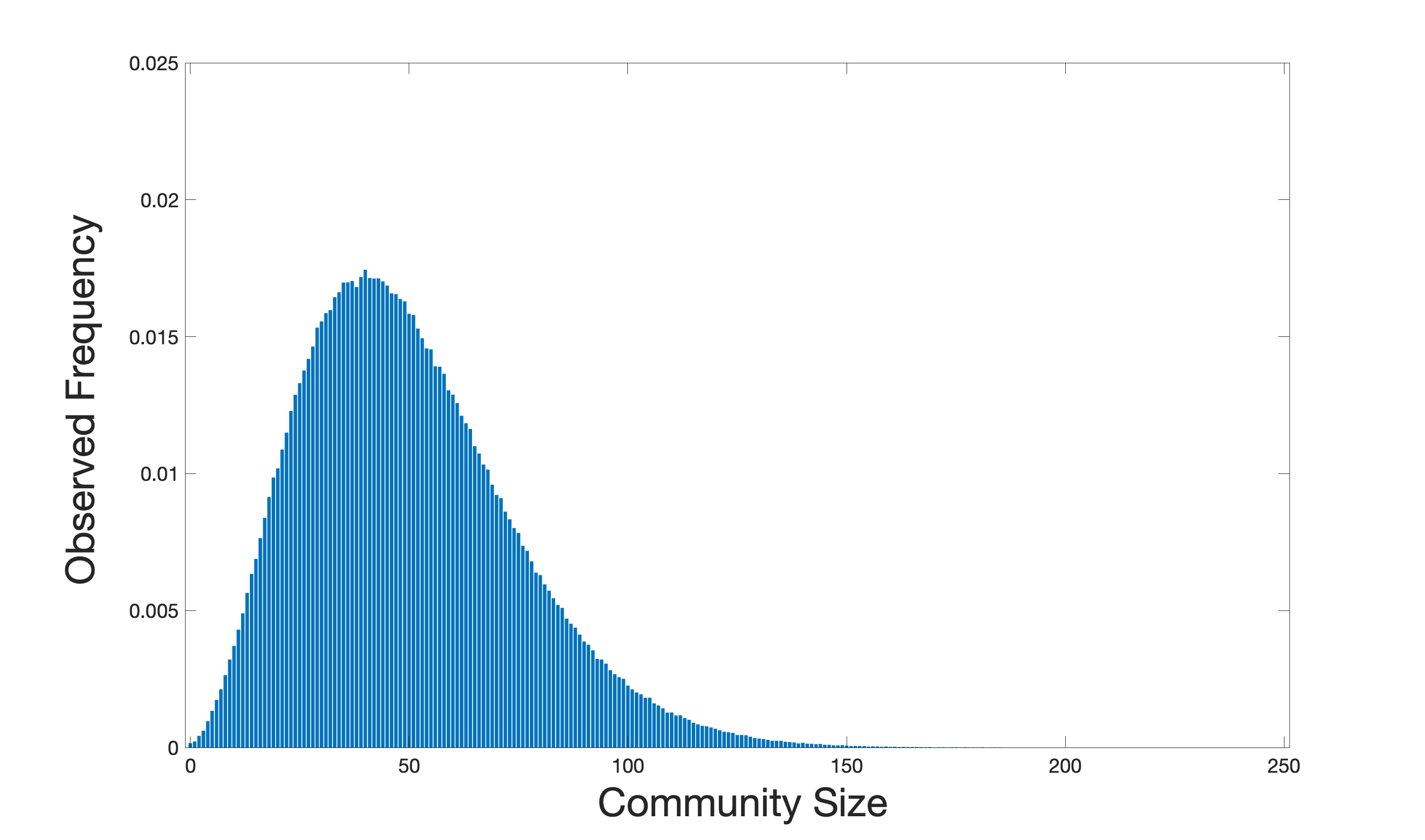}}\;
\subfloat[Our approach]{\includegraphics[width=0.49\textwidth]{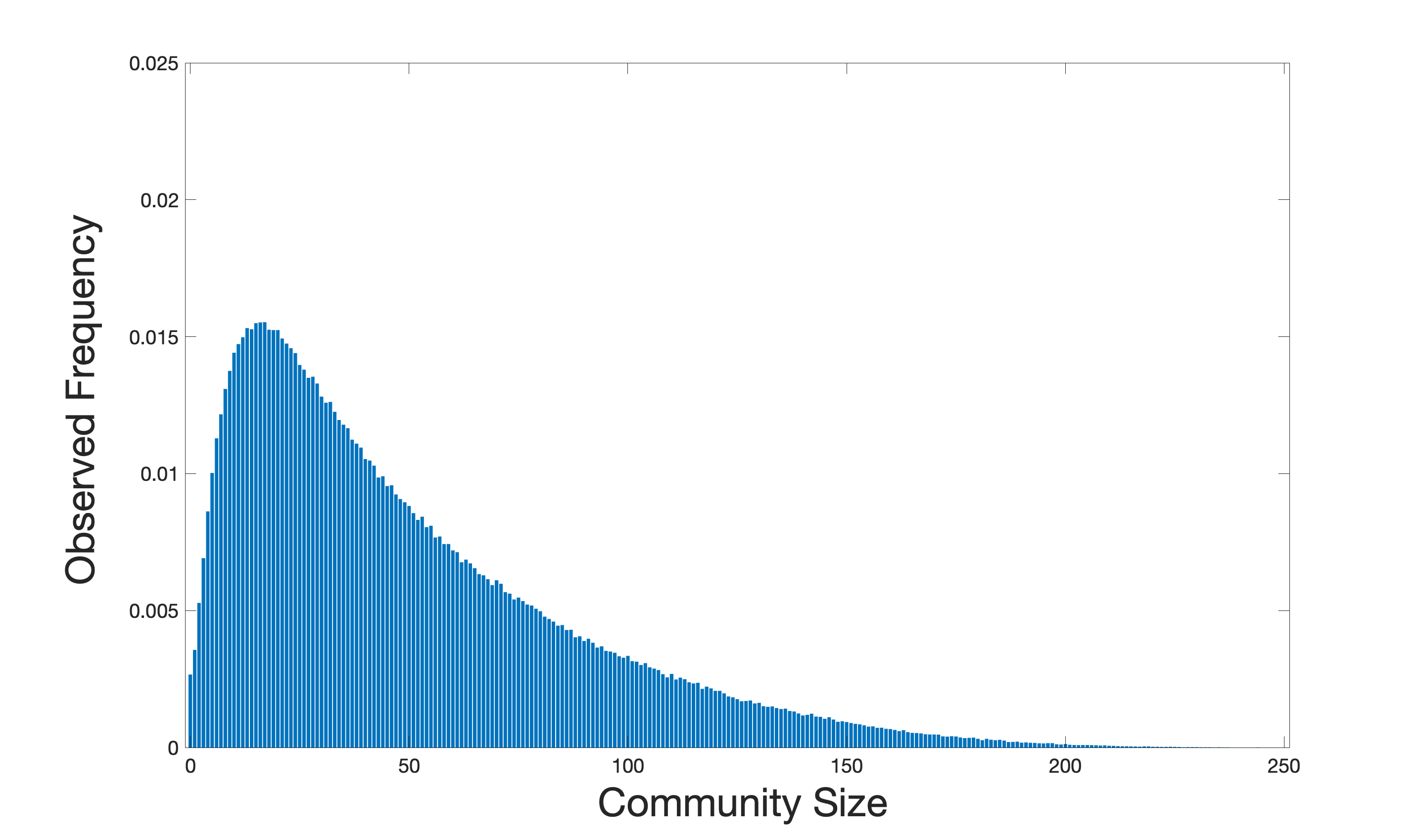}}
\caption[Overall community-size histograms in a 5-layer, 50-node temporal network
for $k = 5$ communities for a uniform distribution on community assignments, the Yang et al.\ approach, the Bazzi et al.\ approach, and our LECS-prior-based approach.]{Overall community-size histograms in a 
5-layer, 50-node temporal network 
for $k = 5$ communities for (a) a uniform distribution on community assignments, (b) the Yang et al.\ approach, (c) the Bazzi et al.\ approach, and (d) our LECS-prior-based approach.
}
\label{multilayercomp4}
\end{figure}

\begin{table}[t]
	\centering
	\small
	\setlength{\tabcolsep}{3pt}
	\renewcommand{\arraystretch}{1.05}
	\begin{tabular}{lcccccccccc}
	\toprule
	& \multicolumn{5}{c}{IPR for $k = 2$ (\%)} & \multicolumn{5}{c}{IPR for $k = 5$ (\%)} \\
	\cmidrule(lr){2-6}\cmidrule(lr){7-11}
	Method & $L_1$ & $L_2$ & $L_3$ & $L_4$ & $L_5$ & $L_1$ & $L_2$ & $L_3$ & $L_4$ & $L_5$ \\
	\midrule
	Uniform on community assignments
	& 7.95 & 7.94 & 7.98 & 7.95 & 7.97 & 9.97 & 9.98 & 9.97 & 9.97 & 9.98 \\
	Yang et al.\ \cite{Yang11}
	& 1.96 & 2.20 & 2.32 & 2.25 & 2.26 & 4.36 & 5.45 & 5.95 & 5.97 & 5.97 \\
	Bazzi et al.\ \cite{Bazzi20}
	& 1.96 & 2.20 & 2.39 & 2.47 & 2.50 & 4.35 & 4.48 & 4.71 & 4.81 & 4.85 \\
	LECS prior
	& 1.96 & 1.96 & 1.96 & 1.96 & 1.96 & 4.35 & 4.37 & 4.40 & 4.42 & 4.43 \\
	\bottomrule
	\end{tabular}
	\caption{IPR of the community-size distribution of each layer for each temporal-network model for a 5-layer, 50-node temporal network. 
	 We denote layer $t$ by $L_t$.
	}
	\label{iprML}
\end{table}

\begin{table}[t]
	\centering
	\small
	\setlength{\tabcolsep}{6pt}
	\renewcommand{\arraystretch}{1.05}
	\begin{tabular}{lcc}
	\toprule
	Method & IPR for $k = 2$ (\%) & IPR for $k = 5$ (\%) \\
	\midrule
	Uniform on community assignments                 & 3.57 & 4.46 \\
	Yang et al.\ \cite{Yang11}         & 0.57 & 1.53 \\
	Bazzi et al.\ \cite{Bazzi20}       & 0.66 & 1.22 \\
	LECS prior                               & 0.41 & 0.95 \\
	\bottomrule
	\end{tabular}
	\caption{IPR of the community-size distribution of the overall network for each temporal-network model for a 5-layer, 50-node temporal network.
	}\label{iprML2}
\end{table}

For convenience, we restate our previous claims (see the introduction of Section \ref{CommSizeDistLoc}) about temporal networks: 
\begin{itemize}
\item The community-size distribution of a uniform distribution on community assignments (see Section \ref{unifDistML}) is highly localized. 
\item The community-size distributions of the Yang et al.\ \cite{Yang11} and Bazzi et al.\ \cite{Bazzi20} discrete-time Markov-process approaches (see Section \ref{NodewiseEvolML}) are much less localized than that of a uniform distribution on community assignments.
However, for both the Yang et al.\ and Bazzi et al.\ approaches, the single-layer community-size distributions of later layers are more localized than those of earlier layers. 
\item In our approach (see Section \ref{novelML}), the localization of the single-layer community-size distributions increases much more slowly than it does for the Yang et al.\  \cite{Yang11} and Bazzi et al.\ \cite{Bazzi20} approaches. Consequently, the overall community-size distribution is much less localized than those of the Yang et al.\ and Bazzi et al.\ approaches. 
\end{itemize}
In Figures \ref{multilayercomp}, \ref{multilayercomp2}, \ref{multilayercomp3}, and \ref{multilayercomp4}, we see that all of these claims appear to hold for both $k = 2$ communities and $k = 5$ communities. In the example in Figure \ref{multilayercomp2}, we see for $k = 2$ that the overall community-size histogram is much more localized for a uniform distribution on community assignments than for the three other methods. We also observe that the overall community-size histogram 
appears to be less localized in our approach than in those of the Yang et al.\ \cite{Yang11} or Bazzi et al.\ \cite{Bazzi20} approaches.

Our qualitative observations are also confirmed by the IPR values in Tables \ref{iprML} and \ref{iprML2}. For example, for $k = 2$, the IPR of the overall community-size distribution is $0.0357$ for a uniform distribution on community assignments, $0.0057$ for the Yang et al. approach, $0.0066$ for the Bazzi et al. approach, and $0.0041$ for our LECS-prior-based approach, confirming that the overall community-size distribution is much more localized for a uniform distribution on community assignments than for the three other methods and that the community-size distribution for our approach is less localized those of the Yang et al.\ and Bazzi et al.\ approaches. Additionally, for $k = 2$, the IPRs of the single-layer community-size distributions for the Yang et al.\ approach increase from $0.0196$ in layer $1$ to $0.0226$ in layer $5$, illustrating that the single-layer community-size distributions become more localized over time.


\subsection{Asymptotic Behavior of our LECS prior} \label{ArashNumerics}

{We have shown empirically (see Section~\ref{CSML}) that the localization of the single-layer community-size distributions increases much more slowly for our LECS prior than it does for the Yang et al.\  \cite{Yang11} and Bazzi et al.\ \cite{Bazzi20} approaches.}

{In this subsection, we study the LECS prior with a minor modification. Specifically, we fix the retention probability $p_{r,\ell}$ in~\eqref{eq:lecs-count-split} to $p_{r,\ell} = p$, rather than sampling it independently from  $\mathrm{Unif}(0,1)$. We refer to this version as the ``LECS prior with fixed retention probability" $p$. We prove the following result.

\begin{theorem}\label{ArashTheorem}
For $k = 2$ communities, as the layer $\ell \to \infty$, the single-layer community-size distributions of {the LECS prior with fixed retention probability $p$} satisfy
\begin{equation}\label{stationary_dist}
	P_i \to \frac{(1 - p^{i + 1})(1 - p^{n - i + 1})}{\sum_{j = 0}^n (1 - p^{j + 1})(1 - p^{n - j + 1})} \, ,
\end{equation}
where $P_i$ is defined in \eqref{Psubi2}. 
\end{theorem}

\begin{proof}
With $k = 2$ communities, we have
$n_2^{(\ell)} = n - n_1^{(\ell)}$, where we recall from Section~\ref{novelML} that $n_r^{(\ell)}$ is the size of community $r$ in layer $\ell$. Therefore, it suffices to consider only the distribution of $n_1^{(\ell)}$. 

Let $\geom(n,p)$ denote the distribution of a random variable $V$ with probability mass function 
\begin{equation*}
	\mathbb P(V = m) = \frac{p^{m} (1 - p)}{1 - p^{n + 1}} \, , 
\end{equation*}
where $m \in \{0,1,\dots,n\}$. Note that $n - V \sim \geomb(n,p)$, which we defined in Section~\ref{novelML}.

By definition, 
\begin{equation}\label{group1sizes}
	n_1^{(\ell + 1)} = c^{(\ell + 1)}_{11} + c^{(\ell + 1)}_{21} \, ,
\end{equation} where we recall from Section~\ref{novelML} that $c^{(\ell)}_{rs}$ is the number of nodes in community $r$ in layer $\ell - 1$ that are in
 community $s$ in layer $\ell$. Similarly, 
\begin{equation} \label{group2sizes}
	c_{21}^{(\ell + 1)} = n_2^{(\ell)} - c_{22}^{(\ell + 1)} \, .
\end{equation}

Using the sampling procedure \eqref{eq:lecs-count-split} and the fact that 
$n - V \sim \geomb(n,p)$, we have 
\begin{equation}\label{V_def}
	c_{rr}^{(\ell + 1)} = n_r^{(\ell)} - V_r 
\end{equation}
for $V_r$ sampled independently from
\begin{equation}\label{V_dist}
	V_r \sim \geom(n_r^{(\ell)}, p) \, ,
\end{equation} 
where $r \in \{1,2\}$.

Combining \eqref{group1sizes}, \eqref{group2sizes}, and \eqref{V_def} yields
\begin{align}\label{c_markov}
	n_1^{(\ell + 1)} = n_1^{(\ell)} - V_1 + V_2 \, .
\end{align}
Therefore, equations {(\ref{V_dist}) and (\ref{c_markov}) define} a Markov chain on $n_1^{(\ell)}$. 
Consequently, to prove Theorem \ref{ArashTheorem}, it suffices to show that the Markov chain (\ref{V_dist}, \ref{c_markov}) 
is ergodic with a unique stationary distribution
\begin{equation}\label{stationary_dist:proof}
	\pi(i) = \frac{(1 - p^{i + 1})(1 - p^{n - i + 1})}{\sum_{j=0}^n (1 - p^{j + 1})(1 - p^{n - j + 1})} \, .
\end{equation}

The Markov chain (\ref{V_dist}, \ref{c_markov}) is clearly ergodic, so it suffices to show that the stationary distribution is 
\eqref{stationary_dist:proof}. 
The transition kernel $P$ of the Markov chain is 
\begin{equation}\label{transition_kernel}
	P(i,k) = \sum_{\substack{0 \, \le \, v_1 \, \le \, i \\ 0 \, \le \, v_2 \, \le \, n - i}} \mathbb{1}\{k = i - v_1 + v_2\}\frac{p^{v_1}(1 - p)}{1 - p^{i + 1}}\frac{p^{v_2}(1 - p)}{1 - p^{n - i + 1}} \, ,
\end{equation}
where $P(i,k)$ is the probability that $n_1^{(\ell + 1)} = k$ given that $n_1^{(\ell)} = i$.
To prove that the stationary distribution of the Markov chain (\ref{V_dist},\,\ref{c_markov}) is
\eqref{stationary_dist:proof}, it suffices to show that the stationary distribution is
\begin{equation}\label{stationary_dist:2}
	\pi(i) = C_{n,p} \,(1 - p^{i+1})(1 - p^{n - i + 1})
\end{equation}
for some constant $C_{n, p} > 0$. Using the expression \eqref{transition_kernel} for the transition kernel $P$, we see that
\begin{align}
	\sum_{i=0}^n \pi(i) P(i,k) 
		&= C_{n,p} \sum_{i=0}^n \sum_{\substack{0 \, \le \, v_1 \, \le \, i \\ 0 \, \le \, v_2 \, \le \, n - i}} \mathbb{1}\{k = i - v_1 + v_2\}	\, p^{v_1} (1 - p) p^{v_2} (1 - p) \notag \\
		&= C_{n,p} \sum_{v_1 = 0}^n \sum_{v_2 = 0}^n \sum_{i=0}^n \mathbb{1}\{k = i - v_1 + v_2 \,, \, v_1 \le i \,, \, v_2 \le n - i\} \, p^{v_1} (1 - p) p^{v_2} (1 - p) \notag \\
		&= C_{n,p} \sum_{v_1 = 0}^n \sum_{v_2 = 0}^n \left[p^{v_1} (1 - p) p^{v_2} (1 - p) \sum_{i=0}^n \mathbb{1}\{k = i - v_1 + v_2 \,, \, v_1 \le i \,, \, v_2 \le n - i\} \right] \, .\label{ArashProofSum1}
\end{align}
We rewrite the conditions $k = i - v_1 + v_2$, $v_1 \le i$, and $v_2 \le n - i$ 
as $i = k + v_1 - v_2$, $v_2 \le k$, and $v_1 \le n - k$, and we thereby obtain
\begin{align}
	\sum_{i = 0}^n \mathbb{1}\{k = i - v_1 + v_2 \,, \,v_1 \le i \,, \,v_2 \le n - i\} 
			&= \sum_{i = 0}^n \mathbb{1}\{i = k + v_1 - v_2 \,, \, v_2 \le k \,, \, v_1 \le n - k\} \notag \\
			&= \mathbb{1} \{v_2 \le k, \, v_1 \le n - k\} \, . \label{ArashProofSum2}
\end{align}
Combining \eqref{ArashProofSum1} and \eqref{ArashProofSum2} yields
\begin{align*}
	\sum_{i=0}^n \pi(i) P(i, k)  &= C_{n, p} \sum_{v_1 = 0}^n \sum_{v_2 = 0}^n \mathbb{1}\{v_2 \le k \,, \,v_1 \le n - k\} \, p^{v_1} (1 - p) p^{v_2} (1 - p) \\
		&= C_{n,p} (1 - p^{k + 1})(1 - p^{n - k + 1}) = \pi(k) \,.
\end{align*}
The above argument holds for all $i$ and $k$, so the stationary distribution of the Markov chain (\ref{V_dist}, \ref{c_markov}) is \eqref{stationary_dist:2}.
\end{proof}

Theorem \ref{ArashTheorem} supports our claim that our LECS prior
results in 
less-localized community-size distributions than Markov-process models (such as the Bazzi et al.~\cite{Bazzi20} and Yang et al.~\cite{Yang11} methods) because the 
distribution \eqref{stationary_dist}
is almost uniform for community sizes $i$ {that are neither very close to $0$ nor very close to the number $n$ of nodes in a network}. In particular, Theorem \ref{ArashTheorem} has the following corollary.

\begin{corollary}\label{ArashCorollary}
Let $C_{n, p} = \left[\sum_{j=0}^n (1 - p^{j + 1})(1 - p^{n - j + 1}) \right]^{-1}$, suppose that $\eta_n$ is a sequence with $\eta_n \to \infty$ and $\eta_n = o(n)$ as $n \to \infty$, and define $I_n = [\eta_n, n - \eta_n]$. As $n \to \infty$, we then have
\begin{equation*}
	\sup_{i \, \in \, I_n} \left| \frac{\pi(i)}{C_{n, p}} - 1\right| \to 0 \,.
\end{equation*}
\end{corollary}
The delocalization that is implied by Theorem \ref{ArashTheorem} and Corollary \ref{ArashCorollary} 
 arises from geometric retention.  We do not sample a weak composition, as movers from community $r$ must go to the only other community because $k = 2$.
 Moreover, we expect that our LECS prior with fixed retention probability $p$ is essentially the most localized case, as we expect (by the law of total variance~{\cite[Theorem 9.5.5]{Blitzstein19}}) that placing a uniform prior on $p_{r,\ell}$ will
 decrease localization. Therefore, we expect our LECS prior with random $p_{r,\ell}$ to be even less localized.

\begin{proof}[Proof of Corollary~\ref{ArashCorollary}]
{By definition of $C_{n,p}$, equation \eqref{stationary_dist:proof} implies that the stationary distribution} $\pi(i)$ of the Markov chain (\ref{V_dist}, \ref{c_markov}) on {$n_1^{(\ell)}$}
satisfies 
\begin{equation}\label{piExpression1}
	\pi(i) = C_{n, p} (1 - p^{i + 1})(1 - p^{n - i + 1}) \, .
\end{equation}
Therefore,
\begin{align*}
	\left| \frac{\pi(i)}{C_{n, p}} - 1 \right| &= \left| -p^{i + 1} -p^{n - i + 1} + p^{n + 2} \right| \\
								&\le p^{i + 1} + p^{n - i + 1} + p^{n + 2} \,.
\end{align*}
Because $j \mapsto p^{j}$ is a decreasing function of $j$
for $p \in (0,1)$, for $i \in I_n$, we have
\begin{equation*}
	\left| \frac{\pi(i)}{C_{n, p}} - 1 \right| \le p^{\eta_n + 1} + p^{n - (n - \eta_n) + 1} + p^{n + 2} \,.
\end{equation*}
That is,
\begin{equation*}
	\sup_{i \in I_n} \left| \frac{\pi(i)}{C_{n, p}} - 1 \right| \le  2 p^{\eta_n + 1} + p^{n + 2} = o(1) \,,
\end{equation*}
which completes the proof.
\end{proof}


\subsection{Summary}

Our {numerical computations and theoretical results} confirm our claims that the choice of generative model of community assignments has a significant effect on the localization of community-size distributions. We observed that a uniform distribution on community assignments yields substantially more localized community-size distributions than the other examined methods for temporal community structure. We also observed that the single-layer community-size distributions of the discrete-time Markov-process approaches become more localized over time, whereas the localization of the single-layer community-size distributions increases much more slowly with time for our LECS-prior-based approach. Accordingly, in temporal networks, we conclude that our LECS-prior-based approach has less localized overall community-size distributions than discrete-time Markov-process models.


{
\section{Posterior Sampling for Temporal SBMs with LECS, Bazzi et al., and Uniform Priors}\label{sec:algo}

In this section, we describe the algorithms that we use to sample from the posterior distributions $\P(g|A)$ for each choice of community-assignment probability distribution $\P(g)$. 
Due to the relatively complicated nature of the expressions for $\P(A|g)$ and $\P(g)$, directly sampling from the posterior distributions is not tractable. Therefore, we use Gibbs updates \cite{Park22} 
{for} individual node-layers $(i,\ell)$ and augment them with occasional multilayer label swaps to escape local extrema so that we approximately sample from the posterior distribution.

\medskip 

\paragraph{Setup of generative models}

All examined methods use the same $\P(A|g)$. We consider independent SBMs in each layer, so
\begin{equation}\label{dcSBM}
	\P(A|g,\omega) = \prod_{\ell = 1}^{L} \P(A^{(\ell)}|g_{(\ell)},\omega^{(\ell)}) = \prod_{\ell = 1}^L \prod_{1 \le i < j \le n} \left(\omega_{g_{(i,\ell)}g_{(j,\ell)}}^{(\ell)}\right)^{A_{ij}^{(\ell)}}\left(1 - \omega_{g_{(i,\ell)}g_{(j,\ell)}}^{(\ell)}\right)^{1 - A_{ij}^{(\ell)}} \,,
\end{equation}
where $\omega^{(\ell)} \in [0,1]^{k \times k}$ (for each $\ell \in [L]$) is a matrix whose entries $\omega^{(\ell)}_{rs}$ control the probabilities of edges between nodes in communities $r$ and $s$ in layer $\ell$.
We set the prior distributions for each $\omega^{(\ell)}_{rs}$ to be independent uniform distributions on $[0,1]$ for each $r$ and $s$ such that $r,s \in [k]$.
We set $\omega^{(\ell)}_{rs} = \omega^{(\ell)}_{sr}$ if $r > s$ (i.e., we force $\omega^{(\ell)}$ to be symmetric). 
Integrating\footnote{See \cite{Polanco23} for the computation of similar integrals.} $\P(A|g,\omega)$ with respect to the probability measure that is induced by $\P(\omega)$ yields the posterior distribution
\begin{equation}\label{dcSBMmarg}
	\P(A|g) = \prod_{\ell = 1}^L \prod_{1 \le r \le s \le k} \frac{m^{(\ell)}_{rs}!(t^{(\ell)}_{rs}-m^{(\ell)}_{rs})!}{(t^{(\ell)}_{rs}+1)!} \,,
\end{equation}
where 
\begin{align*}
	t^{(\ell)}_{rs} &:= \sum_{1 \le i, j \le n;\, i \ne j} \mathbb{1}\{g_{(i,\ell)} =  r \, , \, g_{(j,\ell)} = s\} \,, \\
	m^{(\ell)}_{rs} &:= \sum_{1 \le i, j \le n;\, i \ne j} A_{ij}^{(\ell)}\mathbb{1}\{g_{(i,\ell)} =  r \, , \, g_{(j,\ell)} = s\}
\end{align*}	
for $r \ne s$ and
\begin{align*}
	t^{(\ell)}_{rr} &:= \frac{1}{2}\sum_{1 \le i, j \le n; \, i \ne j} \mathbb{1}\{g_{(i,\ell)} = r \, , \, g_{(j,\ell)} = r\} \,, \\
	m^{(\ell)}_{rr} &:=\frac{1}{2} \sum_{1 \le i, j \le n; \, i \ne j} A_{ij}^{(\ell)}\mathbb{1}\{g_{(i,\ell)} = r \, , \, g_{(j,\ell)} = r\} \,.
\end{align*}
That is, for $r \ne s$, the quantity $t^{(\ell)}_{rs}$ is the number of pairs of distinct node-layers $(i,\ell)$ and $(j,\ell)$ for which $(i,\ell)$ is in community $r$ and $(j,\ell)$ is in community $s$. Similarly, $m^{(\ell)}_{rs}$ is the number of such pairs that are connected directly by an edge. The expressions for $t^{(\ell)}_{rr}$ and $m^{(\ell)}_{rr}$ have an additional factor of $1/2$ to avoid double-counting.

We start by giving our general algorithm (see Algorithm \ref{algo}) to sample from the posterior distribution $\P(g\mid A)$ of temporal SBMs.

\begin{altalgorithm}[One step of our algorithm to sample from the posterior distribution $\P(g\mid A)$]
\begin{enumerate}
\item Consider each node-layer $(i,\ell)$ in the order $g_{(1,1)}, g_{(1,2)}, \ldots, g_{(1,L)}, g_{(2,1)}, g_{(2,2)}, \ldots, g_{(n,L)}$:
\begin{enumerate}
\item With probability $1-p$, perform a Gibbs-sampling step:
\begin{enumerate}
\item Sample $g_{(i,\ell)}$ from $\P(g_{(i,\ell)} | \tilde{g} ,A)$, where $\tilde{g}$ is the set of current community assignments other than $g_{(i,\ell)}$. (See Section \ref{GibbsSamplingStep} for details.)
\end{enumerate}
\item With probability $p$, perform a multilayer swap:
\begin{enumerate}
\item Sample $r, s \in [k]$ and $\ell \in [L]$ independently from uniform distributions. 
\item Propose new group assignment $g^\star$ by swapping labels $r$ and $s$ in all layers $\ell\ge \ell_0$:
\begin{equation*}
	g^\star_{(i,\ell)} = \begin{cases}
		s\,,& g_{(i,\ell)} = r\,,\ \ell\ge \ell_0 \\
		r\,,& g_{(i,\ell)} = s\,,\ \ell\ge \ell_0 \\
		g_{(i,\ell)}\,,& \text{otherwise}\,.
	\end{cases}
\end{equation*}
\item Accept the new group assignment $g^*$ with probability
\begin{equation}\label{MHcommunity}
	\min\left\{1,\frac{\P(g^*|A)}{\P(g|A)} \right\} \, .
\end{equation}	
(See Section \ref{MultiNodeMoveStep} for details.)
\end{enumerate}
\end{enumerate}
\end{enumerate}
\label{algo}
\end{altalgorithm}

In Section \ref{MainAlgDiscussion}, we discuss various aspects of Algorithm \ref{algo}.
In Section \ref{YangBaseline}, we discuss the Yang et al.\ \cite{Yang11} method, which we use as a baseline for our comparisons.


\subsection{Discussion of Algorithm \ref{algo}}\label{MainAlgDiscussion}

In this subsection, we discuss various aspects of Algorithm \ref{algo}. In Section \ref{PgComputation}, we explain how we compute $\P(g)$. In Section \ref{GibbsSamplingStep}, we describe the process 
to sample
from the posterior distribution $\P(g_{(i,\ell)} | \tilde{g},A)$. In Section \ref{MultiNodeMoveStep}, we explain the process to compute the acceptance probability (\ref{MHcommunity}) of a multilayer swap. In Section \ref{MLSwapMovesMotivation}, we discuss the motivation behind our use of multilayer swaps. Finally, in Section \ref{MixedGibbsMH}, we statistically justify our simultaneous use of Gibbs-sampling steps and multilayer swaps.


\subsubsection{Computing $\P(g)$} \label{PgComputation}

In Appendix \ref{BazziCF}, we derive the expression
\begin{equation} \label{this}
	\P(g) = \P(g_{(1)}) \prod_{\ell = 2}^L \P(g_{(\ell)}|g_{(\ell - 1)}) \,.
\end{equation}	
Varying the community assignment of node-layer $(i,\ell)$ with all other community assignments fixed affects only the terms  $\P(g_{(\ell)}|g_{(\ell - 1)})$ and $ \P(g_{(\ell + 1)}|g_{(\ell )})$. 
To simplify the expression for $\P(g)$ for each choice of prior distribution, we perform the following procedures:
\begin{itemize}
  \item \textbf{Uniform prior.} The quantity $\log \P(g)$ is constant and cancels in all ratios, so we do not need to consider this term. (In the present paper, $\log(\cdot)$ denotes the complex base-$e$ logarithm with a branch cut along the negative real axis.)
  \item \textbf{Bazzi et al. approach.} We use the integral form 
\begin{align}\label{BazziIntegral}
	P(g_{(\ell)}\mid g_{(\ell - 1)}) \;=\; \int_{[0,1] \times \Delta^{k-1}} \prod_{i = 1}^n \Big( \alpha_\ell \,\mathbf{1}\{ g_{(i,\ell)} = g_{(i,\ell-1)} \} + (1-\alpha_\ell)\, \kappa^{(\ell)}_{\,g_{(i,\ell)}} \Big)\, d\mu(\alpha_\ell,\kappa^{(\ell)})\,,
\end{align}
where $\mu$ is the product of the uniform measure on $[0,1]$ and the uniform measure on $\Delta^{k - 1}$. 
The expression \eqref{BazziIntegral}, which we derive in Appendix~\ref{BazziCF}, does not have a
closed form, so we approximate it using
Monte Carlo integration with 1000 draws of $(\alpha_\ell,\kappa^{(\ell)})\sim \mathrm{Unif}[0,1]\times \mathrm{Unif}(\Delta^{k - 1})$.   
\item \textbf{LECS prior.}  In Appendix~\ref{novelCF}, we derive the closed-form expression
\begin{equation}\label{novelMarginal}
	\P(g_{(\ell)}\mid g_{(\ell - 1)}) \;=\; \prod_{r = 1}^k \left[\frac{ J\!\left(n_r^{(\ell-1)} - c^{(\ell)}_{rr},\, n_r^{(\ell - 1)}\right)}{\binom{\,n_r^{(\ell - 1)} - c^{(\ell)}_{rr} + k - 2\,}{\,n_r^{(\ell - 1)} - c^{(\ell)}_{rr}\,}\; \binom{\,n_r^{(\ell - 1)}\,}{\,c^{(\ell)}_{r,1},\ldots,c^{(\ell)}_{r,k}\,}} \right] \,,
\end{equation}
with
\begin{equation*}
	J(k_1,k_2) \;=\; \int_0^1 x^{k_1}\,\frac{x-1}{x^{k_2 + 1} - 1}\,dx = \frac{1}{k_2 + 1} \left[ \psi\left(\frac{k_1 + 2}{k_2 + 1}\right) - \psi\left(\frac{k_1 + 1}{k_2 + 1}\right) \right] \,,
\end{equation*}
 where $\psi(z) = \frac{d}{dz} \! \log \Gamma(z)$ is the digamma function \cite{dlmf} and $c^{(\ell)}_{rs}$ are the transition counts. (Recall from Section~\ref{novelML} that $c^{(\ell)}_{rs}$ is the number of nodes in community $r$ in layer $\ell - 1$ that {transition to} community $s$ in layer $\ell$.) 
For efficiency, we precompute $J(k_1,k_2)$ for all integers $k_1$ and $k_2$ such that $0\le k_1\le k_2\le n$. For very small values of $J(k_1,k_2)$ (i.e., for large $k_1$ and $k_2$), we stabilize 
our numerical computations by treating $J(k_1 + 1,k_2)/J(k_1,k_2) \approx 1$ and setting $J(k_1,k_2) = \exp(-16)$ in any computations that yield a value that is smaller than $\exp(-16)$.
\end{itemize}


\subsubsection{Sampling from $\P(g_{(i,\ell)}|\tilde{g},A)$}\label{GibbsSamplingStep}

To sample from the distribution $\P(g_{(i,\ell)} | \tilde{g} ,A)$, we compute
\begin{equation}\label{condProbSingle}
	\P(g_{(i,\ell)} = r | \tilde{g} ,A) = \frac{\P(A, g_{(i,\ell)} = r, \tilde{g})}{ \sum_{s = 1}^k \P(A, g_{(i,\ell)} = s, \tilde{g})} \, ,
\end{equation}
where $\P(A, g_{(i,\ell)} = r, \tilde{g})$ is $\P(A,g)$ evaluated at
the specified adjacency structure $A$ 
for a community assignment $g$ that includes all current community assignments other than
$g_{(i,\ell)}$, which we take to be $r$. Because $\P(A,g) = \P(A|g) \P(g)$, we use the approach in Section \ref{PgComputation} to compute $\P(g)$ and equation \eqref{dcSBMmarg} to compute $\P(A|g)$.


\subsubsection{Computation of the Acceptance Probability of a Multilayer Swap} \label{MultiNodeMoveStep}

Recall from \eqref{MHcommunity} that the acceptance probability of a multilayer swap is
\begin{equation*}
	\min\left\{1,\frac{\P(g^*|A)}{\P(g|A)} \right\} \, .
\end{equation*}
Because $\P(A,g) = \P(A|g) \P(g)$ and $\P(g|A) = \frac{\P(A|g) \P(g)}{\P(A)}$, the acceptance probability \eqref{MHcommunity} is equivalent to
\begin{equation*}
	\min\left\{1,\frac{\P(A, g^*)}{\P(A,g)} \right\}\, .
\end{equation*}	
As we discussed in Section \ref{GibbsSamplingStep}, because $\P(A,g) = \P(A|g) \P(g)$, we can then use the approach in Section \ref{PgComputation} to compute $\P(g)$ and equation \eqref{dcSBMmarg} to compute $\P(A|g)$.


\subsubsection{Motivation for Multilayer Swaps}\label{MLSwapMovesMotivation}

In this subsubsection, we discuss our motivation for including multilayer swaps 
in our posterior sampling algorithm (see Algorithm \ref{algo}). 

If we apply a naive Gibbs-sampling approach (i.e., one that does not use multilayer swaps), we often obtain incorrect
community identifications in temporal networks. 
A key reason for this is that the SBM likelihood \eqref{dcSBMmarg} is invariant to intra-layer permutations of community labels, whereas the temporal priors (the choices of $\P(g)$ for the Bazzi et al. and LECS approaches) break this invariance across layers. 

Phrased mathematically, the temporal priors in the Bazzi et al. and LECS are label-equivariant, whereas the SBM likelihood \eqref{dcSBMmarg} satisfies a single-argument invariance. Suppose that the finite symmetric group $\mathfrak{S}_k$ (i.e., the group of all permutations) acts on labelings by $(\varpi\!\cdot\! g)_{(i,\ell)} := \varpi\bigl(g_{(i,\ell)}\bigr)$ for any permutation $\varpi\in\mathfrak{S}_k$.
A conditional is
\emph{label-equivariant} if
\begin{equation}
	\P\bigl(\varpi\!\cdot\! g_{(\ell)} \,\big|\, \varpi\!\cdot\! g_{(\ell - 1)}\bigr) \;=\; \P\bigl(g_{(\ell)} \,\big|\, g_{(\ell - 1)}\bigr)\quad\text{for all }\varpi\in\mathfrak{S}_k \,,
\end{equation}
{where $\varpi\!\cdot\! g_{(\ell)} := ((\varpi\!\cdot\! g)_{(1,\ell)}, \ldots, (\varpi\!\cdot\! g)_{(n,\ell)})$ is the vector of community assignments in $\varpi\!\cdot\! g$ for each node-layer in layer $\ell$.}
In this equivariance, one permutes both arguments with the same permutation $\varpi$. This property thus differs from single-argument invariance.

In our temporal-network setting, combining an invariant likelihood with an equivarient prior
yields many local extrema in the landscape of the likelihood $\P(g|A)$ that differ only in the numerical labels of communities across layers
(see Figure \ref{localextremacommunity}). 
Naive Gibbs sampling often cannot cross the barriers between such label-misaligned extrema because the label-invariant SBM likelihood can dominate differences in the distributions $\P(g_{(\ell)} | g_{(\ell - 1)})$. 

A multilayer
swap proposes a relabeling of communities $r$ and $s$ across all layers $\ell\ge \ell_0$, and such a relabeling helps align community assignments across layers. Because the SBM likelihood \eqref{dcSBMmarg} is label-invariant within each layer, the Metropolis--Hastings (MH) ratio 
\begin{equation}
	\min\left\{1,\frac{\P(g^*|A)}{\P(g|A)} \right\} 
\end{equation}
(which we repeat from \eqref{MHcommunity}) reduces to a ratio of terms of the form $\P(g_{(\ell)} | g_{(\ell - 1)})$. 
Such a ratio is sensitive to misalignment and thus helps detect and correct it. Moreover, one only needs to evaluate prior factors that straddle the swap boundary (i.e., $\P(g_{(\ell_0)} | g_{(\ell_0 - 1)})$ and $\P(g_{(\ell_0 + 1)} | g_{(\ell_0)})$),
  so it is efficient to compute the acceptance probabilities.}

\begin{figure}[t]	\centering
	\subfloat[Seeded community structure]{\includegraphics[width=0.49\textwidth]{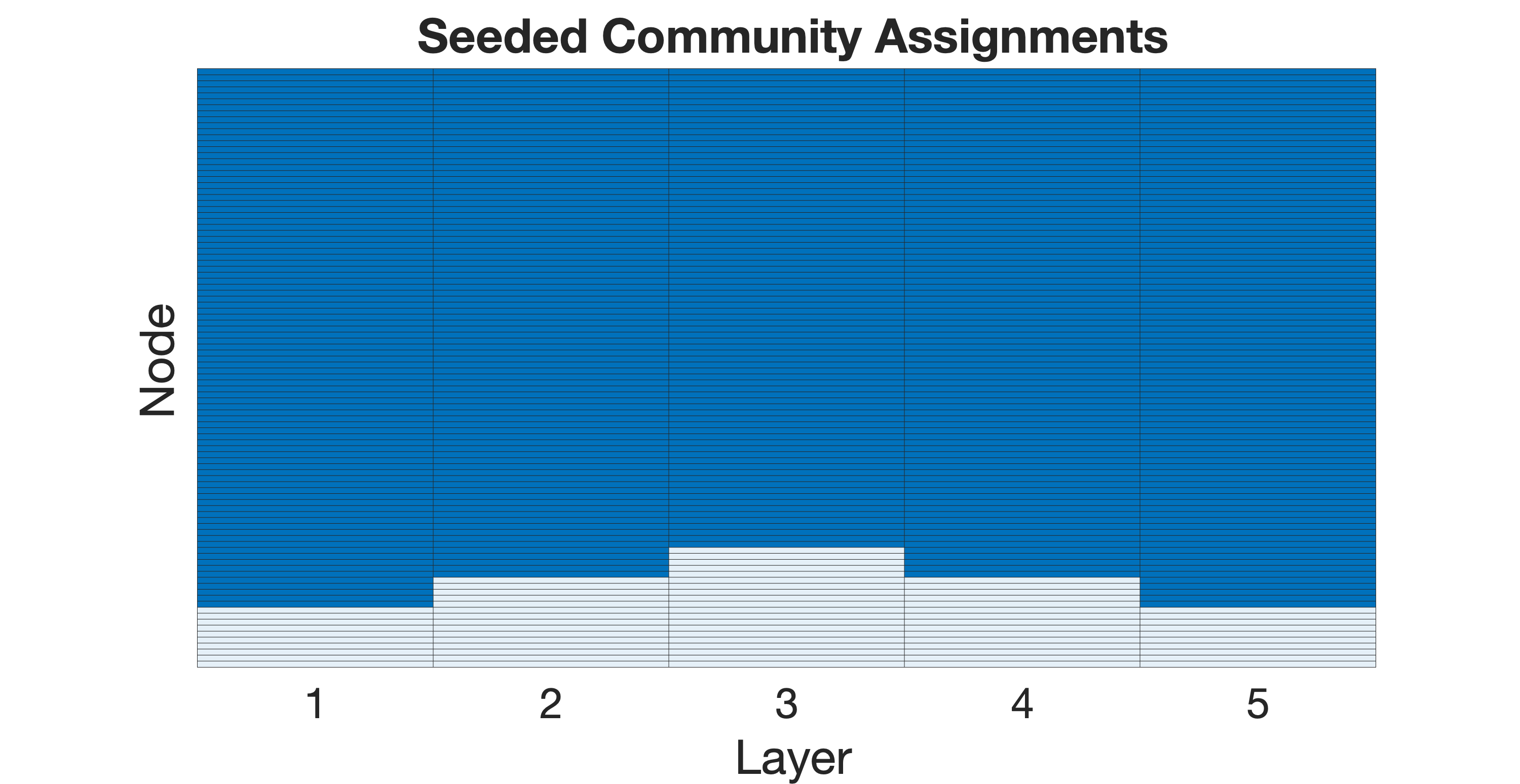}}
	\subfloat[Example of a local extremum]{\includegraphics[width=0.49\textwidth]{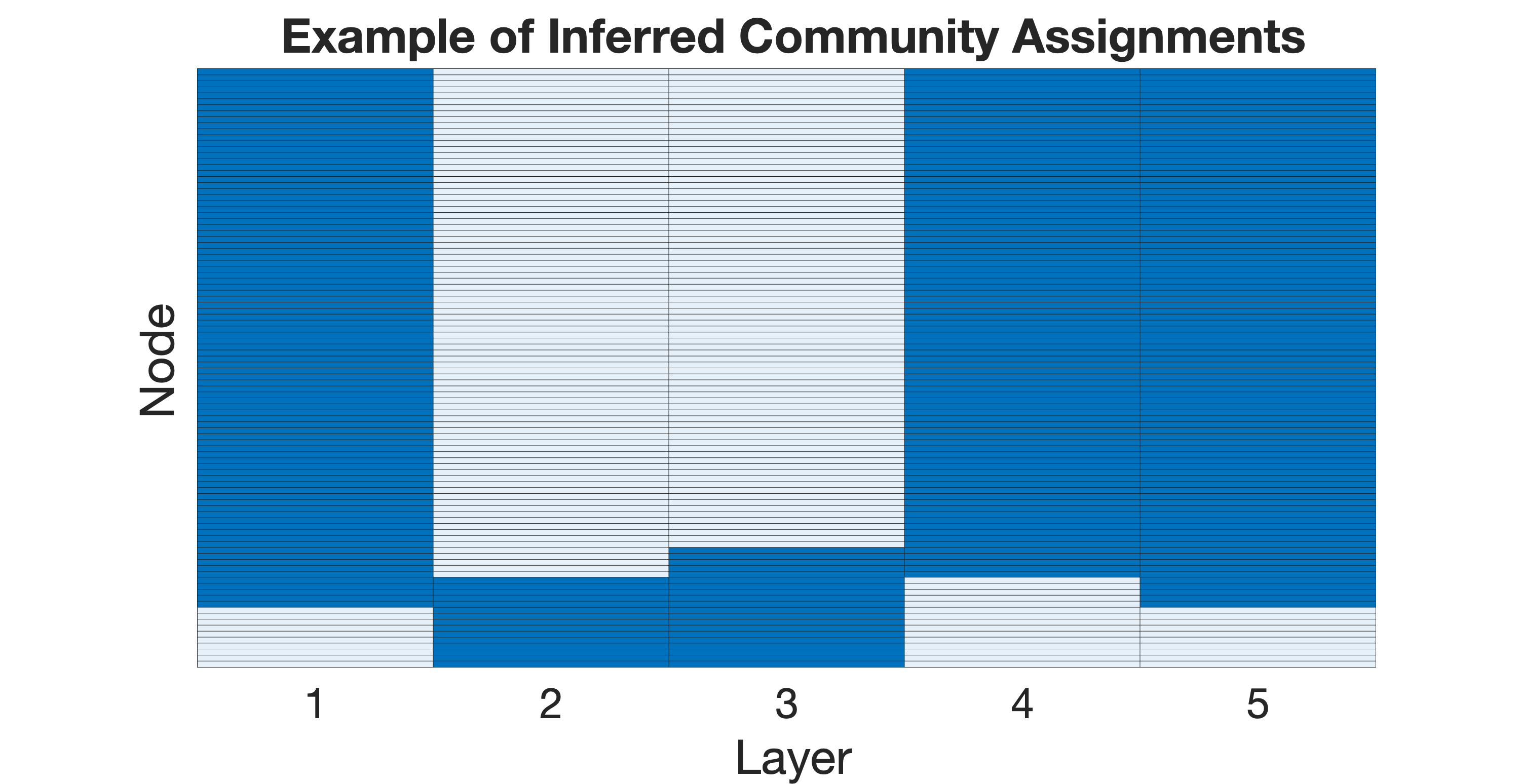}}
	\caption[Heat maps of an example of actual community structure and an illustration of the permuted community structure commonly observed in local maxima of $\P(g|A)$ in a 100-node network with 5 layers.]
	{Heat maps of (a) an example of actual community structure and (b) an illustration of the permuted community structure that we observe commonly at local maxima of $\P(g|A)$ in a 100-node network with 5 layers. Each rectangle in a heat map corresponds to one node-layer $(i,\ell)$. Dark blue rectangles signify the community assignment $g_{(i,\ell)} = 1$, and light blue rectangles signify the community assignment $g_{(i,\ell)} = 2$.}
	\label{localextremacommunity}
	\end{figure}

	
\subsubsection{Correctness of Mixed Gibbs and Metropolis--Hastings Updates}\label{MixedGibbsMH}

As we discussed in Section \ref{sec:algo}, in our posterior-sampling algorithm (see Algorithm \ref{algo}), we interleave exact Gibbs updates for individual node-layers with multilayer MH swap proposals. Each Gibbs step samples from a full conditional distribution $\P(g_{(i,\ell)} | \tilde{g} ,A)$ (where $\tilde{g}$ is the set of current community assignments other than
$g_{(i,\ell)}$) and is thus reversible with respect to $\P(g\mid A)$. Each multilayer swap uses an MH ratio \eqref{MHcommunity} targeting $\P(g\mid A)$ and is thus also reversible. 
The composition of such {Markov chains} preserves the same invariant distribution, so the overall sampler (see Algorithm \ref{algo}) has the desired stationary distribution $\P(g\mid A)$.


\subsection{Yang et al.~\cite{Yang11} Baseline} \label{YangBaseline}

To compare with a different likelihood $\P(A|g)$, prior $\P(g)$, and sampler, we also implement the method of Yang et al.~\cite{Yang11}. 
We use their approach as a baseline to help separate the effects of the temporal prior $\P(g)$ from those of the posterior sampling approach.

In this subsection, we briefly describe Yang et al.'s approach.
Yang et al.'s likelihood shares the SBM form \eqref{dcSBMmarg} but uses the same
SBM parameters $\omega_{rs} := \omega_{rs}^{(1)} = \cdots = \omega_{rs}^{(L)}$
 for all layers and employs the priors
 \begin{equation*}
	\omega_{rs} \sim \mathrm{Beta}(\alpha_{rs},\beta_{rs})
\end{equation*}	
for all $r \ge s$, where $\alpha_{rs}$ and $\beta_{rs}$ are user-specified parameters. 
In our implementation of Yang et al.'s approach, we let $\alpha_{rs} = \beta_{rs} = 1$ for all $r \ge s$ and thereby obtain
\begin{equation*}
	\omega_{rs} \sim \text{Uniform}(0,1)
\end{equation*}	
for all $r \ge s$. 
To sample from the posterior distribution $\P(g|A)$, Yang et al.\ \cite{Yang11} used Gibbs sampling with simulated annealing. 
We do so as well, and we use the same simulated-annealing setup as them. Accordingly, for each step $m \in \{1,\ldots,10\}$, we sample from
$\exp\{\log \P(g\mid A)/T_m\}$} using
$I_m$ iterations of Gibbs sampling. 
For each step $m$, we show the corresponding temperature $T_m$ and iteration count $I_m$ 
in Table \ref{YangGibbsSetup}.

\begin{table}[t]
\begin{center}
\begin{tabular}{c|c|c|c|c|c|c|c|c|c|c}
$m$ & 1 & 2 & 3 & 4 & 5 & 6 & 7 & 8 & 9 & 10 \\ \hline
$T_m$ & 1 & 0.9 & 0.8 & 0.7 & 0.6 & 0.5 & 0.4 & 0.3 & 0.2 & 0.1 \\ \hline
$I_m$ & 20 & 10 & 10 & 10 & 10 & 10 & 10 & 5 & 5 & 5
\end{tabular}
\end{center}
\caption{The temperature $T_m$ and iteration count $I_m$ for each step $m$ in the Yang et al.\ \cite{Yang11} simulated-annealing sampling process.} \label{YangGibbsSetup}
\end{table}


\section{Comparison of Community-Detection Performance of the Different Approaches} \label{ComparisonSI}

As we mentioned in Section \ref{intro}, we expect that statistical-inference methods that employ generative models with more-localized community-size distributions will have poorer performance in networks with large or small communities than methods with less-localized community-size distributions. In this section, we demonstrate that this is indeed the case.


\subsection{Setup of our Comparison}

To compare the behavior of the four approaches on networks with different community sizes, we generate several networks with known community structure and different community sizes. We begin by considering $k = 2$ communities in temporal networks with $n = 100$ nodes and $L=5$ layers.
We then choose the parameters $\omega \in [0,1]^{k \times k}$
(which encode the strength of the community structure) and $q \in \{0,\ldots,100\}$ (which is the size of community $1$). We use the parameter $q$ to assign the seeded community structure $g$ as follows.
For each layer $\ell \in [L]$, we set $g_{(i,\ell)} = 1$ for nodes $i \in \{1,\ldots, q + \tau_\ell\}$ and $g_{(i,\ell)} = 2$ for all other nodes. The quantity $\tau = (\tau_1,\ldots,\tau_L)$ is a vector of ``offsets'' for each layer that we use to avoid having the same seeded community structure for each layer. In Figure \ref{seededstructs}, we show examples of the seeded community structure for $q = 50$ and $q = 90$, with $\tau = (0,-5,-10,-5,0)$ in each case. We use $g$ to generate the networks $A^{(\ell)}$ via independent SBMs with parameters
\begin{equation*}
	\omega^{(\ell)} = \begin{pmatrix} 0.25 & 0.1 \\ 0.1 & 0.25 \end{pmatrix}
\end{equation*}	
for each layer $\ell \in [L]$.

\begin{figure}[t]\centering
\subfloat[$q=50$]{\includegraphics[width=0.49\textwidth]{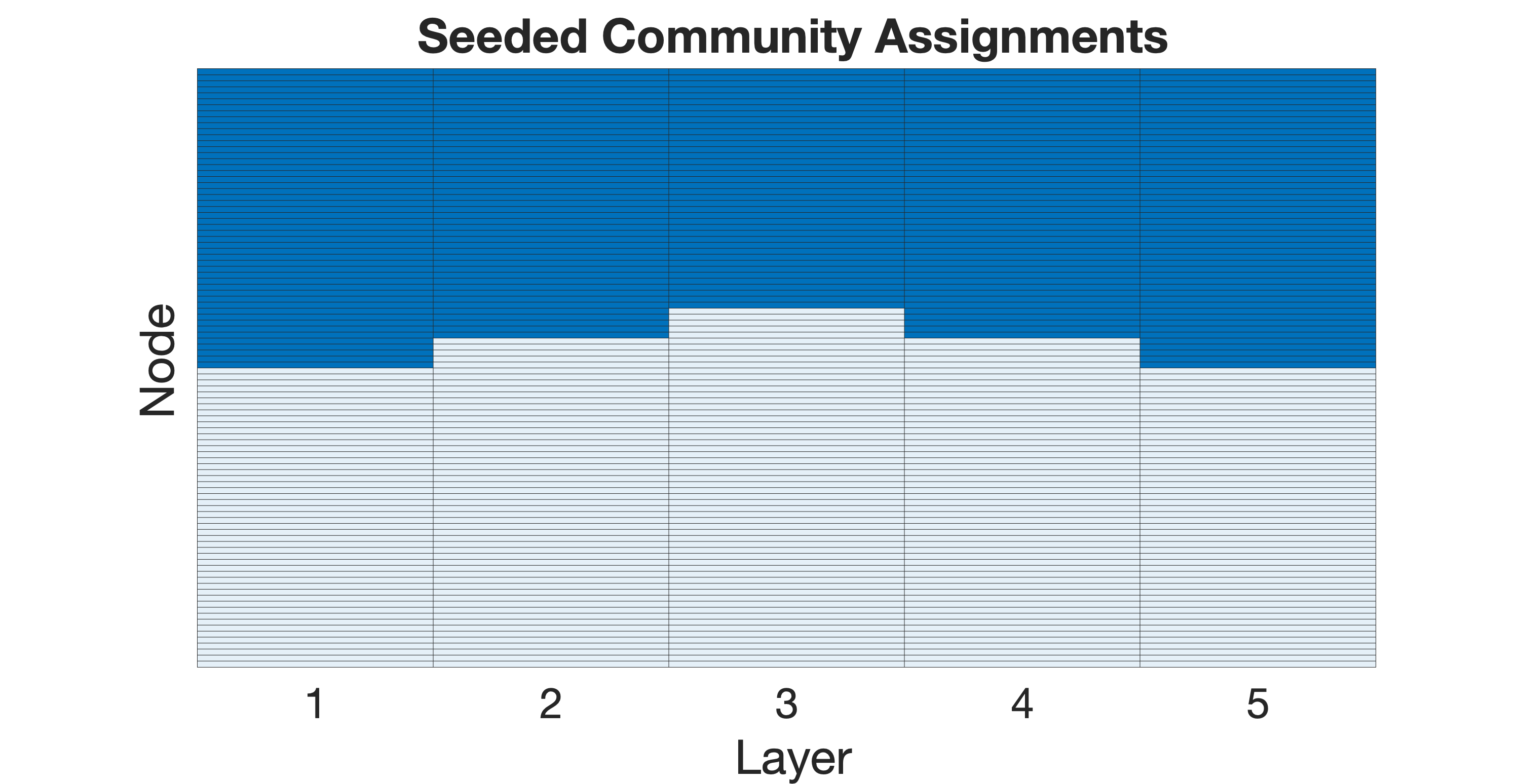}}
\subfloat[$q=90$]{\includegraphics[width=0.49\textwidth]{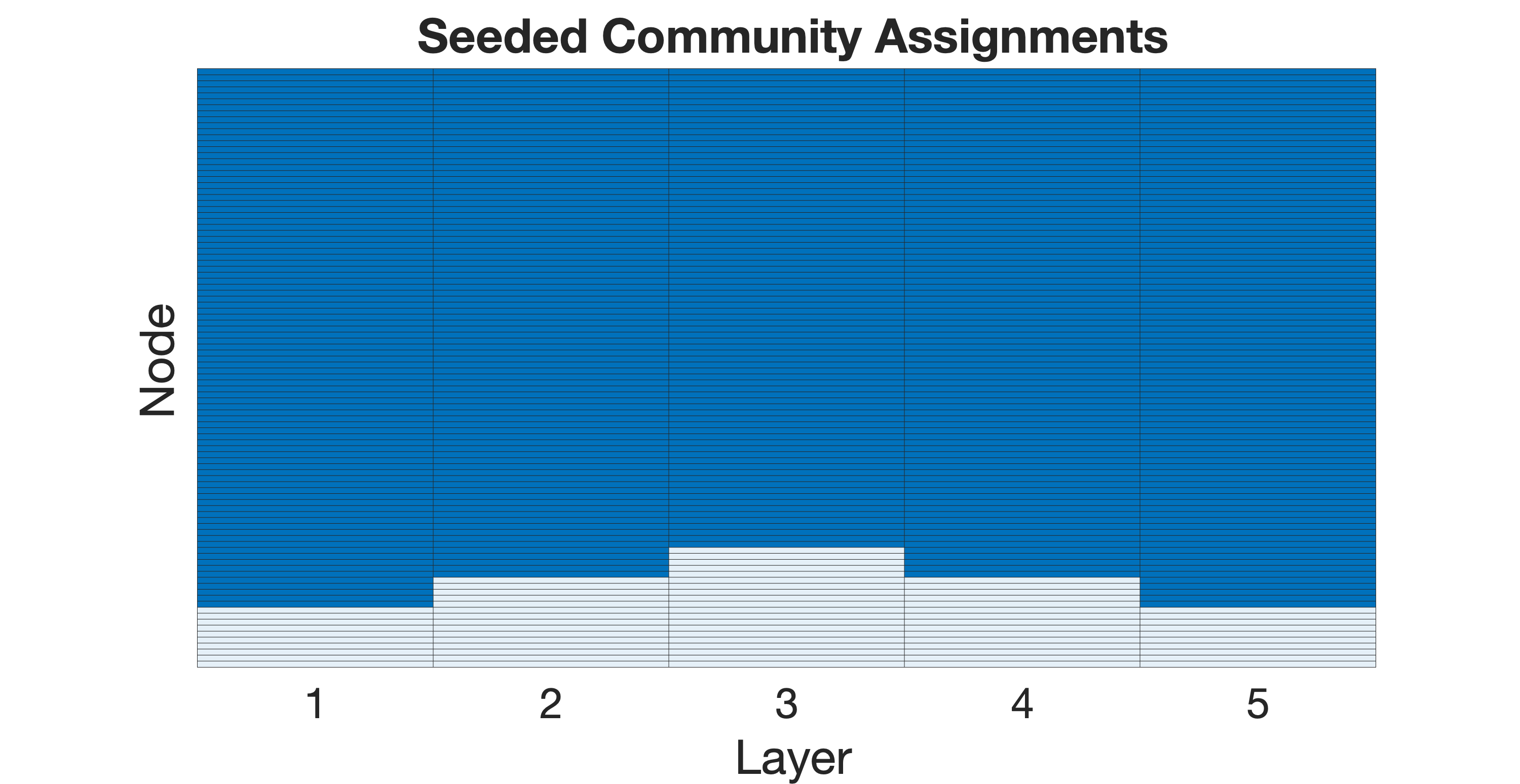}}
\caption[Heat maps of the seeded community structure that we use to generate the adjacency structure $A$ for community-1 size $q = 50$ and community-1 size $q = 90$.]{Heat maps of the seeded community structure that we use to generate the adjacency structure $A$ for (a) community-1 size $q = 50$ and (b) community-1 size $q = 90$. Each rectangle in a heat map corresponds to one node-layer $(i,\ell)$. Dark blue rectangles signify the community assignment $g_{(i,\ell)} = 1$, and light blue rectangles signify the community assignment $g_{(i,\ell)} = 2$.
 }
\label{seededstructs}
\end{figure}

In our experiments, we consider $q \in \{50,60,70,80,90\}$. For each choice of $q$, we run 500 instantiations of each approach and record the inferred community assignments $g'$ of each instantiation of each approach. To quantitatively measure the performance of each approach, we compute normalized mutual information (NMI) $\text{NMI}(g';g)$ \cite{Danon05}, which is a commonly employed similarity measure for analyzing the performance of a classification, between the inferred community structure $g'$ and the seeded community structure $g$ for each instantiation and approach. 
The formula for the NMI is  
\begin{equation}
	\text{NMI}(g';g) = \frac{I_0(g';g)}{\frac{1}{2}[I_0(g;g) + I_0(g';g')]} \, ,
\end{equation}	
where 
\begin{align*}
	I_0(g';g) &= nL \sum_{1 \le r,s \le k} p_{rs}^{(gg')} \log \!\left(\frac{p_{rs}^{(gg')}}{p_r^{(g')}p_s^{(g)}}\right) \, , \\
	p_{rs}^{(gg')} &= \frac{1}{nL} \sum_{i =1}^n \sum_{\ell = 1}^L \mathbb{1}\{g'_{(i,\ell)} = r \, , \, g_{(i,\ell)} = s\} \,, \\
	p_{r}^{(g)} &= \frac{1}{nL} \sum_{i =1}^n \sum_{\ell = 1}^L \mathbb{1}\{g_{(i,\ell)} = r\} \, .
\end{align*}	
The NMI quantifies the similarity between the inferred community structure $g'$ and the seeded community structure $g$. The maximum NMI value of $1$ occurs when $g'$ and $g$ coincide up to a permutation of the community labels.


\subsection{Results}\label{MLResults}

We begin by comparing the performances of the Bazzi et al.~\cite{Bazzi20} approach and our LECS-prior-based approach with and without multilayer swaps (see Section \ref{sec:algo}). We set the probability of a multilayer swap to $3 \times 10^{-3}$. In Figure \ref{NMIfull}, we plot the mean of $\text{NMI}(g';g)$ for each of the posterior samples $g'$ for community-1 sizes $q \in \{50,60,70,80,90\}$.

We obtain larger mean NMI values for both the Bazzi et al.\ approach and our LECS-prior-based approach when we use multilayer swaps, indicating that they are less likely to become stuck at local extrema (of the type in Figure \ref{localextremacommunity}) than when we do not include multilayer swaps. We use hypothesis testing to make this observation statistically rigorous. The distributions of the NMI values are not normal distributions, so we compute the $p$-values for a one-sided Mann--Whitney U test \cite{MannWhitney08} for each layer $\ell$, community-1 size $q$, and community-structure strength $\omega$.
For both approaches, we see that the $p$-values are significant for all values of $q$ except $q = 50$.
Therefore, we conclude that employing multilayer swaps improves the performance of both our LECS-prior-based approach and the Bazzi et al.\ approach.

\begin{figure}[t]\centering
\includegraphics[width=0.95\textwidth]{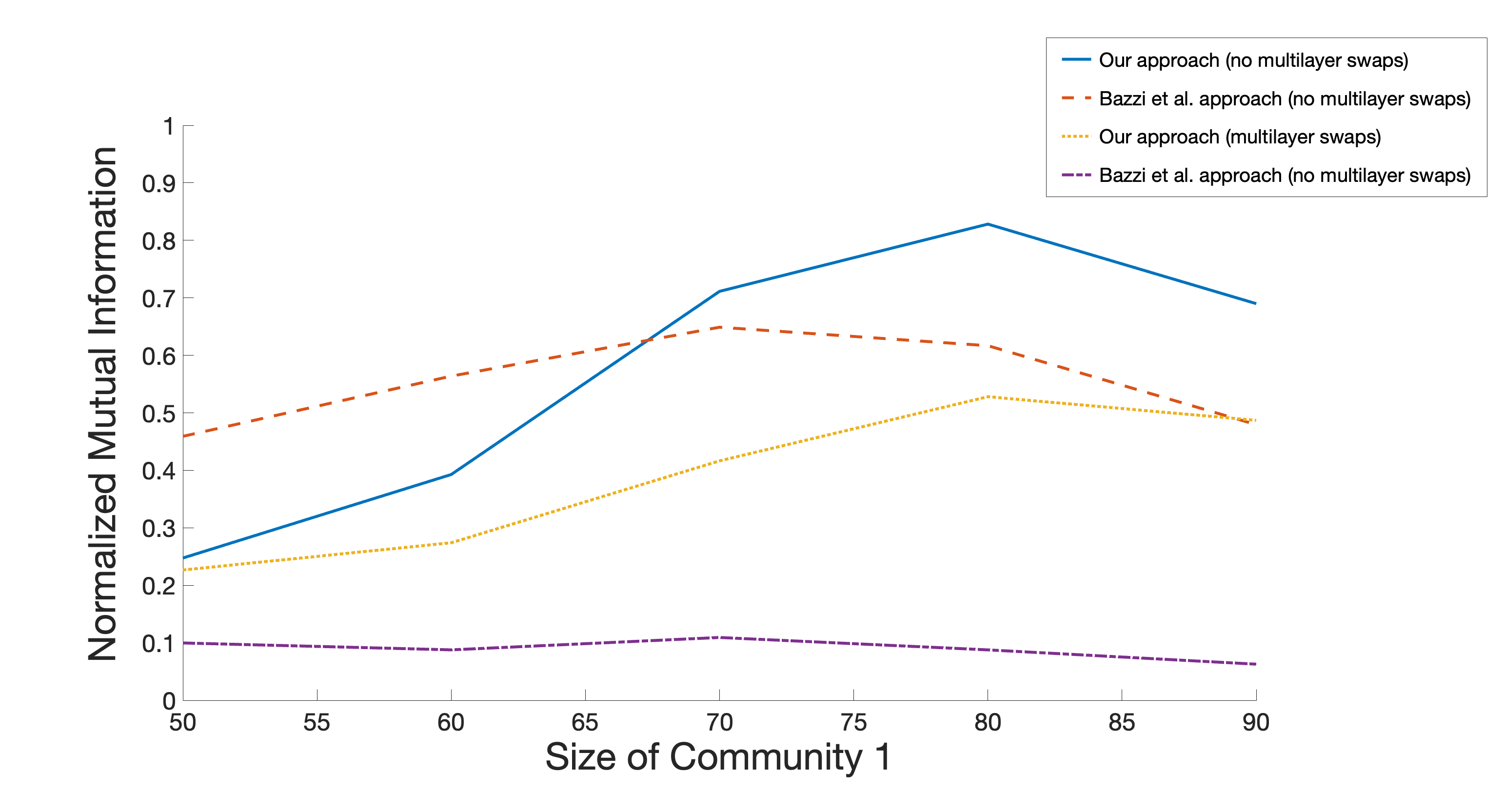}
\caption{The mean NMI for our LECS-prior-based approach and the Bazzi et al.\ approach with and without multilayer swaps for several values of the community-1 size $q$. 
}
\label{NMIfull}
\end{figure}

\begin{table}[t]\begin{center}
\begin{tabular}{| c | c | c | c | c | c |}
\hline
 & $q = 50$ & $q = 60$ & $q = 70$ & $q = 80$ & $q = 90$  \\ \hline
$\omega = 0.28$ & 0.56     &    0.017        &  0.45        & $9.1 \times 10^{-5} $    &   0.96\\ \hline
$\omega = 0.25$ &         0.12   &     0.063 &     $ 9.4 \times 10^{-6}$  &       $3.3 \times 10^{-6}$   &     $3.4 \times 10^{-7}$\\ \hline
 \end{tabular}
 \end{center}
\caption{The $p$-values for our comparison of our LECS-prior-based approach with and without multilayer swaps.
}

\label{pvaluesfullMultiNovel}
\end{table}

\begin{table}[t]\begin{center}
\begin{tabular}{| c | c | c | c | c | c |}
\hline
 & $q = 50$ & $q = 60$ & $q = 70$ & $q = 80$ & $q = 90$  \\ \hline
$\omega = 0.25$ & 0.015 & $ 6.5 \times 10^{-4}$  &  $3.3 \times 10^{-6}$  & $3.1 \times 10^{-16}$   & $7.8 \times 10^{-5}$ \\ \hline
$\omega = 0.28$ &      $2.4 \times 10^{-3}$ &  0.40 &   $8.9 \times 10^{-9}$   & $5.9 \times 10^{-15}$ & $ 2.5000 \times 10^{-6}$\\ \hline
 \end{tabular}
 \end{center}
\caption{The $p$-values for our comparison of the Bazzi et al. \cite{Bazzi20}  with and without multilayer swaps.}
\label{pvaluesfullMultiBazzi}
\end{table}

We also compare the performance of the four approaches --- a uniform distribution on community assignments, the Yang et al.\ \cite{Yang11} and Bazzi et al.\ \cite{Bazzi20} discrete-time Markov-process approaches, and our LECS-prior-based approaches --- when we incorporate multi-mode moves. We again consider community-1 sizes $q \in \{50,60,70,80,90\}$ and plot the overall-network NMI for each of the four approaches in Figure \ref{NMI4methodsfull}.

\begin{figure}[h]
\centering
\includegraphics[width=0.95\textwidth]{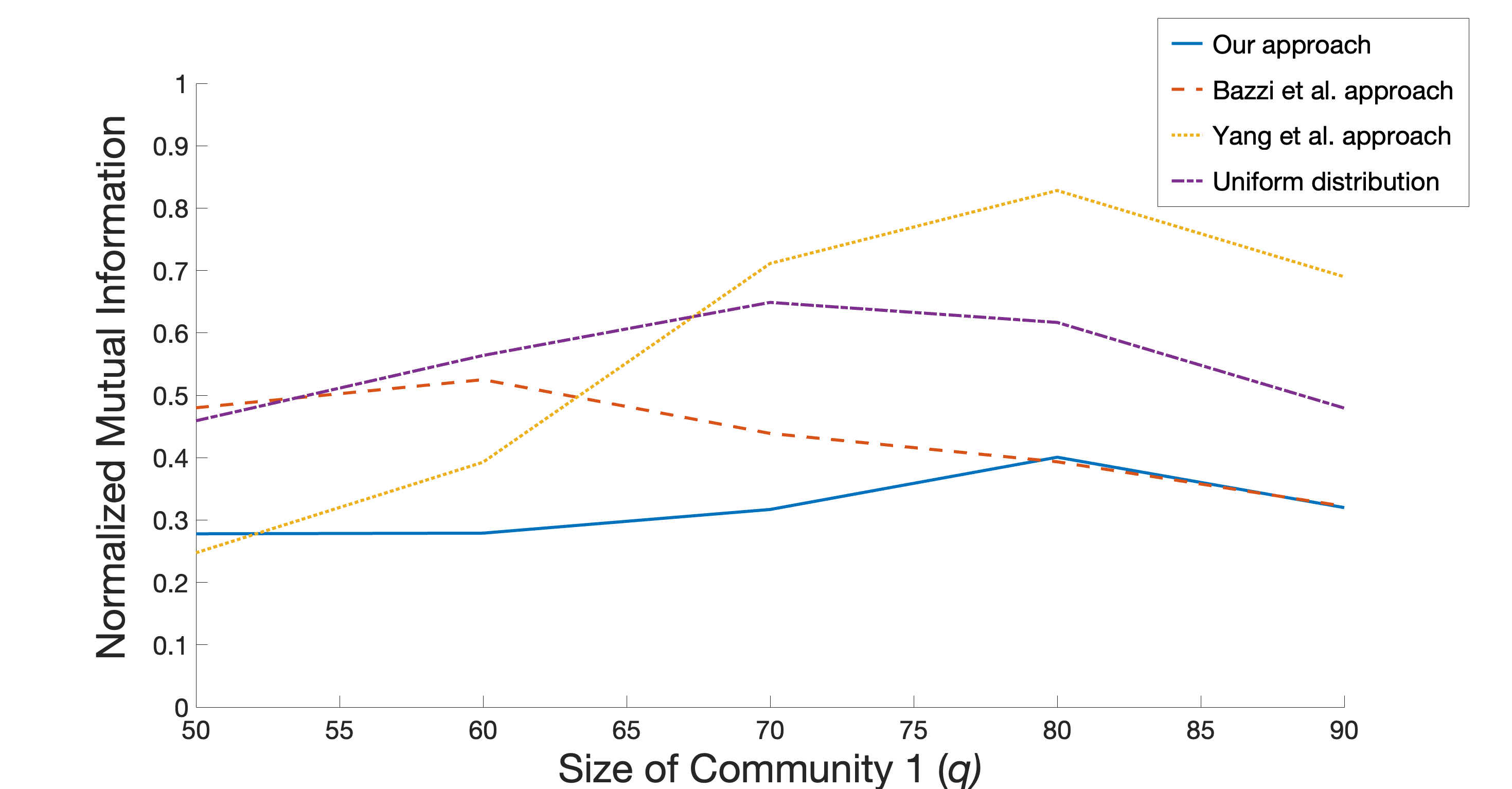}
\caption{The mean NMI for each of the four examined approaches for generating community assignments for various choices of community-1 sizes $q$.
}
\label{NMI4methodsfull}
\end{figure}

Our LECS-prior-based approach performs worse than the Bazzi et al.\ approach for \\ community-1 sizes of $q = 50$ and $q = 60$. However, our approach outperforms the Bazzi et al. approach for $q = 70$, $q = 80$, and $q = 90$. Our approach outperforms the Yang et al.\ approach for all values of $q$. Using Mann--Whitney U tests, we see that these results are significant (with $p$-value $p < 0.01$) for all but one value of $q$ (this value is $q=80$ for the Yang et al.\ method) when our LECS-prior-based approach performs better than the other methods.

In Figures \ref{NMIfull} and \ref{NMI4methodsfull}, the mean NMIs for each method are smaller for community-1 size $q = 90$ than for $q = 80$. We believe that this is an artifact of our choice to use NMI to measure the similarity between the inferred community structure and the seeded community structure.\footnote{For a detailed discussion of biases in NMI, see \cite{Jerdee24}.} When we use the number of correct community assignments to measure such similarity, we find that the mean numbers of correct community assignments for community-1 sizes $q=80$ and $q=90$ are almost identical for each method. 


\section{Conclusions and Discussion} \label{FinalConclusions}

{When using statistical inference to detect communities in networks, it is important to use a generative model that is based on realistic assumptions \cite{peixoto2023}. One important facet is to avoid (or at least mitigate) biases 
against communities with large or small numbers of nodes.  
In this paper, we illustrated empirically that 
statistical-inference models that generate community assignments via either a uniform distribution on community assignments or discrete-time Markov processes are biased against generating communities with large or small numbers of nodes. 
To mitigate this bias, we formulated a generative model for community assignments
with a layerwise-exchangeable count-splitting (LECS) prior, and we proved several
 results about the resulting community-size distributions.
We demonstrated in tests on synthetic networks with small and large communities that our generative model outperforms several existing generative models at statistical inference.
  We analyzed the effect of group-size biases in the setting of community structure, but analogous group-size assumptions also arise in statistical-inference methods for identifying other mesoscale structures. For example, Faust and Porter \cite{Faust25} showed that using discrete-time Markov-process models to detect core--periphery structure causes associated statistical-inference methods to underestimate the numbers of cohesive groups in networks.

There are many viable ways to build on our work. First, we considered the performance of a small number of representative approaches, and it is certainly worthwhile to also study other methods (e.g., see \cite{Ishiguro10}) with similar community-assignment probability distributions to those in Sections \ref{unifDistML} and \ref{NodewiseEvolML}. Second, in our LECS-prior-based approach,
  we assumed that the community assignments for a given layer depend only on those in the previous layer. Relaxing this assumption and allowing community assignments to depend on additional previous layers may lead to improved community-detection performance. Third, it is beneficial to derive bounds on the amount of localization of the community-size distributions both for our approach and for other approaches. In contrast to the situation for monolayer networks, where community-size distributions tend to have a relatively simple form, there often is not a simple closed-form expression for the single-layer community-size distributions in temporal networks. Consequently, for the Yang et al.\ \cite{Yang11} and Bazzi et al.\ \cite{Bazzi20} approaches, we used numerical simulations to examine the community-size distributions instead of obtaining analytical results (such as Theorem \ref{ArashTheorem}) about the single-layer community-size distributions in the limit of infinitely many layers. 
 Deriving {additional} analytical results {that are similar to Theorem \ref{ArashTheorem}} can provide important insights into the behavior of community-assignment approaches, and such insights can then facilitate  the development of better-performing and more efficient methods.}


\appendix

\section{Expressions for the Community-Assignment Probabilties} \label{Expressions}
In this appendix, we derive closed-form expressions for the community-assignment probability distributions $\P(g)$ for the Bazzi et al.\ approach \cite{Bazzi20} (see Section \ref{NodewiseEvolML}) and our LECS-prior-based approach (see Section \ref{novelML}). We use these expressions in the Gibbs-sampling procedures that we described in Section \ref{sec:algo}. 


\subsection{Closed-Form Expression for $\P(g)$ for the Bazzi et al.\ Approach}\label{BazziCF}

Recall from Section \ref{NodewiseEvolML} that the Bazzi et al.\ approach samples community assignments $g$ via the discrete-time Markov process (\ref{MarkovProcessModelNoPrior}). This Markov process is
\begin{alignat}{3}
    \pi & &&\;\sim\; \text{Dir}(\gamma) \, , \notag \\
    g_{(i,1)} &\given \pi &&\;\sim\; \pi \,, \notag\\
    \{g_{(i,\ell)}\}_{\ell = 2}^L & \given \alpha, K &&\;\sim\; \text{Markov}\left( \left\{\alpha_{\ell} I + (1 - \alpha_\ell) K^{(\ell)}\right\} \right) \, . \notag
 \end{alignat}
     where $\gamma = (1,\ldots,1)$ and $\alpha = (\alpha_2,\ldots,\alpha_L)$.
In the Bazzi et al.\ approach,
\begin{align*}
        \alpha_\ell &\sim \text{Unif}(0,1) \, , \\
        \kappa^{(\ell)} &\sim \text{Dir}(\mu^{(\ell)}) \, , \\
        K^{(\ell)}_{s*} &= \kappa^{(\ell)}\, ,
\end{align*}
where we also assume that $\mu^{(\ell)} = (1,\ldots, 1)$ for each layer $\ell \in \{2,\ldots,L\}$.

Because we generate $\{g_{(i,\ell)}\}_{\ell=1}^L$ via a discrete-time Markov process, the community assignment of a node in a given layer depends only on its community assignment in the previous layer. Therefore,
\begin{equation}\label{dtMarkov}
	\P(g) = \P(g_{(1)}) \prod_{\ell=2}^L \P(g_{(\ell)}|g_{(\ell-1)}) \, 
\end{equation}
and
\begin{equation}\label{dtMarkov2}
	\P(g_{(\ell)}|g_{(\ell-1)},K,\alpha) = \prod_{i=1}^n \mathbb{P}\left(g_{(i,\ell)} \Big|g_{(i,\ell - 1)},K, \alpha \right) \,. 
\end{equation}
To derive a closed-form expression for $\P(g)$, it suffices to derive closed-form expressions for $ \P(g_{(1)})$ and $\P(g_{(\ell)}|g_{(\ell - 1)})$. 

Because the procedure 
\begin{alignat}{3}
    \pi & &&\;\sim\; \text{Dir}(\gamma) \, , \notag \\
    g_{(i,1)} &| \pi &&\;\sim\; \pi \,, \notag
\end{alignat}
for sampling $g_{(1)}$ is identical to the nodewise community-assignment approach for monolayer networks that we discussed in Section \ref{nodewiseSL}, it is equivalent to sample $g_{(1)}$ using a uniform distribution on community sizes. 
Namely, we first choose the sizes $n_1,\ldots,n_k$ of communities $1,\ldots,k$ uniformly at random from the set $\{(n_1,\ldots,n_k) : \sum_{i = 1}^k n_i = n\} \cap \{0,\ldots,n\}^k$ of ordered pairs of $k$ non-negative integer elements that sum to $n$, and we then choose $g_{(1)}$ uniformly at random from the set of all community assignments with $n_i$ nodes in community $i$ for all $i \in [k]$. Therefore,
\begin{equation}\label{layer1Prob}
	\P(g_{(1)}) = \frac{1}{\binom{k + n - 1}{n}} \frac{n_1^{((1))}! \times \cdots \times n_k^{((1))}!}{n!} \, ,
\end{equation}
where we recall that $n_r^{((1))}$ is the number of times that $r$ appears in $g_{(1)}$.

To derive a closed-form expression for $\P(g_{(\ell)}|g_{(\ell-1)})$, we recall from (\ref{condsinglenode}) that
\begin{equation*}
	\P( g_{(i,\ell)} = r | g_{(i,\ell - 1)} = s, \alpha, K) = \alpha_\ell \, \mathbb{1}\{r = s\} + (1 - \alpha_\ell) K^{(\ell)}_{sr} \, .
\end{equation*}	
It thus follows directly from (\ref{dtMarkov2}) that
\begin{align*}
	\P(g_{(\ell)}|g_{(\ell - 1)},K,\alpha) &= \prod_{i = 1}^n \mathbb{P}\left(g_{(i,\ell)} \Big|g_{(i,\ell - 1)},K, \alpha \right) \\
		&=  \prod_{i = 1}^n \left( \alpha_\ell \, \mathbb{1}\{g_{(i,\ell)} = g_{(i,\ell-1)}\} + (1 - \alpha_\ell) K^{(\ell)}_{g_{(i,\ell - 1)}g_{(i,\ell)}} \right)\, .
\end{align*}
Finally, because
\begin{align*}
        \alpha_\ell &\sim \text{Unif}(0,1) \, , \\
        \kappa^{(\ell)} &\sim \text{Dir}(\mu^{(\ell)}) \, ,  \\
        K^{(\ell)}_{s*} &= \kappa^{(\ell)}\, ,
\end{align*}
with $\mu^{(\ell)} = (1,\ldots, 1)$ for each $\ell \in \{2,\ldots,L\}$, we have
\begin{align}\label{marginalNodewise}
	\P(g_{(\ell)}|g_{(\ell - 1)})
&=  \dotsint_{[0,1] \times \Delta^{k-1}} \prod_{i=1}^n \left( \alpha_\ell \, \mathbb{1}\{g_{(i,\ell)} = g_{(i,\ell-1)}\} + (1-\alpha_\ell) \kappa^{(\ell)}_{g_{(i,\ell)}} \right) \; d\mu(\alpha_\ell,\kappa^{(\ell)})\,,
\end{align}
where we recall from \eqref{DeltaDefinition} that $\Delta^{k-1}$ denotes the probability simplex in 
$\mathbb R^k$ and where $\mu$ is the product measure of a uniform measure on $[0,1]$ and a uniform measure on $\Delta^{k-1}$. (The measure is uniform because $\mu^{(\ell)} = (1,\ldots, 1)$.)
Combining equation (\ref{marginalNodewise}) with (\ref{dtMarkov}) and (\ref{layer1Prob}) yields a closed-form expression for $\P(g)$. 


\subsection{Closed-Form Expression for $\P(g)$ for our LECS-Prior-Based Approach} \label{novelCF}
Recall from Section \ref{novelML} that our LECS-prior-based approach samples community assignments $g$ via the process in (\ref{novelFirstLayer}) and (\ref{eq:lecs-count-split}). For each $r \in [k]$, we sample $\mathbf{g}'_{r,\ell}$ according to the following procedure:
\begin{alignat}{3}
    \pi & &&\sim \text{Dir}(\gamma) \, , \notag \\
    g_{(i,1)} &\; | \; \pi &&\sim \pi \, , \notag\\
p_{r,\ell} & &&\sim \text{Unif}(0,1) \, , \notag\\
	c_{rr}^{(\ell)} &\; | \; p_{r,\ell} &&\sim \geomb\!\left(n_r^{(\ell-1)},\,p_{r,\ell}\right) \, ,\notag\\
	\big(c_{rs}^{(\ell)}\big)_{s\neq r} & \; \Big| \;  c_{rr}^{(\ell)} &&\sim\ \text{Unif}\!\left(\mathcal{C}^{k-1}_{\,n_r^{(\ell-1)}-c_{rr}^{(\ell)}}\right)\, , \notag
\end{alignat}
where we recall that $\gamma = (1,\ldots,1)$. We then set $g_{(\ell)} = \bigoplus_{r=1}^k \mathbf{g}'_{r,\ell}$. 
The above procedure implies that the community assignments in a given layer depend only on the community assignments in the previous layer. We thus obtain
\begin{equation}\label{novel1}
	\P(g) = \P(g_{(1)}) \prod_{\ell=2}^L \P(g_{(\ell)}|g_{(\ell - 1)}) \, ,
\end{equation}
which takes the same form as (\ref{dtMarkov}).
Additionally, because the procedure to sample $g_{(1)}$ and the choice of $\gamma$ is the same as in the Bazzi et al.\ \cite{Bazzi20} approach, we obtain (using the same logic as in Appendix \ref{BazziCF}) the expression
   \begin{equation}\label{layer1ProbNovel}
	\P(g_{(1)}) = \frac{1}{\binom{k+n-1}{n}} \frac{n_1^{((1))}! \times \cdots \times n_k^{((1))}!}{n!} \, ,
\end{equation}
where we recall that $n_r^{((1))}$ is the number of times that $r$ appears in $g_{(1)}$.

To derive a closed-form expression for $\P(g_{(\ell)}|g_{(\ell - 1)})$, we first note that we obtain the set $g_{(\ell)}$ of community assignments for layer $\ell$ by concatenating $\mathbf{g}'_{r, \ell}$, which is the set of community assignments for layer $\ell$ when restricted to nodes in community $r$ in layer $\ell - 1$ for each $r \in [k]$. 
Therefore, 
\begin{equation}\label{novelprevlayer}
	\P(g_{(\ell)}|g_{(\ell-1)}) = \prod_{r=1}^k \P(\mathbf{g}'_{r, \ell}) \,.
\end{equation}

Recall that we generate $\mathbf{g}'_{r, \ell}$ using the procedure
\begin{alignat}{3}
 c^{(\ell)}_{rr} & \;|\; p_{r,\ell} &&\;\sim\; \geomb(n_r^{((\ell-1))},p_{r,\ell}) \, , \notag \\
    \mathbf{c}^{(\ell)}_{r,-r} & \;|\; c^{(\ell)}_{rr} &&\;\sim\; \text{Unif}\left(\mathcal{C}^{k-1}_{n_r^{((\ell-1))} - c^{(\ell)}_{rr}}\right) \, ,\notag \\
    \mathbf{g}'_{r,\ell} & \;|\; \mathbf{c}^{(\ell)}_r &&\;\sim\; \text{Unif}\left(\mathcal{G}_r^{(\ell-1)}(\mathbf{c}^{(\ell)}_r)\right) \, . \notag
      \end{alignat}
By the definition (see Section~\ref{novelML}) of $\geomb(n_r^{((\ell-1))},p_{r,\ell})$, we have
\begin{equation}\label{probNovelStep1}
	\P(c^{(\ell)}_{rr} | p_{r,\ell}) = p_{r,\ell}^{n_r^{((\ell-1))}  - c^{(\ell)}_{rr}}\frac{p_{r,\ell}-1}{p_{r,\ell}^{n_r^{((\ell-1))}+1} - 1} \, .
\end{equation}
Because $\mathbf{c}^{(\ell)}_{r,-r} | c^{(\ell)}_{rr}\sim\text{Unif}\left(\mathcal{C}^{k-1}_{n_r^{((\ell-1))} - c^{(\ell)}_{rr}}\right)$, a combinatorial argument combined with the definition of weak compositions implies that
\begin{equation}\label{probNovelStep2}
	\P(\mathbf{c}^{(\ell)}_{r,-r} | c^{(\ell)}_{rr}) = \frac{1}{\binom{n_r^{((\ell-1))} - c^{(\ell)}_{rr} + k - 2}{n_r^{((\ell-1))} - c^{(\ell)}_{rr}} } \, .
\end{equation}
Recall that we sample $\mathbf{g}'_{r,\ell} | \mathbf{c}^{(\ell)}_r \sim\text{Unif}(\mathcal{G}(\mathbf{c}^{(\ell)}_r))$. It follows from a combinatorial argument along with the definition of $\mathcal{G}$ (see Section \ref{novelML}) that
\begin{equation}\label{probNovelStep3}
	\P(\mathbf{g}'_{r,\ell} | \mathbf{c}^{(\ell)}_r) = \frac{1}{\binom{n_r^{((\ell-1))}}{c_{r,1},\ldots,c_{r,k}}} \, ,
\end{equation}
where $\binom{n}{k_1,\ldots,k_j} = \frac{n!}{k_1!k_2!\ldots k_j!}$ and $\sum_{i=1}^j k_i = n$. 
Combining (\ref{probNovelStep1}), (\ref{probNovelStep2}), and (\ref{probNovelStep3}) yields
\begin{align*}
	\P(\mathbf{g}'_{r, \ell} | p_{r,\ell}) &= \P(c^{(\ell)}_{rr} | p_{r,\ell})\P(\mathbf{c}^{(\ell)}_{r,-r} | c^{(\ell)}_{rr})\P(\mathbf{g}'_{r,\ell} | \mathbf{c}^{(\ell)}_r)  \notag \\
 		&= p_{r,\ell}^{n_r^{((\ell - 1))}  - c^{(\ell)}_{rr}}\frac{p_{r,\ell} - 1}{p_{r,\ell}^{n_r^{((\ell-1))} + 1} - 1} \times \frac{1}{\binom{n_r^{((\ell-1))} - c^{(\ell)}_{rr} + k - 2}{n_r^{((\ell-1))} - c^{(\ell)}_{rr}} \binom{n_r^{((\ell - 1))}}{c_{r,1},\ldots,c_{r,k}}} \, .
\end{align*}		
Because $p_{r,\ell}  \sim \text{Unif}(0,1)$, integrating $\P(\mathbf{g}'_{r, \ell} | p_{r,\ell})$ with respect to the probability measure induced by $\P(p_{r,\ell})$ yields
\begin{equation}\label{novelGPrime}
	\P(\mathbf{g}'_{r, \ell}) = \frac{1}{\binom{n_r^{((\ell - 1))} - c^{(\ell)}_{rr} + k - 2}{n_r^{((\ell - 1))} - c^{(\ell)}_{rr}} \binom{n_r^{((\ell - 1))}}{c_{r,1},\ldots,c_{r,k}}} \times \int_0^1 p_{r,\ell}^{n_r^{((\ell - 1))}  - c^{(\ell)}_{rr}}\frac{p_{r,\ell} - 1}{p_{r,\ell}^{n_r^{((\ell - 1))} + 1} - 1} \; d p_{r,\ell} \, .
\end{equation}

Although we can directly compute $\P(\mathbf{g}'_{r, \ell})$ using equation (\ref{novelGPrime}), we make the computations more efficient by precomputing the integrals 
\begin{equation*} \int_0^1 p_{r,\ell}^{n_r^{((\ell-1))}  - c^{(\ell)}_{rr}}\frac{p_{r,\ell}-1}{p_{r,\ell}^{n_r^{((\ell-1))}+1} - 1} \; d p_{r,\ell} \, .\end{equation*}
For notational simplicity, we write 
\begin{equation}\label{IDefinition}
	J(k_1,k_2) = \int_0^1 x^{k_1} \frac{x-1}{x^{k_2+1} - 1} \; dx \,.
\end{equation}
Combining (\ref{novelprevlayer}), (\ref{novelGPrime}), and \eqref{IDefinition} then yields
\begin{equation}\label{novelFinal}
\P(g_{(\ell)}|g_{(\ell-1)}) = \prod_{r = 1}^k \left[\frac{1}{\binom{n_r^{((\ell-1))} - c^{(\ell)}_{rr} + k - 2}{n_r^{((\ell-1))} - c^{(\ell)}_{rr}} \binom{n_r^{((\ell-1))}}{c_{r,1},\ldots,c_{r,k}}} \times J\left(n_r^{((\ell-1))}  - c^{(\ell)}_{rr}, n_r^{((\ell-1))}\right)\right] \, .
\end{equation}
In concert, the relations \eqref{novelFinal}, (\ref{novel1}), and (\ref{layer1ProbNovel}) yield a closed-form expression for $\P(g)$. 
 
 
\subsection{Computing $J(k_1,k_2)$} 
\label{Ik1k2Computation} 
To compute \eqref{novelFinal}, we need to calculate $J(k_1,k_2)$ (see \eqref{IDefinition}). Because $n_r^{((\ell - 1))}  - c^{(\ell)}_{rr} \le n_r^{(\ell - 1)}$, we are able to compute \eqref{novelFinal} by calculating $J(k_1,k_2)$ only when $0 \le k_1 \le k_2$. In this region,
we use a partial-fractions expansion
to obtain
\begin{align}\label{PFInt}
	J(k_1,k_2) &= \int_0^1 x^{k_1} \frac{x - 1}{x^{k_2 + 1} - 1} \; dx \notag\\
	&= \sum_{r=1}^{k_2} c_r \left[\log\left(1 - e^{\frac{2\pi \mathrm{i} r}{k_2 + 1}}\right) - \log\left(-e^{\frac{2\pi \mathrm{i} r}{k_2 + 1}}\right)\right] \,,
\end{align}
where
\begin{equation}\label{PFIntConst}
	c_r = \frac{e^{\frac{2\pi \mathrm{i} r k_1}{k_2 + 1}}}{\prod_{s = 1, \, s \ne r}^{k_2} \left(e^{\frac{2\pi \mathrm{i} r}{k_2 + 1}} - e^{\frac{2\pi \mathrm{i} s}{k_2 + 1}}\right)} \, ,
\end{equation}
the symbol $\mathrm{i}$ denotes the imaginary unit, and $\log(\cdot)$ denotes the complex base-$e$ logarithm with a branch cut along the negative real axis.
To calculate $J(k_2,k_2)$, we note that
\begin{equation*}
	\sum_{k_1 = 0}^{k_2} \left(\int_0^1 x^{k_2} \frac{x-1}{x^{k_2 + 1} - 1} \; dx\right) =  \int_0^1 \left(\sum_{k_1 = 0}^{k_2} x^{k_2} \frac{x-1}{x^{k_2 + 1} - 1}\right) \; dx = \int_0^1 1 \; dx = 1 \, .
\end{equation*}	
We thereby obtain
\begin{align}\label{PFIntEqual}
	J(k_2,k_2) &= \int_0^1 x^{k_2} \frac{x-1}{x^{k_2 + 1} - 1} \; dx \notag \\
	&= 1 - \sum_{k_1 = 0}^{k_2 - 1} \int_0^1 x^{k_2} \frac{x-1}{x^{k_2 + 1} - 1} \; dx \notag \\
	&= 1 - \sum_{k_1 = 0}^{k_2 - 1} J(k_1,k_2) \, .
\end{align}	

We now prove Lemma \ref{JSimplerForm}, which provides a simpler expression 
 for $J(k_1,k_2)$. However, in our code, we use the expressions \eqref{PFInt} and \eqref{PFIntEqual} to compute $J$.

{
\begin{lemma}\label{JSimplerForm}
	For any $k_1$ and $k_2$ in $\mathbb Z_{\ge 0}$, we have
	\begin{equation}\label{JSimplerExp}
		J(k_1,k_2) := \int_0^1 x^{k_1} \frac{x - 1}{x^{k_2 + 1} - 1} \; dx = \frac{1}{k_2 + 1} \left[ \psi\left(\frac{k_1 + 2}{k_2 + 1}\right) - \psi\left(\frac{k_1 + 1}{k_2 + 1}\right) \right]  \, ,
	\end{equation}
where $\psi(z) = \frac{d}{dz} \! \log \Gamma(z)$ is the digamma function~\cite{dlmf}.
\end{lemma}

\begin{proof}
	Let $N = k_2 + 1$. For $x \in [0, 1)$, we use the geometric-series expansion $\frac{1}{1 - u} = \sum_{n = 0}^\infty u^n$ to obtain
	\begin{equation}\label{this-two}
		\frac{x - 1}{x^N - 1} = (1 - x) \sum_{n = 0}^\infty (x^N)^n = \sum_{n=0}^\infty (x^{nN} - x^{nN + 1}) \, .
	\end{equation}
	Substituting \eqref{this-two} into \eqref{PFInt}	
	and switching the order of integration and summation (which is valid due to uniform convergence) yields
	\begin{equation*}
		J(k_1, k_2) = \sum_{n = 0}^\infty \int_0^1 (x^{k_1 + nN} - x^{k_1 + 1 + nN}) \, dx \, .
	\end{equation*}
	The integral of $x^a$ from $0$ to $1$ is $1/(a + 1)$, so
	\begin{equation}\label{series}
		J(k_1, k_2) = \sum_{n=0}^\infty \left( \frac{1}{k_1 + nN + 1} - \frac{1}{k_1 + nN + 2} \right) \, .
	\end{equation}
One can express the series in \eqref{series} in terms of the digamma function $\psi$,
which has the series representation
	\begin{equation*}
		\psi(z) = -\gamma - \sum_{n = 0}^\infty \left( \frac{1}{n + z} - \frac{1}{n + 1} \right). 
	\end{equation*}
	In particular, 
	\begin{equation}\label{digammaidentity}
		\sum_{n = 0}^\infty \left( \frac{1}{n + a} - \frac{1}{n + b} \right) = \psi(b) - \psi(a) \, .
	\end{equation}
	Rescaling the sum in \eqref{series} by $1/N$ and using the identity \eqref{digammaidentity} yields
	\begin{equation*}
		J(k_1, k_2) = \sum_{n = 0}^\infty \frac{1}{N} \left( \frac{1}{n + \frac{k_1 + 1}{N}} - \frac{1}{n + \frac{k_1 + 2}{N}} \right)  = \frac{1}{N} \left[ \psi\left(\frac{k_1 + 2}{N}\right) - \psi\left(\frac{k_1 + 1}{N}\right) \right]  \, ,
	\end{equation*}	
	which is the desired result.
\end{proof}
}


\section{Proof of Proposition~\ref{prop:uniform-compositions}}\label{app:remaining:proofs}
{In this appendix, we prove Proposition \ref{prop:uniform-compositions}, which establishes the equivalence of the following two sampling techniques:
\begin{enumerate}
 \item{Sample the sizes $n_1,\ldots,n_k$ of communities $1,\ldots,k$ uniformly at random from the set $\mathcal{C}_n^k$ of ordered pairs of $k$ non-negative integers that sum to $n$.}
 \item{Sample the community assignment $g$ uniformly at random from the set of all community assignments with $n_r$ nodes in community $r$ for all $r \in [k]$.} 
\end{enumerate}
For convenience, we restate the proposition.

\begin{proposition}
	Let $\pi \sim \Dir(1,\ldots,1)$ and, conditional on $\pi$, let $g_1,\ldots,g_n\sim\pi$.
	Let $n_r = \sum_{i = 1}^n\mathbf{1}\{g_i = r\}$ and $\bm n = (n_1,\ldots,n_k)$.
	We then have that $\bm n$ is \emph{uniform on the set $\mathcal{C}_n^k$ of weak $k$-compositions} of $n$. That is,
	\begin{equation*}
		\mathbb{P}(\bm n = c) = \frac{1}{\binom{n + k - 1}{k - 1}} \quad\text{for all } \, c\in\mathcal{C}_n^k \,.
	\end{equation*}
	Moreover, conditional on the counts $c = (c_1,\ldots,c_k)$, the community labels are \emph{uniform on community assignments with those counts}. That is, 
	for any community assignment $g$ with count vector $c$, we have $\mathbb{P}(g\,\mid\,\bm n = c) = \binom{n}{c_1,\ldots,c_k}^{-1}$ and $\mathbb{P}(g\,\mid\,\bm n = c) = 0$ otherwise.
	\end{proposition}

\begin{proof}
	Because $n_r = \sum_{i = 1}^n\mathbf{1}\{g_i = r\}$ with $g_1,\ldots,g_n\sim\pi$, we have $\bm n\mid\pi\sim\mathrm{Multinomial}(n,\pi)$. 
	Since $\pi \sim \Dir(1,\ldots,1)$, marginalizing according to this distribution, we find that $\bm n$ is Dirichlet--multinomial with PMF
	\begin{equation}\label{this-three}
		\mathbb{P}(\bm n = c) = \frac{n!}{\prod_r c_r!}\,\frac{\Gamma(\alpha_0)}{\Gamma(n + \alpha_0)}\, \prod_{r = 1}^k\frac{\Gamma(c_r + \alpha_r)}{\Gamma(\alpha_r)} \,.
	\end{equation}
	
With $\alpha_r\equiv1$ (and hence $\alpha_0 = k$), the expression \eqref{this-three} simplifies to
	$\mathbb{P}(\bm n = c) = n!(k - 1)!/(n + k - 1)!$, which is constant with respect to $c\in\mathcal{C}_n^k$ and sums to $1$ because $|\mathcal{C}_n^k| = \binom{n + k - 1}{k - 1}$.
	Then, to prove the result about $\mathbb{P}(g\,\mid\,\bm n = c)$, we note that given $\pi$, we know that any labeling $g$ with count vector
	$c$ has probability
	$\prod_{r = 1}^k \pi_r^{c_r}$. This probability is the same for any $g$ with count vector $c$. 
	Therefore,
	\begin{equation*}
		\mathbb{P}(g\,\mid\,\bm n = c,\pi) = \frac{1}{\binom{n}{c_1,\ldots,c_k}} \,, 
	\end{equation*}
	which is independent of $\pi$.
	Removing the conditioning on $\pi$ thus yields the same value, and we can write
	$\mathbb{P}(g\,\mid\,\bm n = c) = \binom{n}{c_1,\ldots,c_k}^{-1}$.
	\end{proof}


\section{Monolayer-Network Community-Size Localization: Theory} \label{sec:monolayer-theory}
In this appendix, we explain the empirical observations in Figure~\ref{singlelayercomp} and Table~\ref{iprSL} by proving various theoretical results. 

We start by proving a result about the distribution of community sizes when one samples
community assignments using
a uniform distribution on community assignments.
Let $(n_1,\ldots,n_k)$, with $\sum_r n_r = n$, denote the vector of community sizes.

\begin{proposition}[Uniform distribution on community assignments]\label{UnifDistProp}
For a uniform distribution on community assignments (equivalently, for independent and identically distributed (IID)
labels with probabilities $p_r = 1/k$), we have $n_1\sim\mathrm{Bin}(n,1/k)$. Therefore,
\begin{equation*}
	\sqrt{n}\left(\frac{n_1}{n} - \frac{1}{k}\right) \ \rightsquigarrow \ \mathcal N \left(0,\frac{k - 1}{k^2}\right)
\end{equation*}
and
\begin{equation*}
	\mathrm{Var} \left(\frac{n_1}{n}\right) = \frac{k - 1}{k^2} \times \frac{1}{n} \,,
\end{equation*}
where $\rightsquigarrow$ denotes convergence in distribution.
\end{proposition}
The first result in Proposition \ref{UnifDistProp} is an application of the central limit theorem, and the second result follows directly from the first result.

In Theorem \ref{theorem-fin}, we prove a deep connection between the hierarchical Bayesian model in Section~\ref{nodewiseSL} and the limiting behavior of the community-size distribution of {such a model}. In this theorem, we sample a distribution $\Pi$ from an arbitrary distribution, rather than sampling specifically from a Dirichlet distribution as in Section~\ref{nodewiseSL}.
	
	\begin{theorem}[Mixture LLN / finite-$k$ de Finetti]\label{theorem-fin}
	Let $Q$ be any probability distribution on the simplex $\Delta^{k - 1}$ and sample
	\begin{equation*}
		\Pi \sim Q\,,\quad g_i\mid \Pi = \pi \sim \pi \,,\ i \in [n] \,.
	\end{equation*}
	Let $\widehat p_n = (n_1/n,\ldots,n_k/n)$, where $n_r = \frac{1}{n} \sum_{i = 1}^n \bm 1 \{g_i = r\}$.
	Conditional on $\Pi = \pi$, we then have that $\widehat p_n\to \pi$ almost surely.
	Unconditionally, $\widehat p_n$ converges in distribution to $Q$. That is,
	\begin{equation*}
		\widehat p_n\ \rightsquigarrow\ Q\quad\text{on }\Delta_{k - 1} \,.
	\end{equation*}
	Equivalently, for any bounded continuous function $f : \Delta_{k - 1} \to \mathbb{R}$, we have
	\begin{equation}
		\mathbb{E}[f(\widehat p_n)]\to \mathbb{E}[f(\Pi)] \,.
	\end{equation}
	\end{theorem}

Theorem \ref{theorem-fin} guarantees
that the equivalence that we described in Proposition~\ref{prop:uniform-compositions} holds in an asymptotic sense for a general distribution $Q$ on the simplex $\Delta^{k - 1}$.
    One can view this result as a mixture form of the
    law of large numbers (LLN).
    Conditionally on $\Pi = \pi$, the labels $g_1,\ldots,g_n$ are IID with distribution
    $\pi$, so the 
    strong LLN
    guarantees that
    $\widehat p_n \to \pi$ almost surely. The key
    point is what happens \emph{unconditionally}:
    the limit
    $\pi$ is 
    random under the hierarchical model $\Pi\sim Q$, so the empirical proportion vector
    $\widehat p_n$ converges in distribution to the mixing distribution
    $Q$. In other words, the prior $Q$ is encoded in the large-$n$
    behavior of an observable statistic.

    The above viewpoint is closely related to the finite-alphabet case of de Finetti's theorem~{\cite[Theorem 1.1]{Kallenberg05}}.
    De Finetti's theorem asserts that an infinite exchangeable sequence with
    values in the alphabet $\{1,\ldots,k\}$ has
    a hierarchical
    representation as an IID sampling from a random probability vector $\Pi\in\Delta^{k - 1}$.
    Theorem \ref{theorem-fin} addresses the complementary direction. Assuming such a hierarchical representation,
    the empirical frequency vector $\widehat p_n$ recovers $\Pi$ in the limit. Additionally, across replicates, the limiting distribution of $\widehat p_n$
    recovers the mixing 
    distribution $Q$.

We view Theorem~\ref{theorem-fin} as conceptually deep because it makes the mixing distribution $Q$ observable in the limit. Although
$\Pi$ is latent for any individual sample, the distribution of the empirical proportions $\widehat p_n$ converges to
$Q$. 
Across many independent data sets that are generated from the
same hierarchical model, each large-$n$ replicate yields $\widehat p_n \approx \Pi$, so the empirical
distribution of $\widehat p_n$ across replicates directly informs $Q$.
This perspective aligns with the empirical-Bayes intuition 
that repeated related problems allow one to learn an ``objective" prior from
data~\cite{efron2024}.

We now prove Theorem \ref{theorem-fin}.

	\begin{proof}
	Conditional on $\Pi = \pi$, it follows from the strong LLN that $\widehat p_n\to \pi$ almost surely.
	 Fix a bounded continuous function $f$. Conditional on $\Pi = \pi$, the continuity of $f$ guarantees that $f(\widehat p_n) \to f(\pi)$ almost surely. By the bounded convergence theorem (applied to conditional expectation), we then have 
	 $h_n(\pi) := \mathbb{E} [f(\widehat p_n) \given \Pi = \pi] \to f(\pi)$. This, in turn, implies that $h_n(\Pi) \to f(\Pi)$ almost surely. Again applying the bounded convergence theorem and invoking the law of iterated expecation gives
	\begin{equation*}	
		\mathbb E [f(\widehat p_n)] = \mathbb E [h_n(\Pi)] \rightarrow \mathbb E[f(\Pi)] \,.
	\end{equation*}
	Because the function $f$ is general, we have established weak convergence of $\widehat p_n$ to $Q$. \qed
	\end{proof}

	\begin{corollary}\label{QCoordinate}
	If $Q = \Dir(1,\ldots,1)$, then $\widehat p_n\rightsquigarrow \Dir(1,\ldots,1)$ and each coordinate
	$\widehat p_{n,1}\rightsquigarrow \mathrm{Beta}(1,k - 1)$.
	For $k = 2$, this reduces to $\widehat p_{n,1}\rightsquigarrow \mathrm{Unif}(0,1)$.
	\end{corollary}

}

Combining Corollary \ref{QCoordinate} with Proposition~\ref{prop:uniform-compositions} yields the following corollary.

\begin{corollary}[Uniform distributions on
weak compositions]
\label{UnifDistLimitTheory}
For a uniform distribution on weak $k$-compositions of $n$, we have
\begin{equation*}
	\frac{n_1}{n}\ \rightsquigarrow\ \mathrm{Beta}(1,k - 1)
\end{equation*}
and
\begin{equation*}
	\mathrm{Var}\!\Big(\frac{n_1}{n}\Big)\to \frac{k - 1}{k^2(k + 1)} \,.
\end{equation*}
\end{corollary}

The following corollary summarizes the contrasts the localization properties of a uniform distribution on community assignments and a uniform distribution on weak compositions.

\begin{corollary}[Localization contrast via variance]
For uniform assignments, $\mathrm{Var}(n_1/n) = O(1/n)$, which entails a sharp concentration as one increases $n$. For uniform compositions, $\mathrm{Var}(n_1/n) = \Theta(1)$, which entails that the distribution yields a nonzero spread of community sizes even for $n \rightarrow \infty$.
 For $k = 2$, we have
 $\P(n_1 = i) = 1/(n + 1)$ (i.e., a discrete uniform distribution) for finite $n$, so $n_1/n\rightsquigarrow \mathrm{Unif}(0,1)$.
 \end{corollary}

One obtains a similar contrast in localization by calculating the inverse participation ratio (IPR) (see e.g. Table \ref{iprSL}).

\begin{proposition}[Localization contrast via IPR]\label{prop:ipr-contrast}
 Consider the IPR $\mathrm{IPR}_n = \sum_{i = 0}^n [\P(n_1 = i)]^2$ and fix the number $k$ of communities to be $k \ge 2$. 
 We then have the following results:
\begin{enumerate}
\item[(a)] (Uniform assignments) If $n_1\sim\mathrm{Bin}(n,1/k)$, then
\begin{equation*}
	\sqrt n\cdot \mathrm{IPR}_n \;\rightarrow\; \frac{k}{2\sqrt{\pi\,(k - 1)}} \,.
\end{equation*}
\item[(b)] (Uniform compositions) If $(n_1,\ldots,n_k)$ is uniform on weak $k$-compositions of $n$,
then, with $f(x) = (k - 1)(1 - x)^{k - 2}$ on $[0,1]$, we have
\begin{equation*}
	n\cdot \mathrm{IPR}_n \;\rightarrow\; \int_0^1 f(x)^2\,dx \;=\; \frac{(k - 1)^2}{2k - 3} \,.
\end{equation*}
\end{enumerate}
\end{proposition}

\begin{proof}[Proof of Proposition \ref{prop:ipr-contrast}]

\textbf{Part~(a).} Let $X$ and $X'$ be IID with $X \sim \mathrm{Bin}(n,p)$ and $p = 1/k$.
It then follows that
\begin{equation*}
	\mathrm{IPR}_n = \sum_{i = 0}^n [\Pr(X = i)]^2 = \Pr(X = X') = \Pr(S_n = 0) \,,
\end{equation*}
where $S_n = \sum_{j = 1}^n (B_j - B'_j)$ and $B_j,B'_j\stackrel{\text{IID}}{\sim}\mathrm{Bern}(p)$.
Each increment $Z := B_1 - B'_1$ takes values $\{-1,0,1\}$ and
\begin{equation*}
	\Pr(\pm1) = q := p(1 - p)\,,\quad \Pr(0) = 1 - 2q \,, \quad \mathbb{E}[Z] = 0\,,\quad \mathrm{Var}(Z) = 2q \,.
\end{equation*}
The support of $Z$ has span $1$, so one can apply the lattice local central limit theorem 
(see, e.g., {Durrett~\cite[Thm.\ 3.5.3]{DurrettPTE5} or Petrov~\cite[Ch.\ VII]{PetrovSums}}) and thereby obtain
\begin{equation*}
	\Pr(S_n = m) = \frac{1}{\sqrt{2\pi n\,\sigma^2}}\exp\!\Big(-\frac{m^2}{2\sigma^2 n}\Big) + o(n^{-1/2}) \quad\text{uniformly in }m\in\mathbb Z \,,
\end{equation*}
with $\sigma^2 = \mathrm{Var}(Z) = 2q$. For $m = 0$, this yields
\begin{equation*}
	\mathrm{IPR}_n = \Pr(S_n = 0) = \frac{1}{\sqrt{2\pi n\,\cdot\,2q}}\,(1 + o(1)) \,.
\end{equation*}
With $p = 1/k$, we have $q = (k - 1)/k^2$. Therefore,
\begin{equation*}
	\sqrt n\cdot \mathrm{IPR}_n \rightarrow \frac{k}{2\sqrt{\pi\,(k - 1)}} \,.
\end{equation*}

\textbf{Part~(b).}
	Our proof 
	has two key steps:
	we (1) establish a local limit theorem, and we then (2) use dominated convergence to justify the limit of a Riemann sum \eqref{RiemannSum} and evaluate the corresponding integral.

	Let $p_{n,i} = \P(n_1 = i)$. The exact PMF is
	\begin{equation*}
		p_{n,i} = \frac{\binom{n - i + k - 2}{k - 2}}{\binom{n + k - 1}{k - 1}} \,. 
	\end{equation*}
	The limiting density function $f(x)$ of the distribution of $n_1/n$ is the probability density function $(k - 1)(1 - x)^{k - 2}$ (where $x \in [0,1]$)
	of a $\mathrm{Beta}(1, k - 1)$ distribution.

	\paragraph{Step 1: Establishing a local limit theorem}
	
	We will show that
	\begin{equation*}
		\lim_{n\to\infty} n \cdot p_{n, \lfloor nx \rfloor} = f(x) 
	\end{equation*}
	for any fixed $x \in (0,1)$.
	For a fixed integer $m$ and large $N$, the binomial coefficient has the asymptotic behavior
	\begin{equation*}
		\binom{N}{m} \sim \frac{N^m}{m!} \sim \frac{(N - m)^m}{m!}
	\end{equation*}
	as $N \rightarrow \infty$.
	Let $i = \lfloor nx \rfloor$. As $n \to \infty$, we have $i \sim nx$ and $n - i \sim n(1 - x)$.
	The $n \rightarrow \infty$ asymptotic behavior of the denominator of $p_{n,i}$ is
	\begin{equation*}
		\binom{n + k - 1}{k - 1} \sim \frac{n^{k - 1}}{(k - 1)!} \,,
	\end{equation*}
	and the $n \rightarrow \infty$ asymptotic behavior of the numerator of $p_{n,i}$ is
	\begin{equation*} 
		\binom{n - i + k - 2}{k - 2} \sim \frac{(n - i)^{k - 2}}{(k - 2)!} \sim \frac{(n(1 - x))^{k - 2}}{(k - 2)!} \,. 
	\end{equation*}
	Taking the ratio of the numerator and the denominator then yields the asymptotic behavior of the PMF:
	\begin{equation} \label{ratio}
		p_{n,\lfloor nx \rfloor} \sim \frac{(k - 1)(1 - x)^{k - 2}}{n} = \frac{f(x)}{n} \,. 
	\end{equation}
	Multiplying \eqref{ratio} by $n$ then gives
	\begin{equation*}
		\lim_{n\to\infty} n \cdot p_{n, \lfloor nx \rfloor} = f(x) \,,
	\end{equation*}
	 which is the desired pointwise convergence.

	\paragraph{Step 2: Convergence of the sum \eqref{RiemannSum} to an integral}
	
	We want to evaluate the $n \rightarrow \infty$ limit of 	
	\begin{equation}\label{RiemannSum}
		n \cdot \mathrm{IPR}_n = n \sum_{i = 0}^n p_{n,i}^2 = \sum_{i = 0}^n (n \cdot p_{n,i})^2 \cdot \frac{1}{n} \,. 
	\end{equation}
	This expression resembles a Riemann sum of the integral of $f(x)^2$. 
	To formalize this observation, we apply the dominated convergence theorem (DCT). Consider the sequence $\left(g_n\right)_{n = 1}^\infty$ of step functions $g_n(x) = \left( n \cdot p_{n, \lfloor nx \rfloor} \right)^2$ on the interval $[0,1]$. With this sequence, we write
	\begin{equation*}
		\sum_{i = 0}^{n - 1} (n \cdot p_{n,i})^2 \cdot \frac{1}{n} = \sum_{i=0}^{n-1} \int_{i/n}^{(i+1)/n} (n \cdot p_{n,i})^2 \,dx = \int_0^1 g_n(x) \,dx \,,
	\end{equation*}
	which gives
	\begin{equation*}
		n \cdot \mathrm{IPR}_n  =  \int_0^1 g_n(x) \,dx + n p_{n,n}^2 \,.
	\end{equation*}
	By step 1, $g_n(x) \to f(x)^2$ $n \to \infty$ pointwise for $x \in (0,1)$. The step function $g_n(x)$ is dominated by a constant function, which is integrable on $[0,1]$. To see this, we note that $i \mapsto p_{n,i}$ is maximized at $i = 0$ and write
	\begin{equation*} 
		n \cdot p_{n,i} \le n \cdot p_{n,0} = n \cdot \frac{k - 1}{n + k - 1} < k - 1 \,. 
	\end{equation*}
	Therefore, for all $n \ge 1$ and all $x \in [0,1]$, the step function $g_n(x) = (n \cdot p_{n, \lfloor nx \rfloor})^2 \le (k - 1)^2$. By the DCT, we then have $\lim_{n\to\infty} \int_0^1 g_n(x) \,dx = \int_0^1 f(x)^2 \,dx$.
				
	The contribution from the endpoint $i = n$ is $n \cdot p_{n,n}^2 = n \cdot O(n^{-2(k - 1)}) = O(n^{-2k + 3})$, which vanishes as $n \to \infty$ for $k \ge 2$. We thus conclude that
	\begin{align}
		\lim_{n\to\infty} n \cdot \mathrm{IPR}_n &= \int_0^1 f(x)^2 \,dx \notag \\	 
		&= (k - 1)^2 \int_0^1 (1 - x)^{2k - 4} dx = \frac{(k - 1)^2}{2k - 3} \,. \label{appendixintegral}
	\end{align} 
	The condition $k \ge 2$ ensures that $2k - 3 \ge 1$, so the integral \eqref{appendixintegral} is well-defined.
\end{proof}


\section*{Acknowledgements}

We thank Aaron Clauset and Peter Mucha for helpful comments.


\bibliographystyle{siamplain} 
\bibliography{CommSizePaper-v22} 

\end{document}

%% file: ex_shared-v03.tex
\usepackage{cite}
\usepackage{algorithmic}
\usepackage{lipsum}
\usepackage{amsfonts}
\usepackage{graphicx}
\usepackage{mathbbol}
\usepackage{epstopdf}

\usepackage{geometry}                		
\usepackage{amssymb,amsmath,hyperref,esint,multirow,subfig,subfloat}

\renewcommand{\P}{\mathbb{P}}
\ifpdf
  \DeclareGraphicsExtensions{.eps,.pdf,.png,.jpg}
\else
  \DeclareGraphicsExtensions{.eps}
\fi

\newsiamremark{remark}{Remark}
\newsiamremark{hypothesis}{Hypothesis}
\crefname{hypothesis}{Hypothesis}{Hypotheses}
\newsiamthm{claim}{Claim}
\newsiamthm{altalgorithm}{Algorithm}

\headers{Community-Size Biases in Temporal Networks}{T. Y. Faust, A. A. Amini, and M. A. Porter}

\title{Community-Size Biases in Statistical Inference of Communities in Temporal Networks\thanks{Submitted to the editors DATE.
\funding{The first author was supported by the National Science Foundation Graduate Research Fellowship Program (under Grant No.\ DGE-2034835). }}}

\author{Theodore Y. Faust\thanks{Department of Mathematics, Southern Methodist University, Dallas, TX, USA
  (\email{tyfaust@smu.edu}).} \and Arash A. Amini\thanks{Department of Statistics and Data Science, University of California, Los Angeles, Los Angeles, CA, USA 
  (\email{aaamini@stat.ucla.edu}).}
\and Mason A. Porter\thanks{Department of Mathematics and Department of Sociology, University of California, Los Angeles, Los Angeles, CA, USA and Santa Fe Institute, Santa Fe, NM, USA
  (\email{mason@math.ucla.edu}).}}

\usepackage{amsopn}
